\documentclass[prb, aps, reprint, twocolumn, superscriptaddress,
notitlepage, nofootinbib, longbibliography,
floatfix]{revtex4-2}

\usepackage{natbib}
\usepackage{graphicx}
\usepackage{mathtools}
\usepackage{amsfonts} 
\usepackage{amsmath} 
\usepackage{amssymb}
\usepackage{amsthm}
\usepackage{bm}
\usepackage{braket}
\usepackage{color} 
\usepackage{dsfont}
\usepackage{hyperref}
\usepackage{subcaption}
\usepackage{caption}
\renewcommand{\raggedright}{\leftskip=0pt \rightskip=0pt plus 0cm}
\usepackage{import}
\usepackage{microtype}
\usepackage{relsize}
\usepackage{enumitem}
\usepackage{stackengine}
\usepackage{wasysym}
\usepackage[dvipsnames]{xcolor}
\graphicspath{{images/}}

\usepackage{tikz}
\usetikzlibrary{calc}
\usetikzlibrary{arrows.meta} 

\newcommand{\E}{\mathcal{E}}

\newcommand{\myketbra}[1]{\ketbra{#1}{#1}}

\usepackage{widebar}
\newcommand{\comp}[1]{\widebar{#1}}

\newcommand{\norm}[1]{\left\lVert#1\right\rVert}
\newcommand{\abs}[1]{|#1|}
\newcommand{\defeq}{\vcentcolon=}

\newcommand{\ketbra}[2]{\ket{#1}\!\bra{#2}}

\newcommand{\one}{\mathds{1}}
\newcommand{\R}{\mathbb{R}}
\newcommand{\Z}{\mathbb{Z}}
\newcommand{\N}{\mathbb{N}}

\newcommand{\Hilb}{\mathcal{H}}

\DeclareMathOperator{\Tr}{Tr}
\DeclareMathOperator{\rk}{Rank}
\DeclareMathOperator{\range}{rng}

\newtheorem{theorem}{Theorem}
\newtheorem{conjecture}{Conjecture}
\newtheorem{corollary}{Corollary}
\newtheorem{lemma}{Lemma}
\newtheorem{example}{Example}

\theoremstyle{definition}
\newtheorem{definition}{Definition}

\begin{document}

\title{Higher-form anomaly and long-range entanglement of mixed states}

\author{Leonardo A. Lessa}
\email{llessa@pitp.ca}
\affiliation{Perimeter Institute for Theoretical Physics, Waterloo, Ontario N2L 2Y5, Canada}
\affiliation{Department of Physics and Astronomy, University of Waterloo, Waterloo, Ontario N2L 3G1, Canada}

\author{Shengqi Sang}
\email{sangsq@stanford.edu}
\affiliation{Perimeter Institute for Theoretical Physics, Waterloo, Ontario N2L 2Y5, Canada}
\affiliation{Department of Physics and Astronomy, University of Waterloo, Waterloo, Ontario N2L 3G1, Canada}
\affiliation{Department of Physics, Stanford University, Stanford, CA 94305, USA}

\author{Tsung-Cheng Lu}
\email{tclu@umd.edu}
\affiliation{Perimeter Institute for Theoretical Physics, Waterloo, Ontario N2L 2Y5, Canada}
\affiliation{Joint Center for Quantum Information and Computer Science, University of Maryland, College Park, Maryland 20742, USA}

\author{Timothy H. Hsieh}
\email{thsieh@pitp.ca}
\affiliation{Perimeter Institute for Theoretical Physics, Waterloo, Ontario N2L 2Y5, Canada}

\author{Chong Wang}
\email{cwang4@pitp.ca}
\affiliation{Perimeter Institute for Theoretical Physics, Waterloo, Ontario N2L 2Y5, Canada}

\begin{abstract}
In open quantum systems, we directly relate anomalies of higher-form symmetries to the long-range entanglement of any mixed state with such symmetries.
First, we define equivalence classes of long-range entanglement in mixed states via stochastic local channels (SLCs), which effectively ``mod out'' any classical correlations and thus distinguish phases by  differences in long-range quantum correlations only. 
It is then shown that strong symmetries of a mixed state and their anomalies (non-trivial braiding and self-statistics) are intrinsic features of the entire phase of matter. For that, a general procedure of symmetry pullback for strong symmetries is introduced, whereby symmetries of the output state of an SLC are dressed into symmetries of the input state, with their anomaly relation preserved. This allows us to prove that states in (2+1)-D with anomalous strong 1-form symmetries exhibit long-range bipartite entanglement, and to establish a lower bound for their topological entanglement of formation, a mixed-state generalization of topological entanglement entropy. For concreteness, we apply this formalism to the toric code under Pauli-X and Z dephasing noise, as well as under ZX decoherence, which gives rise to the recently discovered intrinsically mixed-state topological order.
Finally, we conjecture a connection between higher-form anomalies and long-range multipartite entanglement for mixed states in higher dimensions.

\end{abstract}

\maketitle

\tableofcontents

\section{Introduction}

A remarkable feature of quantum many-body states is the ability to exhibit long-range entanglement (LRE), which cannot be destroyed by any local perturbations on their constituents \cite{zeng_quantum_2019}.
In the past decades, LRE has been prominently studied in states with topological order (TO) \cite{wen_topological_1990, chen_local_2010}, many of whom are realized as ground states of gapped Hamiltonians. Crucially, some of their defining features, such as the ground state degeneracy on a torus, are protected to not only local perturbations of the Hamiltonian, but also to sufficiently weak decoherence, arising from interactions with the environment. The robustness to the latter can be understood when TO is viewed as a quantum error correction code with distance scaling with the system size, as any encoded logical information can be recovered after decoherence noise with strength $p$, up to a threshold $p_c$ \cite{dennis_topological_2002}. In addition to this view, a novel perspective on noisy topological states has been forming, which describes the different regimes of recoverability as distinct mixed-state phases of matter, and several information-theoretic measures that probe the phase transition have been proposed \cite{sang_stability_2024, fan_diagnostics_2024, bao_mixedstate_2023}. However, these methods rely on the comparison to the noiseless pure-state TO, and no complete characterization of TO intrinsic to a single mixed-state $\rho$ has been accomplished. This is particularly relevant in light of the recent discovery of intrinsically mixed-state TO (imTO) phases, with no pure-state counterpart \cite{wang_intrinsic_2025, ellison_classification_2025, sohal_noisy_2025}.

The primary goal of this paper is, then, to understand topologically ordered mixed states as part of long-range entanglement mixed-state phases of matter.
For that, we employ the higher-form symmetry formalism, which extends Landau's paradigm to symmetry operators supported on manifolds of codimension $p$, with $p=0$ corresponding to ordinary global symmetries \cite{gaiotto_generalized_2014, mcgreevy_generalized_2023}. For TO in $(2+1)$-D, the braiding and self-statistics of anyonic excitations become nontrivial commutation relations of the operators that transport them around, which, when disposed in a loop, form \emph{anomalous} one-form symmetries. For technical simplicity, we only consider abelian topological orders in this work, since their higher form symmetries are invertible and can be implemented as finite-depth local unitary circuits.

Since the anomaly content is contained in the operator algebra of the higher-form symmetries, it can be readily generalized to symmetric mixed states $\rho$, which can be strongly or weakly symmetric \cite{buca_note_2012, groot_symmetry_2022}. An ensemble $\rho = \sum_i p_i \ketbra{\psi_i}{\psi_i}$ of pure states $\ket{\psi_i}$ is strongly symmetric under a symmetry operator $U$ if  each pure state in its decomposition is symmetric with the same symmetry charge $\lambda \in U(1)$, which is equivalent to $U \rho = \lambda \rho$. If the restriction on equal charge is dropped, $U$ remains only a \emph{weak} symmetry of $\rho$, i.e. $U \rho U^\dagger = \rho$ . The distinction between strong and weak symmetries is absent in pure states, so their interplay can reveal phenomena found only in open quantum systems. This has been the subject of many recent works, including investigations on spontaneous symmetry breaking from strong to weak symmetry \cite{sala_spontaneous_2024, lessa_strongweak_2024, zhang_strongweak_2024}, and strong-weak mixed symmetry protected topological phases \cite{ma_average_2023, ma_topological_2023, ma_symmetry_2024} and anomalies \cite{lessa_mixedstate_2024, wang_anomaly_2024}.

In this work, we consider states that are strongly symmetric under (a subset of) anomalous symmetries, as opposite to only weakly symmetric. For one, even featureless states, such as the maximally mixed state $\one / \dim{\mathcal{H}}$, can be weakly symmetric under (anomalous) symmetries; but, most importantly, we show that strong symmetries have good inheritance properties under action of a large class of ``local'' quantum operations, called stochastic local channel (SLC), which generalizes finite-depth quantum circuits by allowing for geometrically local interactions with an environment described by an arbitrary classical probability distribution (See Def. \ref{def:SLC} and Eq. \eqref{eq:SLC_decomposition}). More precisely, if $\rho$ is transformed via an SLC $\E$, resulting in a noisy state $\E(\rho)$, then any strong symmetry $g$ of $\E(\rho)$ gives rise to a strong symmetry $\tilde{g}$ of $\rho$, 
up to a product-state ancilla addition to $\rho$ and a locality-preserving modification of $g$. We call $\tilde{g}$ the \emph{symmetry pullback} of $g$ under $\E$ and denote it by $\tilde{g} = \E^*(g)$.

We take the symmetry pullback mechanism as the guiding principle to argue that mixed-states with a strong anomalous symmetries feature nontrivial patterns of long-range entanglement. It naturally leads to a notion of phase equivalence of mixed states $\rho$ and $\sigma$ that share the same strong symmetries by requiring them to be two-way connected via SLCs. Similar to the finite-depth local channels and to evolutions under local Lindbladians, both of which have been widely used to define mixed state phases \cite{coser_classification_2019a, ma_average_2023, ma_topological_2023, sang2024mixed, sang_stability_2024, wang_intrinsic_2025, wang_anomaly_2024, lessa_mixedstate_2024, lessa_strongweak_2024} and are themselves generalizations of the quasiadiabatic evolution of gapped ground states to mixed states, SLCs cannot create long-range entanglement. However, SLCs can create long-range classical correlations, so  the trivial phase defined via SLCs 
consists of arbitrary mixtures of short-range entangled states, which has also been deemed ``trivial'' or \textit{separable} in another set of works \cite{hastings2011topological, chen2024separability, chenUnconventionalTopologicalMixedstate2024, wang_analog_2024}.

We apply these novel techniques to TO states by proving that the anomaly of their one-form symmetries is invariant under symmetry pullback, and that no strongly symmetric anomalous state is bipartite separable, i.e. is not of the form $\rho_{\text{2-sep}} = \sum_i p_i \rho^A_i \otimes \rho_i^B$, for states $\rho^A_i$ and $\rho^B_i$ in regions $A$ and $B = A^c$. Together, these facts imply that, for example, every state in the toric code phase is long-range bipartite entangled. Furthermore, under sufficiently strong noise, the toric-code state may undergo a phase transition, which can be explained by certain changes in the strong one-form symmetries. One possibility, as occurs with $X$ or $Z$ dephasing, is that one of the strong symmetries associated with the $e$ and $m$ anyons becomes weak and the LRE is destroyed. However, another possibility is to keep the fermionic symmetry associated to $f=em$ strong, realizing the fermionic imTO \cite{wang_intrinsic_2025}, which is still long-range bipartite entangled, due to the non-bosonic self-statistics.

As described above, a mixed state with strong anomalous 1-form symmetries must be long-range bipartite entangled. Can it be diagnosed by certain information-theoretic quantity? For pure states, the topological entanglement entropy (TEE) \cite{kitaev_topological_2006, levinDetectingTopologicalOrder2006} provides one such diagnostic. Unfortunately, for mixed states, TEE can receive contributions from long-range classical correlations, so even fully classical mixed states (e.g. an ensemble of Z-basis product states) may have a non-zero TEE. Drawing inspiration from the entanglement of formation $E_F$ \cite{bennett_mixedstate_1996}, a mixed-state entanglement measure defined as $E_F= \min_{\{p_i, \rho_i \}}\sum_{i} p_i S_A(\rho_i)$ by minimizing the averaged entanglement entropy $S_A$ over all decompositions of the mixed state $\rho$, we propose the topological entanglement of formation $\gamma_F$ to diagnose mixed-state long-range entanglement via the same procedure applied to TEE. In particular, built on the approach in Refs. \cite{kim_universal_2023, levin_physical_2024}, which prove a lower bound for the TEE of topologically ordered pure states, we prove a lower bound for the TEF of mixed states with strong 1-form anomaly. Namely, $\gamma_F \geq \frac{1}{2} \log n$, where $n$ is related to the number of nontransparent strong symmetry anyons. For anomalies between strong and weak symmetries, a similar lower bound is described, but for the TEE instead of the TEF. We also prove that TEF must decrease monotonically under onsite noise channels. This result, together with the lower bound on TEF, shows that the toric code under Pauli-Z or X noise has $\gamma_F = \log 2 $ throughout the entire long-range entanglement phase. Finally, using a particular decomposition of the noisy toric-code state studied in \cite{chen2024separability, wang_analog_2024}, we show that the TEF is zero outside the toric-code phase.

The rest of the paper is organized as follows: in section \ref{sec:TO_mixed-state_phases}, we introduce a new definition of mixed-state phase of matter based on patterns of long-range entanglement (Sec. \ref{sec:LRE_vs_LRC}), and prove that strong symmetries of one state can be extended to other states in the same phase via symmetry pullback (Sec. \ref{sec:symmetry_pullback}).  In section \ref{sec:anomaly_def}, we define anomalies of one-form symmetries in $(2+1)$-D by local operatorial relations which implement braiding (Sec. \ref{sec:anomaly-braiding}) and particle exchange (Sec. \ref{sec:anomaly-top_twist}). Having defined anomaly, we argue in section \ref{sec:long-range} that it implies long-range bipartite entanglement (Sec. \ref{sec:long-range_bip_ent}) and provides a lower bound for the topological entanglement of formation (Sec. \ref{sec:long-range_TEF}), which generalizes the topological entanglement entropy to mixed states. In Sec. \ref{sec:example_toric_code}, we apply the preceding techniques to the toric code under several types of noise: Pauli-$Z$ and $X$ (Sec. \ref{sec:pauli-x_z_dephasing}) and $ZX$ dephasing (Sec. \ref{sec:zx_dephasing}), and discuss how they form quantum or classical memories. Finally, in section \ref{sec:higher-form_higher-dim}, we conjecture a generalized correspondence between anomaly of higher-form symmetries in higher dimensions and long-range multipartite entanglement, and provide an explicit example of mutual anomaly between 0-form and 1-form symmetries in section \ref{sec:set}.  

\subsection{Relation to previous works}

The results and discussions in this paper complements and partially overlaps with a number of recent works on mixed-state topological phases of matter that were published in parallel with the preparation of the current work. To help the reader navigate this rapidly evolving field, we clarify the differences and similarities between the present manuscript and a selection of five other works: 
\begin{itemize}
    \item Wang, Wu, Wang \cite{wang_intrinsic_2025}. The authors study the decoherence-induced proliferation of fermionic anyons $f$ in some exactly solvable models, and how it leads to an intrinsic mixed-state topological order. Two of their remarkable features are the long-range bipartite entanglement, which we generalize to any system with strong fermionic one-form symmetries (See Theorem \ref{thm:LRBE}), and the nonzero topological entanglement negativity, which we derive from an alternative method in Appendix \ref{appendix:negativity}. We also discuss the $ZX$-dephased toric code in greater depth in Sec. \ref{sec:zx_dephasing}.
    
    \item Ellison and Cheng \cite{ellison_classification_2025}, and Sohal and Prem \cite{sohal_noisy_2025}. Both works aim to classify mixed-state phases of matter, and to especially describe the intrinsically mixed-state topological ordered states, such as the one proposed by \cite{wang_intrinsic_2025}. By employing generalized ``gauging'' and ``incoherent proliferation'' mechanisms, they argued for a classification by pre-modular topological anyon theories, with possibly degenerate braiding. Although our analysis agrees with the view above, we do not attempt a full classification of mixed-state topological order. Furthermore, the focus on the interplay of strong and weak symmetries is present here as well. In particular, \cite{ellison_classification_2025} employed the same symmetry pullback mechanism of Lemma \ref{lemma:symmetry_pullback} to transfer the strong symmetries of one state to another in the same phase.
    
    \item Li, Lee, Yoshida \cite{li_how_2024}. The authors prove the long-range entanglement of states with emergent anyons and fermions within the stabilizer formalism, along with lower bounds to geometrical measures of long-range entanglement. Similarly, we prove the long-range \emph{bipartite} entanglement of states with anyons and fermions, but from the more general formalism of one-form symmetries developed in Sec. \ref{sec:anomaly_def}.
    
    \item Wang, Song, Meng, Grover \cite{wang_analog_2024}. The authors define a generalization of the topological entanglement entropy for mixed states by using the convex roof construction \cite{uhlmann_entropy_1998, horodecki_quantum_2009}, and argue that it captures the transition out of a topological quantum memory as the strength of Pauli noise increases. Here, we define a very similar quantity, denoted by topological entanglement of formation (See Def. \ref{def:TEF}), and prove a lower bound based on the braiding of strong one-form symmetries, thus generalizing analogous results for the TEE to mixed-states \cite{kim_universal_2023, levin_physical_2024}. This complements the discussion on \cite{wang_analog_2024}, since, together with their results, we can more confidently assert the value of the TEF for the toric code under Pauli-$Z$ and $X$ dephasing for all values of the noise strength (See Fig. \ref{fig:phase_diagram}).
\end{itemize}

\section{Topologically ordered mixed-state phases}\label{sec:TO_mixed-state_phases}

\subsection{Long-range entanglement vs. correlation}\label{sec:LRE_vs_LRC}

To define long-range entanglement, we first need to define which operations create only short-range entanglement. For pure states, finite-time evolutions under local Hamiltonians and their discrete counterparts, \textit{i.e.} the (finite-depth) local unitary (LU) circuits, are the paradigmatic choices, as exponentially small or no correlation is built up, respectively. Accordingly, two pure states are in the same phase if one can be transformed into the other by a $U_{\rm LU}$. If we think of long-range entanglement as a quantum resource, the LUs are the free operations, which don't create any resource, and the phases of matter are the states with the same resource content.

For open quantum systems, a natural generalization of the LU is the local channel (LC) defined as follows~\cite{hastings2011topological}: 
\begin{definition}[Local channels (LC)]\label{def:LC}
A quantum channel $\E_{\text{LC}}$ is a (finite-depth) LC if it can be Stinespring dilated to a (finite-depth) local unitary $U_{\text{LU}}$ acting on the original system and an ancilla $a$ with the same geometry, which is initialized in a product state $\ket{\bf 0}_a$:  
\begin{equation}\label{eq:LC_def}
    \E_{\rm LC}(\sigma) = \Tr_a\left[U_{\rm LU} \left(\sigma \otimes \ket{\bf 0}\bra{\bf 0}_a\right) U_{\rm LU}^\dagger \right].
\end{equation}    
\end{definition}
In other words, an LC is defined by the composition of three operations. First, an ancillary Hilbert space is added to each site, collectively referred to as $a$, and initialized in a product state $\ket{\bf 0}$. Then, an LU is applied to the joint system. Finally, the ancilla $a$ is traced out.

With this construction, we can define an equivalence relation\cite{coser_classification_2019a} between mixed states, $\rho$ and $\sigma$, via two-way convertibility under LCs $\E_{\rm LC}$ and $\mathcal{F}_{\rm LC}$, i.e. $\E_{\rm LC}(\rho) = \sigma$ and $\mathcal{F}_{\rm LC}(\sigma) = \rho$. The equivalence classes are characterized by patterns of long-range correlation and have been used in the recent literature to define mixed-state phases of matter, with many successful uses that extend the pure state definition.

However, mixed states can have long-range correlations of classical or quantum nature. The definition above does not discriminate between the two, as LCs can create neither. This is exemplified by the (fully separable) classical ferromagnetic state $\frac{1}{2}[00\cdots 0] + \frac{1}{2}[11\cdots 1]$, where $[\psi]$ is shorthand for $\ketbra{\psi}{\psi}$. The state has no entanglement whatsoever, as it is a mixture of product states, yet belongs to a different phase than the product state $[00\cdots0]$ according to the LC definition since no LCs can create the long-range correlation $\lim_{|i-j|\rightarrow \infty}\langle Z_i Z_j \rangle - \langle Z_i \rangle \langle Z_j \rangle = 1$.

In order to define genuine long-range entangled mixed-state phases of matter beyond just correlations, we need to enlarge the set of free operations to allow free supply of long-range classical correlations. To this end, we define:
\begin{definition}[Stochastic local channels (SLC)]\label{def:SLC}
    A quantum channel $\mathcal{E_{\text{SLC}}}$ is a (finite-depth) SLC if it can be written as a convex combination of (finite-depth) local channels $\E_i$:\footnote{For infinite convex combinations of LCs, we further require uniformly bounded range $\sup_i \range(\E_{{\rm LC}, i}) < \infty$. For the definition of range of a channel, see Def. \ref{def:CP} in Appendix \ref{appendix:locality_SLC}.} 
    \begin{equation}
        \E_{\rm SLC} = \sum_i p_i \E_{{\rm LC}, i}.
    \end{equation}
\end{definition}

\begin{figure}[t]
    \centering
    \includegraphics[width=0.7\linewidth]{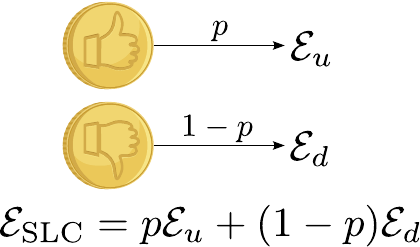}  
    \caption{Example of how a stochastic local channel (Def. \ref{def:SLC}) can be operationally realized by flipping a (biased) coin and having access to local channels $\E_u$ and $\E_d$. If the coin lands upwards, with probability $p$, the experimenter applies the LC $\E_u$, otherwise they apply $\E_d$, resulting in $\E_{\text{SLC}} = p \E_u + (1-p) \E_d$.
    }
    \label{fig:SLC_coin_toss}
\end{figure}

Operationally, an SLC can be realized by first sampling from the classical distribution $p$ and applying the LC $\E_i$ if the result $i$ is chosen. See Fig. \ref{fig:SLC_coin_toss} for a simple example with a coin toss. 

Since all quantum operations, \textit{i.e.} unitary gates and ancilla additions, are still local, SLCs cannot generate long-range entanglement. 
On the other hand, we observe that by acting SLCs on a product state, one can obtain any short-range entangled (SRE) mixed-states~\cite{hastings2011topological, chen2024separability}, \textit{i.e.} states of the form:
\begin{equation}\label{eq:mixed_SRE_states}
    \rho= \sum_i p_i \ketbra{\rm SRE_i}{\rm SRE_i},
\end{equation}
where each $\ket{\rm SRE_i}=U_{{\rm LU}, i}\ket{\bf 0}$ is a short-range entangled pure state. In particular, it contains any classical density matrix in the 0-1 basis. We thus see that SLCs generate any classical correlation. Borrowing nomenclature from the quantum resource theory literature~\cite{chitambar_quantum_2019}, allowing for any classical mixtures of LCs convexifies the long-range entanglement resource theory, as the set of SLCs is convex. This implies that the set of free states -- the ones that do not possess any resource, i.e. long-range entanglement -- is also convex, and thus exactly corresponds to the SRE states in Eq. \eqref{eq:mixed_SRE_states}.

Another way of realizing SLCs is to consider an LU acting on an extended space, similarly to an LC of Eq. \ref{eq:LC_def}, but starting with classically correlated ancillae. In this way, we introduce classical randomness as a resource not on the choice of channel but on the initial state. 
More explicitly, an SLC channel $\E_{\rm SLC} = \sum_{i=1}^n p_i \E_{{\rm LC}, i}$ always admits the following decomposition:
\begin{equation}\label{eq:SLC_decomposition}
    \E_{\rm SLC}[\sigma] = \Tr_a\left[U \left(\sigma \otimes \rho^{C}_a\right) U^\dagger \right]
\end{equation}
where $\rho^C_a$ is a classical distribution (i.e. being diagonal in the 0-1 basis) on the ancilla $a$, and $U$ is an LU circuit determined by $\{\E_{\rm LC, i}\}$. See Appendix \ref{appendix:SLC_decomp_derivation} for the explicit construction and the proof of the claim. 

From this perspective, in the same way that the finite-depth local channels (unitaries) generalize on-site local operations (unitaries) by allowing quantum gates to interact with their close neighbors, the stochastic local channels generalize the set of local operations and shared randomness (LOSR)~\cite{buscemi_all_2012,schmid_understanding_2023}, which is an alternative to the more widely used local operations and classical communications (LOCC) as a set of non-entangling operations\footnote{The geometrically local counterpart of LOCC, on the other hand, are the local adaptive circuits, which further allow for local measurements and feedforward via non-local classical communication. They are known to be much more powerful than LCs (and SLCs), and can create long-range entanglement in short depth~\cite{piroli_quantum_2021, friedman_locality_2022, lee_measurementprepared_2022, lu_measurement_2022, tantivasadakarn_hierarchy_2023, foss-feig_experimental_2023, tantivasadakarn_longrange_2024, iqbal_nonabelian_2024}
.}. 

Naturally, SLCs lead to a definition of mixed-state phases of matter based on long-range entanglement:
\begin{definition}[Mixed-state long-range entanglement phase of matter]\label{def:phase_of_matter}
    Two mixed-states, $\rho$ and $\sigma$, are in the same \textit{long-range entanglement phase of matter} if there exists two SLCs, $\E_{\rm SLC}$ and $\mathcal{F}_{\rm SLC}$, such that
    \begin{equation}
        \sigma = \E_{\rm SLC} (\rho) \quad \text{and} \quad \rho = \mathcal{F}_{\rm SLC} (\sigma).
    \end{equation}
\end{definition}
From now on, unless explicitly stated otherwise, we refer to these LRE phases of matter as just ``phases of matter'' for brevity.

For our purposes here, this phase equivalence via SCLs is more adequate than the one via LCs because it still differentiates topologically ordered states from trivial states, as will be argued below, while being indifferent to classical correlations. Further properties of SLCs, such as how the pure state phases are (un)changed, and how to characterize the locality properties of SLCs and other channels, are discussed in Appendix \ref{appendix:further_props_SLC}.

\subsection{Symmetry pullback}
\label{sec:symmetry_pullback}

Symmetries play central roles in understanding pure state quantum phases. Here we show that this remains to be the case when studying mixed-state entanglement phases.

For mixed-states, symmetry comes with two types: strong symmetry and weak symmetry~\cite{buca_note_2012, groot_symmetry_2022}. For a unitary $U$, we say $U$ is a strong (weak) symmetry of the state $\rho$, if $U\rho = \lambda \rho$, for $|\lambda| = 1$ ($U\rho U^\dagger=\rho$). If we view the state as an ensemble of pure states, e.g. $\rho=\sum_i p_i\myketbra{\psi_i}$, then $U$ being a strong symmetry means each individual $\ket{\psi_i}$ in the ensemble is symmetric under $U$ with equal charge $\lambda$, irrespective of how we decompose $\rho$; while $U$ being a weak symmetry means even though the total ensemble is unchanged under $U$, certain $\ket{\psi_i}$  might change non-trivially under $U$, or have different charges with respect to other states in the ensemble.

In this work we focus primarily on strong symmetries. For a mixed-state $\rho$, we define its strong symmetry group $G_\rho$ as follows:
\begin{equation}
    G_\rho = \{g \in U(\Hilb) \mid \exists \lambda \in U(1),\ g\rho = \lambda \rho\},
\end{equation}
Where we call the $g$-dependent $\lambda \in U(1)$ the charge of the state $\rho$ under the symmetry operator $g \in G_\rho$. We remark that, even though we will later specialize to subgroups of $G_\rho$, namely $k$-form symmetries, we keep full generality of discussion in this section. 

An important property of strong symmetry is that it is inherited by all states in the ensemble:
\begin{lemma}[Inheritance of strong symmetry]\label{lemma:decomposition_inheritance}
    If $\rho = \sum_i p_i \rho_i$, $p_i > 0$, is strongly symmetric under $g \in G_\rho$, i.e. $g \rho = \lambda \rho$, with charge $\lambda \in U(1)$, then each state $\rho_i$ in its ensemble \textit{inherits} the strong symmetry with the same charge: $\forall i, g \rho_i = \lambda \rho_i$.
\end{lemma}
\begin{proof}
    The strong symmetry condition $g \rho = \lambda \rho$ is equivalent to $\forall \ket{v} \in \mathrm{Im}(\rho),\ g \ket{v} = \lambda \ket{v}$. 
    Thus, it suffices to prove that $\mathrm{Im}(\rho_i) \subseteq \mathrm{Im}(\rho)$ for any $\rho_i$ in the ensemble, or equivalently, that $\mathrm{ker}(\rho) \subseteq \mathrm{ker}(\rho_i)$. Indeed, for any $\ket{w} \in \mathrm{ker}(\rho)$, we have $0 = \braket{w|\rho|w} = \sum_i p_i \braket{w|\rho_i|w},$
    which, from the positivity of $p_i>0$ and $\rho_i \geq 0$, implies $\ket{w} \in \mathrm{ker}(\rho_i)$.
\end{proof}
A useful corollary is that extensions of $\rho$ also inherit strong symmetries:
\begin{corollary}\label{cor:extension_inheritance}
    If $\tilde{\rho} \in Q(\Hilb \otimes \Hilb')$ is an extension of $\rho \in Q(\Hilb)$, i.e. $\Tr_{\Hilb'}[\tilde{\rho}] = \rho$, and $\rho$ is strongly symmetric under $g$, then $\tilde{\rho}$ is strongly symmetric under $g \otimes \one_{\Hilb'}$.
\end{corollary}
\begin{proof}
    Let $\tilde{\rho} = \sum_i p_i \ketbra{\tilde{\psi}_i}{\tilde{\psi}_i}$ be a decomposition of $\tilde{\rho}$ and $\ket{\tilde{\psi}_i} = \sum_j \sqrt{q_{i,j}} \ket{\psi_{i,j}}\ket{j}_{\Hilb'}$ be the Schmidt decompositions of each of its pure states. Since $\rho = \sum_{i,j} p_i q_{i, j} \ketbra{\psi_{i,j}}{\psi_{i,j}}$, each $\ket{\psi_{i,j}}$ inherits the strong symmetry $g$ from $\rho$, by Lemma \ref{lemma:decomposition_inheritance}. Thus $\ket{\tilde{\psi}_i}$, and ultimately $\tilde{\rho}$, are strongly symmetric under $g \otimes \one_{\Hilb'}$.
\end{proof}
In the language of the strong symmetry group $G_\rho$, the preceding two results imply that
\begin{align}
    \forall \rho \in Q(\Hilb),\ \rho = \sum_i p_i \rho_i\ \Rightarrow\ & G_{\rho} \subseteq G_{\rho_i} \\
    \forall \rho \in Q(\Hilb),\ \rho = \Tr_{\Hilb'}[\tilde{\rho}]\ \Rightarrow\ & G_{\rho} \subseteq G_{\tilde{\rho}}.
\end{align}

With the inheritance property in mind, the following Lemma relates strong symmetries of the output state to that of the input state through an SLC\footnote{
We note that the symmetry pullback mechanism was also proposed and used in \cite{ellison_classification_2025} for a similar purpose, although not with this name.
}: 

\begin{lemma}[Symmetry pullback]\label{lemma:symmetry_pullback}
Suppose $\E=\sum_i p_i \E_i$ is an SLC. Then there exists $U$ being an LU and some ancillary qubits $A$ added locally on the lattice, such that for any state $\rho$, the map
\begin{equation}
g \mapsto \E_i^*(g) \defeq U^\dagger(g\otimes\mathbb{I}_A)U
\end{equation}
is an injective homomorphism from $G_{\E(\rho)}$ to $G_{\rho\otimes\ket{\bf 0}\bra{\bf 0}_A}$ that preserves the strong charge. Moreover, $U$ can be taken as the dilated LU of any $\E_i$ in the ensemble.
\end{lemma}
\begin{proof}
We call $\rho$'s Hilbert space $\Hilb_Q$. Consider any $\E_i$'s Stinespring dilation form:
\begin{equation}
    \E_i[\rho] = \Tr_A[U(\rho\otimes\ketbra{\bf 0}{\bf 0}_A)U^\dagger]
\end{equation}
where $U$ is an LU on $\Hilb_Q\otimes\Hilb_A$. By the inheritance property of strong symmetry (Lemma \ref{lemma:decomposition_inheritance}), $g\in G_{\E(\rho)}$ implies $g \in G_{\E_i(\rho)}$, which by Corollary \ref{cor:extension_inheritance} further implies
\begin{equation}
    g\otimes \mathbb{I}_A \in G_{U(\rho\otimes\ketbra{\bf 0}{\bf 0}_A)U^\dagger}
\end{equation}
which is equivalent to:
\begin{equation}
    U^\dagger(g\otimes\mathbb{I}_A)U \in G_{\rho\otimes\ketbra{\bf 0}{\bf 0}_A}.
\end{equation}
\end{proof}
Thus, all of the strong symmetries of the output state $\E(\rho)$ must have a counterpart for the input state $\rho$ plus a trivial ancillary state $\ket{\mathbf 0}$
\footnote{We remark that the trivial ancilla state $\ketbra{\bf 0}{\bf 0}_A$ is in general necessary for the symmetry pullback to be valid. Otherwise, the maximally mixed state $\rho_{\rm MMS} \propto \one$ would inherit all strong symmetries from pure product states in the same phase, contradicting the fact that it has no strong symmetries: $G_{\rho} = \{\mathbb{I}\}$.}. 
Moreover, the mapping is locality-preserving, due to $\E$ being an SLC. We name this mechanism \emph{symmetry pullback}, and denote the resulting symmetry operator by $\E_i^*(g)$ or simply $\E^*(g)$, if the LC $\E_i$ is implicit. 

\begin{figure}[t]
    \centering
    \includegraphics[width=\linewidth]{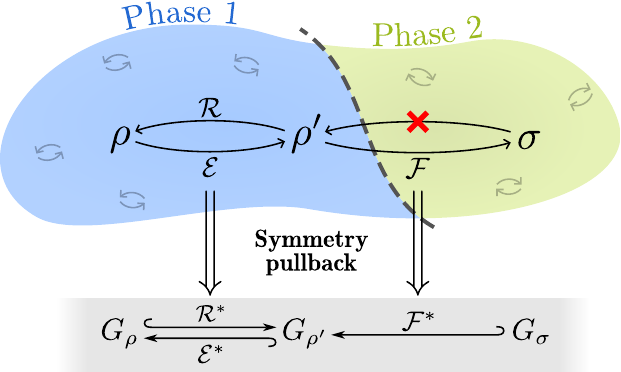}
    \caption{Illustration of the use of symmetry pullback to characterize the strong symmetries within and between phases. Suppose a state $\rho$ in phase 1 is acted upon by an ``error'' channel $\E$. If the resulting state $\rho' \defeq \E(\rho)$ is in the same phase, then there exists a recovery SLC $\mathcal{R}$, such that $\mathcal{R}(\rho') = \rho$. With these two channels, the strong symmetry group $G_\rho$ of $\rho$ can be pulled back to $G_{\rho'}$, and vice-versa. States in different phases are not two-way connected. Hence, if there is a channel $\mathcal{F}$ that takes $\rho'$ out of phase 1 to $\sigma = \mathcal{F} (\rho')$ in phase 2, it cannot be reversed. In this case, the symmetry pullback can only guarantee the inclusion $G_{\sigma} \xhookrightarrow{\mathcal{F}^{*}} G_{\rho'}$, and it may happen that $\mathcal{F}^*(G_\sigma)$ is strictly smaller than $G_{\rho'}$.
    }
    \label{fig:symmetry_pullback}
\end{figure}

Importantly, under Definition~\ref{def:phase_of_matter} of phases of matter based on two-way connectivity by SLCs, the strong symmetries of a state are pulled back to strong symmetries of all other states in the same phase, up to ancilla addition and local unitary conjugation (See Fig. \ref{fig:symmetry_pullback}). 
In this sense, strong symmetries are universal properties of the phase. This is the organizing principle by which the long-range entanglement of anomalous states will be argued in the following sections.

We end by noting that all the preceding results (Lemmas \ref{lemma:decomposition_inheritance} and \ref{lemma:symmetry_pullback}, and Corollary \ref{cor:extension_inheritance}) are not valid in general for weakly symmetric states. For example, the maximally mixed state $\rho_{\rm MMS} \propto \one$ is weakly symmetric with respect to any unitary operator, which may not be inherited by states in its decomposition (which are all states) or by symmetry pullback. However, in a future work \cite{sang_mixedstate_2025}, it will be shown that the weak symmetry still survives under pullback $\E^*$ in a generalized sense for a refined class of channels $\E$ (See also Sec. \ref{sec:classical_memory}).
.

\section{Anomaly of abelian 1-form symmetries in (2+1)-D} \label{sec:anomaly_def}

Having discussed the general properties of strong and weak symmetries, we now restrict our attention to abelian 1-form symmetries in (2+1)-D by first defining them and their anomaly in this section. Here, we do not attempt to provide a complete treatment of 1-form symmetries or topological order in 2d, and how our definitions compare to alternative approaches. Instead, we present a minimal framework that enables us to prove long-range entanglement results coming from anomalous 1-form symmetries, and is applicable to the prototypical cases of interest, such as the toric code state under different types of noise.

To define 1-form symmetries microscopically, we assume our system is defined on a regular 2D lattice (e.g. square or honeycomb lattice).
A sequence of links, each ending where the next starts, is a path $\gamma$. If the curve $\gamma$ has no endpoints, then it is a loop.

We define a string operator $W$ to be an assignment of a finite-depth local unitary $W(\gamma)$ to every curve $\gamma$ such that $W(\gamma)$ is supported around $\gamma$ and that
\begin{align}
    W(\gamma \circ \tau) & = W(\gamma) W(\tau) \\
    W(\emptyset) & = \one,
\end{align}
where $\gamma \circ \tau$ is the juxtaposition of the curve $\tau$ followed by $\gamma$, and $\emptyset$ is the empty curve. In particular, the relations above imply that $W(\bar\gamma) = W(\gamma)^{-1}$, where $\bar\gamma$ is the curve $\gamma$ with link directions reversed. 

A string operator $W$ is a 1-form (weak or strong) symmetry of $\rho$ if $\rho$ is (weakly or strongly) symmetric under $W(\ell)$ for all contractible loops $\ell$. For the strong symmetry case, the charge $\lambda_\ell$ of $W(\ell)$ in $W(\ell) \rho = \lambda_\ell \rho$ can be arbitrary, but we will usually take it to be trivial $\lambda_\ell = 1$. This will not incur any loss of generality in our forthcoming results, since they rely only on the change of eigenvalue $\lambda_\ell$ after the application of a charged operator on the state $\rho$.   

Importantly, when $\rho$ has a 1-form symmetry $W$, the action of a open string operator $W(\gamma)$ in $\rho$ is invariant under deformations of $\gamma$ that maintain the endpoints fixed. For example, suppose $\rho$ is strongly symmetric under $W$ and $\gamma'$ is a deformation of $\gamma$. Then, by definition, $\bar\gamma \circ \gamma'$ is a contractible loop, and so
\begin{equation}
    W(\gamma') \rho = W(\gamma) W(\bar\gamma \circ \gamma')\rho = W(\gamma) \rho
\end{equation}
The same is true for a weakly symmetric state if we replace left multiplication by conjugation.

For one-form symmetries pertaining to topologically ordered systems, we should interpret that the string operators transport anyonic quasiparticles along its path.
Moreover, the deformability property above is crucial to interpret the action of such string operators on symmetric states as creating \emph{localized} excitations at its endpoints, since their in-between paths can continuously change. Thus, to each anyon $a$ we assign a string operator $W_a(\gamma)$ that creates $a$ quasiparticles localized at the end of $\gamma$ and $\bar{a}$ quasiparticles at the start of $\gamma$. In particular, we have that $W_1(\gamma) = \one$, where $1$ is the trivial anyon, and $W_a(\gamma) = W_{\bar{a}}(\gamma)^\dagger$.  For non-Abelian anyons, such string operator has necessarily long-depth if it acts unitarily \cite{bravyi_adaptive_2022}. 
Because of this, we restrict our attention to Abelian topological order and FDLU string operators.

In the following two sections, we define the two basic anyonic transport operations, braiding and twisting, using the microscopical formalism of 1-form symmetries described above. When string operators braid or twist nontrivially, we say they are \emph{anomalous}, which will be shown in Sec. \ref{sec:long-range} to imply long-range entanglement for their strongly symmetric states.

\captionsetup{justification=centering,singlelinecheck=false}

\begin{figure*}[t]
    \centering
    \begin{subfigure}[t]{0.43\textwidth}
        \centering
        {\def\svgwidth{\linewidth}
\begingroup%
  \makeatletter%
  \providecommand\color[2][]{%
    \errmessage{(Inkscape) Color is used for the text in Inkscape, but the package 'color.sty' is not loaded}%
    \renewcommand\color[2][]{}%
  }%
  \providecommand\transparent[1]{%
    \errmessage{(Inkscape) Transparency is used (non-zero) for the text in Inkscape, but the package 'transparent.sty' is not loaded}%
    \renewcommand\transparent[1]{}%
  }%
  \providecommand\rotatebox[2]{#2}%
  \newcommand*\fsize{\dimexpr\f@size pt\relax}%
  \newcommand*\lineheight[1]{\fontsize{\fsize}{#1\fsize}\selectfont}%
  \ifx\svgwidth\undefined%
    \setlength{\unitlength}{290.0828348bp}%
    \ifx\svgscale\undefined%
      \relax%
    \else%
      \setlength{\unitlength}{\unitlength * \real{\svgscale}}%
    \fi%
  \else%
    \setlength{\unitlength}{\svgwidth}%
  \fi%
  \global\let\svgwidth\undefined%
  \global\let\svgscale\undefined%
  \makeatother%
  \begin{picture}(1,0.41298125)%
    \lineheight{1}%
    \setlength\tabcolsep{0pt}%
    \put(0,0){\includegraphics[width=\unitlength,page=1]{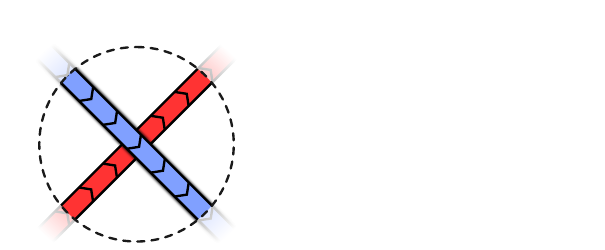}}%
    \put(0.28053867,0.28718792){\makebox(0,0)[t]{\lineheight{1.25}\smash{\begin{tabular}[t]{c}B\end{tabular}}}}%
    \put(0.31560026,0.33208019){\makebox(0,0)[t]{\lineheight{1.25}\smash{\begin{tabular}[t]{c}A\end{tabular}}}}%
    \put(0,0){\includegraphics[width=\unitlength,page=2]{braiding_def.pdf}}%
    \put(0.49973504,0.16645565){\makebox(0,0)[t]{\lineheight{1.25}\smash{\begin{tabular}[t]{c}{\Large $= S_{{\color{red}a}{\color{blue}b}}$}\end{tabular}}}}%
    \put(0.1739615,0.16799965){\makebox(0,0)[t]{\lineheight{1.25}\smash{\begin{tabular}[t]{c}$p$\end{tabular}}}}%
    \put(0.06143049,0.01013089){\makebox(0,0)[t]{\lineheight{1.25}\smash{\begin{tabular}[t]{c}$W_{\color{red}a}(\gamma)$\end{tabular}}}}%
    \put(0.06143049,0.3384911){\makebox(0,0)[t]{\lineheight{1.25}\smash{\begin{tabular}[t]{c}$W_{\color{blue}b}(\tau)$\end{tabular}}}}%
  \end{picture}%
\endgroup%
}
        
        \vspace{2em}
        
        {\def\svgwidth{\linewidth}
\begingroup%
  \makeatletter%
  \providecommand\color[2][]{%
    \errmessage{(Inkscape) Color is used for the text in Inkscape, but the package 'color.sty' is not loaded}%
    \renewcommand\color[2][]{}%
  }%
  \providecommand\transparent[1]{%
    \errmessage{(Inkscape) Transparency is used (non-zero) for the text in Inkscape, but the package 'transparent.sty' is not loaded}%
    \renewcommand\transparent[1]{}%
  }%
  \providecommand\rotatebox[2]{#2}%
  \newcommand*\fsize{\dimexpr\f@size pt\relax}%
  \newcommand*\lineheight[1]{\fontsize{\fsize}{#1\fsize}\selectfont}%
  \ifx\svgwidth\undefined%
    \setlength{\unitlength}{614.28570124bp}%
    \ifx\svgscale\undefined%
      \relax%
    \else%
      \setlength{\unitlength}{\unitlength * \real{\svgscale}}%
    \fi%
  \else%
    \setlength{\unitlength}{\svgwidth}%
  \fi%
  \global\let\svgwidth\undefined%
  \global\let\svgscale\undefined%
  \makeatother%
  \begin{picture}(1,0.30086789)%
    \lineheight{1}%
    \setlength\tabcolsep{0pt}%
    \put(0,0){\includegraphics[width=\unitlength,page=1]{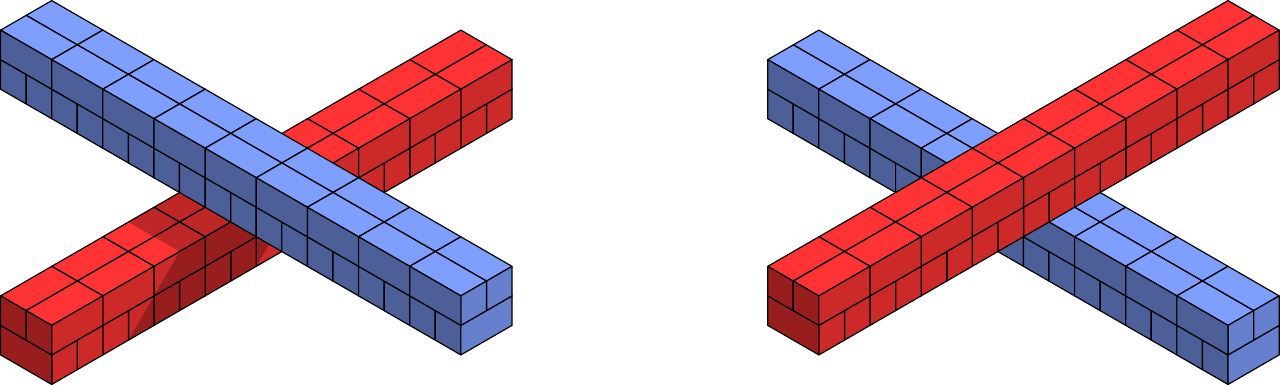}}%
    \put(0.50000003,0.15043397){\makebox(0,0)[t]{\lineheight{1.25}\smash{\begin{tabular}[t]{c}\Large $= S_{{\color{red}a}{\color{blue}b}}$\end{tabular}}}}%
  \end{picture}%
\endgroup%
}
        \caption{ }
        \label{fig:braiding_def}
    \end{subfigure}
    \hfill
    \begin{subfigure}[t]{0.43\textwidth}
        \centering
        {\def\svgwidth{\linewidth}
        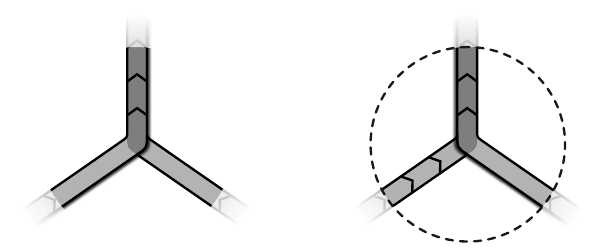}
        
        \vspace{2em}
        
        {\def\svgwidth{\linewidth}
\begingroup%
  \makeatletter%
  \providecommand\color[2][]{%
    \errmessage{(Inkscape) Color is used for the text in Inkscape, but the package 'color.sty' is not loaded}%
    \renewcommand\color[2][]{}%
  }%
  \providecommand\transparent[1]{%
    \errmessage{(Inkscape) Transparency is used (non-zero) for the text in Inkscape, but the package 'transparent.sty' is not loaded}%
    \renewcommand\transparent[1]{}%
  }%
  \providecommand\rotatebox[2]{#2}%
  \newcommand*\fsize{\dimexpr\f@size pt\relax}%
  \newcommand*\lineheight[1]{\fontsize{\fsize}{#1\fsize}\selectfont}%
  \ifx\svgwidth\undefined%
    \setlength{\unitlength}{614.28574449bp}%
    \ifx\svgscale\undefined%
      \relax%
    \else%
      \setlength{\unitlength}{\unitlength * \real{\svgscale}}%
    \fi%
  \else%
    \setlength{\unitlength}{\svgwidth}%
  \fi%
  \global\let\svgwidth\undefined%
  \global\let\svgscale\undefined%
  \makeatother%
  \begin{picture}(1,0.32394068)%
    \lineheight{1}%
    \setlength\tabcolsep{0pt}%
    \put(0,0){\includegraphics[width=\unitlength,page=1]{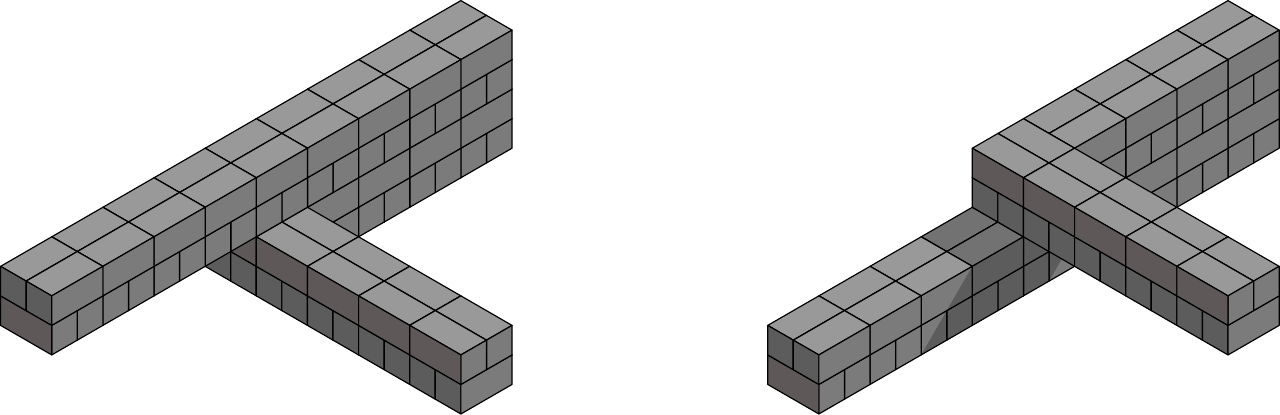}}%
    \put(0.49999996,0.12713377){\makebox(0,0)[t]{\lineheight{1.25}\smash{\begin{tabular}[t]{c}\Large $= \theta_a$\end{tabular}}}}%
  \end{picture}%
\endgroup%
}
        \caption{ }
        \label{fig:twist_def}
    \end{subfigure}
    \caption{Diagrammatic definitions of (a) braiding and (b) self-statistics along with their operator counterparts.
    }
    \label{fig:top_defs}
\end{figure*}

\captionsetup{justification=raggedright,singlelinecheck=false}

\subsection{Braiding}
\label{sec:anomaly-braiding}
The braiding of two string operators, $W_a(\gamma)$ and $W_b(\tau)$, that intersect at a point $p$ is the interchange of their order of operation only around point $p$. When the exchange amounts to a phase factor, $S_{ab} \neq 1$, we say it is anomalous with scattering phase $S_{ab} \in U(1)$. 

To define the local operator order change more precisely, we require $\gamma$ and $\tau$ to intersect at only one point $p$ within a radius large compared to the depth of $W_a(\gamma)$ and $W_b(\tau)$. Then, $W_b(\tau) W_a(\gamma)$ will look like an X crossing in a large neighborhood $B$ around $p$, as shown in Fig. \ref{fig:braiding_def}. Since both $W_b(\tau)$ and $W_a(\gamma)$ are quantum circuits, we can \emph{locally} change the order of the gates of both operators in the product $W_b(\tau)W_a(\gamma)$. Specifically, we move all the gates of $W_a(\gamma)$ strictly inside region $B$ to act after all gates of $W_b(\tau)$ (See bottom of Fig. \ref{fig:braiding_def}), maintaining the order between the layers of each string operator. This is a well-defined operation given a FDLU representation because the gates near the boundary of $B$ pertain to either $W_a(\gamma)$ or $W_b(\tau)$. This means that the action of the two string operators is altered solely near the intersection point $p$ after braiding, regardless of their paths far from $p$.

Furthermore, this procedure coincides with the twist product between $W_a(\gamma)$ and $W_b(\tau)$, defined as~\cite{haah_invariant_2016}
\begin{equation}\label{eq:twist_product}
    W_b(\tau) \infty W_a(\gamma) \defeq \sum_{i,j} A^{(b)}_j A^{(a)}_i \otimes B^{(a)}_i B^{(b)}_j,
\end{equation}
where we have decomposed $W_a(\gamma) = \sum_i A^{(a)}_i \otimes B^{(a)}_i$, and $W_b(\tau) = \sum_j A^{(b)}_j \otimes B^{(b)}_j$ into sums of tensor product operators in $A$ and $B$ regions. Indeed, the twist product also inverts the order of operation of the product $W_{b}(\tau) W_a(\gamma)$ only in the region $B$, so $W_b(\tau) \infty W_a(\gamma) = S_{ab} W_b(\tau) W_a(\gamma)$. This equivalence shows that the braiding phase is independent of a particular FDLU representation of the string operators.

\subsection{Self-statistics}
\label{sec:anomaly-top_twist}
A single anyon type may acquire a nontrivial global phase when two exchange positions. This is the self-statistics, or topological twist phase, and is another significant piece of topological data. The braiding of an anyon $a$ with itself only probes the self-statistics up to a sign, as $S_{aa} = \theta_a^2$. In particular, it cannot discern if a anyon with no self-braiding, $S_{aa} = 1$, is a boson $\theta_a = 1$ or a fermion $\theta_a = -1$. However, we wish to identify fermionic quasiparticles at the end of string operators as sufficient condition for anomalous symmetry, because it still guarantees long-range entanglement (See Sec. \ref{sec:long-range_bip_ent}).

To that end, we detect the self-statistics phase $\theta_a$ via another exchange of operator order implementing a ``half-braiding''. More precisely, we consider two curves, $\gamma$ and $\tau$, that meet at a point $p$ inside a region $B$ and later take the exact same path from $p$ to a boundary point $s$, as shown in Fig. \ref{fig:twist_def}. To exchange the operators $W_a(\gamma)$ and $W_a(\tau)$ inside $B$, we proceed in the same way as the braiding exchange: each gate of $W_a(\tau)$ acting on $B$ is moved to act after $W_b(\gamma)$, preserving their internal order. Importantly, this is well-defined because the gates straddling the boundary of $B$ are either of only one of the string operators (near $q$ and $r$), or is the same in both (near $s$). Thus, again, the change is local to the point $p$, and does not modify the product $W_a(\gamma) W_a(\tau)$ outside it, even in the region where $\gamma$ and $\tau$ coincide. 

Unlike braiding, however, it is not clear if the twist product of $W_a(\gamma)$ with $W_a(\tau)$ will match the procedure described above, since the local operators $B_i^{(a)}$ of Eq.~\eqref{eq:twist_product} might not commute near the boundary point $s$. Despite this, we conjecture that $\theta_a$ depends only on the string operators themselves, instead of their FDLU representations.

To justify why the topological twist defined above matches the one from TQFT, we note that the diagram of Fig.~\ref{fig:twist_def} is equivalent to the widely used hopping equation~\cite{levin_fermions_2003, zeng_quantum_2019}
\begin{equation}\label{eq:T_diagram}
    t_{pr} t_{sp} t_{pq} = \theta_a t_{pq} t_{sp} t_{pr},
\end{equation}
where $t_{ji} = W_a(i \to j)$ is a hopping operator that transports the $a$ anyon from $i$ to $j$. When acted on a two-particle state at $q$ and $r$, both sides of the equations transport the anyons to $p$ and $s$ in two different ways that differ by a particle exchange, and hence the phase difference $\theta_a$. It is easy to see that our definition is the same, up to another hopping operator $t_{sp}$.

\section{Long-range entanglement}
\label{sec:long-range}

Having defined braiding and self-statistics based on operatorial commutation relations, we now argue that any mixed state strongly symmetric under anomalous one-form symmetries feature nontrivial patterns of long-range entanglement. These come in two forms. In Section~\ref{sec:long-range_bip_ent}, we prove that such anomalous state is necessarily long-range bipartite entangled. That is, it is not bipartite separable, nor connected to any bipartite separable state via SLCs. In Section~\ref{sec:long-range_TEF}, we define topological entanglement of formation, a mixed-state generalization of topological entanglement entropy that has a lower bound depending on the number of non-transparent anyons. A lower bound for anomalies of strong and weak symmetries is also proven.

\subsection{Long-range bipartite entanglement proof from anomaly}
\label{sec:long-range_bip_ent}

To prove long-range bipartite entanglement of strongly symmetric anomalous states, we divide the argument into two parts. First, we argue in Lemma \ref{lemma:anomaly_invariance} that the braiding and self-statistics phases defined above are invariant under symmetry pullback, meaning that they are topological invariants across the whole phase of matter. Then, in Lemma \ref{lemma:BE}, we prove that no bipartite separable mixed state is strongly symmetric under an anomalous one-form symmetry. By itself, this latter result is fragile to local perturbations on the boundary of the bipartition, which can entangle the state. However, with the former anomaly invariance statement, we further conclude \emph{long-range} bipartite entanglement in Theorem \ref{thm:LRBE}. This strategy showcases the practical utility of the symmetry pullback, as it provides a way to extend results that are valid for a single state to its entire phase of matter.

\begin{figure*}
    \def\svgwidth{0.8\linewidth}
\begingroup%
  \makeatletter%
  \providecommand\color[2][]{%
    \errmessage{(Inkscape) Color is used for the text in Inkscape, but the package 'color.sty' is not loaded}%
    \renewcommand\color[2][]{}%
  }%
  \providecommand\transparent[1]{%
    \errmessage{(Inkscape) Transparency is used (non-zero) for the text in Inkscape, but the package 'transparent.sty' is not loaded}%
    \renewcommand\transparent[1]{}%
  }%
  \providecommand\rotatebox[2]{#2}%
  \newcommand*\fsize{\dimexpr\f@size pt\relax}%
  \newcommand*\lineheight[1]{\fontsize{\fsize}{#1\fsize}\selectfont}%
  \ifx\svgwidth\undefined%
    \setlength{\unitlength}{913.67936262bp}%
    \ifx\svgscale\undefined%
      \relax%
    \else%
      \setlength{\unitlength}{\unitlength * \real{\svgscale}}%
    \fi%
  \else%
    \setlength{\unitlength}{\svgwidth}%
  \fi%
  \global\let\svgwidth\undefined%
  \global\let\svgscale\undefined%
  \makeatother%
  \begin{picture}(1,0.21444511)%
    \lineheight{1}%
    \setlength\tabcolsep{0pt}%
    \put(0,0){\includegraphics[width=\unitlength,page=1]{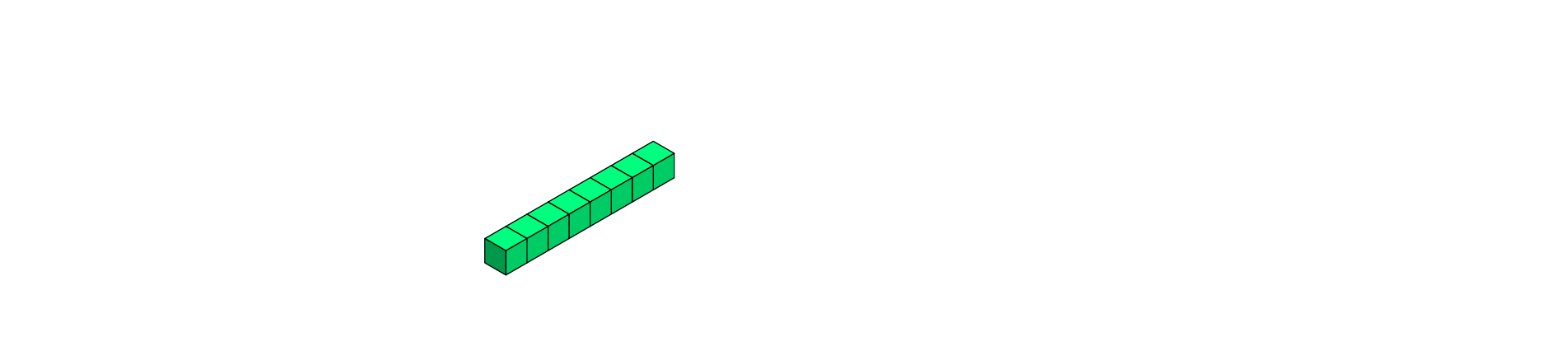}}%
    \put(0.2692631,0.06059455){\makebox(0,0)[t]{\lineheight{1.25}\smash{\begin{tabular}[t]{c}$\circ$\end{tabular}}}}%
    \put(0,0){\includegraphics[width=\unitlength,page=2]{pullback_3d.pdf}}%
    \put(0.47046314,0.06059455){\makebox(0,0)[t]{\lineheight{1.25}\smash{\begin{tabular}[t]{c}$\circ$\end{tabular}}}}%
    \put(0.10778226,0.1924487){\makebox(0,0)[t]{\lineheight{1.25}\smash{\begin{tabular}[t]{c}$U^\dagger$\end{tabular}}}}%
    \put(0.3626808,0.1924487){\makebox(0,0)[t]{\lineheight{1.25}\smash{\begin{tabular}[t]{c}$W(\gamma) \otimes \one$\end{tabular}}}}%
    \put(0.61827457,0.1924487){\makebox(0,0)[t]{\lineheight{1.25}\smash{\begin{tabular}[t]{c}$U$\end{tabular}}}}%
    \put(0.77948276,0.06059455){\makebox(0,0)[t]{\lineheight{1.25}\smash{\begin{tabular}[t]{c}$=$\end{tabular}}}}%
    \put(0.8735207,0.1924487){\makebox(0,0)[t]{\lineheight{1.25}\smash{\begin{tabular}[t]{c}$\E^*(W(\gamma))$\end{tabular}}}}%
  \end{picture}%
\endgroup%

    \caption{The pullback $\E^*(W(\gamma)) = U^\dagger (W(\gamma) \otimes \one) U$ of an on-site string operator $W(\gamma)$ under an SLC $\E$, where each unitary gate is represented by a block. The unitary $U$ is a FDLU that dilates the action of a local channel of $\E$ to an extended space. Thus, the support of $W(\gamma)$ is locally enlarged after conjugation by $U$, as gates far for it cancel to identity.}
    \label{fig:one-form_pullback}
\end{figure*}

\begin{lemma}\label{lemma:anomaly_invariance}
    The braiding matrix $S_{ab}$ and self-statistics phase $\theta_a$ as defined in Secs.~\ref{sec:anomaly-braiding} and \ref{sec:anomaly-top_twist} are invariant under symmetry pullback $W(\gamma) \mapsto \E^*(W(\gamma))$ of an SLC $\E$.
\end{lemma}
\begin{proof}
    Up to an ancilla, the symmetry pullback of a one-form symmetry $W(\gamma)$ dresses it with gates from the Stinespring dilated FDLU $U$ of the SLC $\E$ as $\E^*(W(\gamma)) = U^\dagger\ W(\gamma) \otimes \one\ U$ (See Fig. \ref{fig:one-form_pullback}). Importantly, gates of $U$ and $U^\dagger$ far from the support of $W(\gamma)$ cancel each other. In both the braiding and self-statistics equations, further cancellation occurs in the product of two string operators near their intersection point $p$ (See Fig. \ref{fig:top_defs}). Their order of operation can thus be interchanged locally, resulting in the same phase factors.
\end{proof}

\begin{lemma}\label{lemma:BE}
    A state $\rho$ with an anomalous strong 1-form symmetry is bipartite entangled, i.e. $\rho$ is not bipartite separable with respect to a sufficiently large bipartition $A|B$,
    \begin{equation}\label{eq:bipartite_entanglement}
        \rho \neq \rho_{\text{2-sep}} = \sum_{k} p_k \ketbra{A_k}{A_k} \otimes \ketbra{B_k}{B_k},
    \end{equation}
\end{lemma}
\begin{proof}
    By contradiction, let us assume a bipartite separable state $\rho_{\text{2-sep}} = \sum_{k} p_k \ketbra{A_k}{A_k} \otimes \ketbra{B_k}{B_k}$ is strongly symmetric under an anomalous 1-form symmetry. Then, by strong symmetry inheritance (Lemma \ref{lemma:decomposition_inheritance}), each of the pure states in the mixture is also strongly symmetric. Let $\ket{A}\ket{B}$ be one of the them.
    
    For the moment, we assume that there exists two anyons, $a$ and $b$, with nontrivial braiding $S_{ab} \neq 1$, with the all-fermion case treated similarly later. Then, let $\gamma$ and $\tau$ be loops that intersect at a point inside $B$ far away from its boundary, and don't intersect at any other point in $B$. Schematically, one can think of the curves illustrated in Fig. \ref{fig:braiding_def}. The corresponding string operators are symmetries of $\ket{A}\ket{B}$:
    \begin{equation}
        W_a(\gamma) \ket{A}\ket{B} = W_b(\tau) \ket{A}\ket{B} = \ket{A}\ket{B}.
    \end{equation}
    We reduce the equation above to region $B$ by taking the inner product with any product state $\ket{a} \in \Hilb_A$ in $A$ that has nonzero overlap with $\ket{A}$:
    \begin{equation}
        K_a(\gamma) \ket{B} = K_b(\tau) \ket{B} = \ket{B},
    \end{equation}
    where $K_a(\gamma) \defeq \braket{a | A}^{-1} \braket{a| W_a(\gamma)| A}$ and similarly to $K_b(\tau)$. Importantly, $K_a(\gamma)$ and $K_b(\tau)$ act in the same way as $W_a(\gamma)$ and $W_b(\tau)$, respectively, far from the boundary of $B$, due to the LU nature of the string operators and the absence of entanglement in $\ket{a}$. Hence, they acquire a braiding factor $S_{ab}$ when commuted:
    \begin{equation}
        K_a(\gamma) K_b(\tau) = S_{ab} K_b(\tau) K_a(\gamma),
    \end{equation}
    coming from the exchange of the gates of $W_a(\gamma)$ with the ones of $W_b(\tau)$ near the intersection point. 
    
    However, $K_a(\gamma) K_b(\tau)$ and $K_b(\tau) K_a(\gamma)$ are restrictions of the operators $W_a(\gamma) W_b(\tau)$ and $W_b(\gamma) W_a(\tau)$, respectively, which are both symmetries of $\ket{A}\ket{B}$ with eigenvalue $+1$. Accordingly, $K_a(\gamma) K_b(\tau)$ and $K_b(\tau) K_a(\gamma)$ must have $\ket{B}$ as an eigenvector with same eigenvalue $+1$, which contradicts $S_{ab} \neq 1$.

    For the fermionic case, one proceeds similarly by considering loops $\gamma$ and $\tau$ that intersect at a point inside $B$ and become the same curve until they go outside $B$ again, as considered in the definition of topological twist $\theta_a$ in Sec. \ref{sec:anomaly-top_twist} and shown in Fig. \ref{fig:twist_def}. Then, the same argument as above can be employed and reach a contradiction with $\theta_a \neq 1$. 
\end{proof}

\begin{theorem}\label{thm:LRBE}
    A state $\rho$ with an anomalous strong 1-form symmetry is long-range bipartite entangled, i.e. $\rho$ cannot be prepared from a bipartite separable state $\rho_{\text{2-sep}}$ via a stochastic local channel $\E = \sum_i p_i \E_i$, with respect to a sufficiently large bipartition $A|B$,
    \begin{equation}\label{eq:long-range_bipartite_entanglement}
        \rho \neq \sum_i p_i \E_i(\rho_{\text{2-sep}}).
    \end{equation}
\end{theorem}
\begin{proof}
    By contradiction, let us assume $\rho = \sum_i p_i \E_i(\rho_{\text{2-sep}})$. In this case, the bipartite separable state $\rho_{2-sep} \otimes \ketbra{\mathbf 0}{\mathbf 0}_a$ would inherit the anomalous symmetries of $\rho$ by symmetry pullback (Lemma \ref{lemma:anomaly_invariance}), which is prohibited by Lemma \ref{lemma:BE}.
\end{proof}

We end this section by noting that the theorem above, when applied to systems with strong fermionic one-form symmetries, solves an open question posed in \cite{wang_intrinsic_2025}. There, the authors proved that the ZX decohered toric code (See Sec. \ref{sec:zx_dephasing}) exhibits long-range bipartite entanglement, and conjectured that it generalizes to other systems with deconfined fermions. By Theorem \ref{thm:LRBE}, we see that it does.

\subsection{Topological entanglement of formation}
\label{sec:long-range_TEF}

Above, we have proved that a mixed state with strong anomalous 1-form symmetries must be long-range entangled. Now, we propose an information-theoretic measure to directly diagnose such long-range entanglement, namely the \textit{topological entanglement of formation}. As the name suggests, the TEF is a generalization to mixed states of the topological entanglement entropy for pure states \cite{kitaev_topological_2006, levinDetectingTopologicalOrder2006}, in a similar fashion to how the entanglement of formation \cite{bennett_mixedstate_1996}, a faithful mixed-state entanglement measure, generalizes the von Neumann entanglement entropy of pure states. The main result of this section is a lower bound to the TEF for states strongly symmetric under certain anomalous 1-form symmetries.

For topologically ordered fixed-point pure states, the topological entanglement entropy $\gamma$ measures the total quantum dimension $\mathcal{D} = \sqrt{\sum_a d_a^2}$ of the associated topological field theory, where $d_a$ is the quantum dimension of the anyon $a$ ($d_a = 1$ for Abelian anyons), as $\gamma = \log \mathcal{D}$ \cite{kitaev_topological_2006, levinDetectingTopologicalOrder2006}. The TEE can be calculated in multiple ways via linear combinations of entropic quantities: It is the subleading coefficient of the entanglement entropy of a large enough subregion $A$, $S_A = \alpha |\partial A| - \gamma$, but we will use the expression by Levin and Wen \cite{levinDetectingTopologicalOrder2006}, which relates the conditional mutual information (CMI) $I_{A:C|B} = S_{AB} + S_{BC} - S_B - S_{ABC}$ between a particular set of regions $A$ and $C$, conditioned on $B$, to the TEE via $I_{A:C|B} = 2 \gamma$. In the Levin-Wen prescription, $ABC$ forms an annulus, with $A$ and $C$ simply connected subregions separated by $B$, as in Fig.~\ref{fig:levin_wen}. 
The choice of the Levin-Wen scheme is because the CMI enjoys enough good properties to warrant a generalization to mixed-states useful to our purposes, with the most important being the strong subadditivity~\cite{lieb_fundamental_1973, lieb_proof_1973, nielsen_quantum_2010}, which says that the CMI is non-negative.

\begin{figure}[h]
    \centering
    \begin{tikzpicture}
    \definecolor{colorA}{HTML}{55AFE0}
    \definecolor{colorB}{HTML}{FFAE00}
    \definecolor{colorC}{HTML}{009E73}

    \fill[gray!30, rounded corners] (-1.5,-1.5) rectangle (1.5,1.5);

    \def\R{1}
    \def\r{0.3}
    \def\Rmid{0.5*\R+0.5*\r}
    \def\Rprime{0.66*\R+0.33*\r}
    \def\rprime{0.66*\r+0.33*\R}
    \def\RA{(0.85*\R+0.15*\r)}
    \def\centerarc[#1](#2)(#3:#4:#5)
    { \fill[#1] (0,0) -- ($(#2)+({(#5)*cos(#3)},{(#5)*sin(#3)})$) arc (#3:#4:#5); }
    
    \centerarc[color=colorA](0,0)(120:240:\R) 
    \centerarc[color=colorB](0,0)(60:120:\R) 
    \centerarc[color=colorB](0,0)(240:300:\R) 
    \centerarc[color=colorC](0,0)(-60:60:\R) 
    \fill[gray!30] (0,0) circle (\r);
    \draw (0,0) circle (\R);
    \draw (0,0) circle (\r);
    
    \foreach \angle in {60, 120, 240, 300} {
        \draw[black] ({\r*cos(\angle)},{\r*sin(\angle)}) -- ({\R*cos(\angle)},{\R*sin(\angle)});
    }
    
    \node at (0, \Rmid) {B};
    \node at ({-(\Rmid)}, 0) {A};
    \node at (\Rmid, 0) {C};
    \node at (0, {-(\Rmid)}) {B};
    
\end{tikzpicture}
    \caption{Regions in the Levin-Wen tripartition \cite{levinDetectingTopologicalOrder2006}. The CMI $I_{A:C|B}$ in this configuration is used to calculate the topological entanglement entropy and the topological entanglement of formation (See Def. \ref{def:TEF}).}
    \label{fig:levin_wen}
\end{figure}
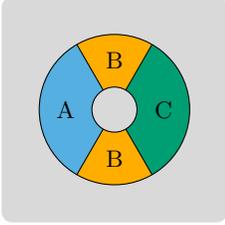

To start, we first define the mixed convex roof conditional mutual information (mcoCMI), denoted by $I^{\sqcup}_{A:C|B}(\rho)$,\footnote{For the mcoCMI $I^{\sqcup}_{A:C|B}(\rho)$, we adapt the notation used in \cite{szalay_multipartite_2015} to denote the convex roof extension of $f$ by $f^\cup$.} as follows:
\begin{definition}[Mixed convex roof conditional mutual information (mcoCMI)] \label{def:mcoCMI}
    The \textit{mixed convex roof conditional mutual information} $I^{\sqcup}_{A:C|B}(\rho)$ of a density matrix $\rho$ acting on $\Hilb = \Hilb_A \otimes \Hilb_B \otimes \Hilb_C \otimes \Hilb_{E}$, is the minimum over all mixed-state decompositions $\{p_i, \rho_i\}$ of $\rho = \sum_i p_i \rho_i$ of the average CMI among the states in each decomposition. That is, 
    \begin{equation}
        I^{\sqcup}_{A:C|B}(\rho) \defeq \min_{\{p_i, \rho_i \}} \sum_i p_i I_{A:C|B}(\Tr_E[\rho_i]).
    \end{equation}
\end{definition}

\begin{definition}[Topological entanglement of formation (TEF)]
\label{def:TEF}
    The \textit{topological entanglement of formation} $\gamma_F$ of $\rho$ is half of its mcoCMI $I^{\sqcup}_{A:C|B} = 2 \gamma_F$ in the Levin-Wen configuration (See Fig. \ref{fig:levin_wen}). It is the minimum over all decompositions of the global state $\rho$ of the average topological entanglement entropy.
\end{definition}

As mentioned earlier, the term ``of formation'' is named after the entanglement of formation, which is similarly defined as the (mixed) convex roof extension of the entanglement entropy\footnote{Due to the concavity of the von Neumann entropy, the entanglement of formation can be defined as a minimization over all pure or mixed decompositions of its argument. The same cannot be said about the mcoCMI.}. The above definition is closely related to the convex-roof extension of quantum conditional mutual information (``co(QCMI)'') studied in \cite{wang_analog_2024} as an analog of the TEE for mixed states. The only difference is that the minimization in the co(QCMI) is over all \textit{pure-state} decompositions $\{p_i, \ket{\psi_i}\}$ of $\rho = \sum_i p_i \ketbra{\psi_i}{\psi_i}$, a subset of the decompositions considered in the mcoCMI. We use the mixed convex roof extension here primarily because we can prove a lower bound to $\gamma_F$ based on the anomalous braiding of 1-form symmetries, which automatically implies a lower bound to the co(QCMI).

We state some useful properties of the mcoCMI below, leaving other properties and their proofs to Appendix \ref{appendix:further_props_mcoCMI}.

\begin{enumerate}
    \item \textit{Positivity}: $I^{\sqcup}_{A:C|B}(\rho) \geq 0$, by strong subadditivity.
    \item \textit{Convexity}: $I^{\sqcup}_{A:C|B}(p \rho + (1-p) \sigma) \leq p I^{\sqcup}_{A:C|B}(\rho) + (1-p) I^{\sqcup}_{A:C|B}(\sigma)$. \label{item:convex}
    \item \textit{Reduction to CMI for pure states}: $I^{\sqcup}_{A:C|B}(\ketbra{\psi}{\psi}) = I_{A:C|B}(\ketbra{\psi}{\psi})$. \label{item:CMI_pure-states}
    \item \textit{Upper bound by CMI}: $I^{\sqcup}_{A:C|B}(\rho) \leq I_{A:C|B}(\rho)$. \label{item:CMI_upper_bound}
    \item \textit{Monotonicity under strictly local mixed-unitary channels}: $I^{\sqcup}_{A:C|B}(\E_{\text{m}U}(\rho)) \leq I^{\sqcup}_{A:C|B}(\rho)$ for $\E_{\text{m}U} = \sum_i p_i U_i (\cdot) U_i^\dagger$, where all $U_i = U_A^{(i)} \otimes U_B^{(i)} \otimes U_C^{(i)} \otimes U_E^{(i)}$ act strictly locally. \label{item:monotonic_mixed_unitary}
\end{enumerate}

By definition, all of the properties above are valid for the TEF $\gamma_F$ as well. In particular, property \ref{item:CMI_pure-states} implies the TEF equals the TEE for pure states, and, more generally, property \ref{item:CMI_upper_bound} guarantees the TEF is always less than or equal to the TEE. Property \ref{item:monotonic_mixed_unitary} implies that the TEF cannot increase under on-site Pauli noise. In fact, all qubit-to-qubit unital channels are mixed unitary channels~\cite{landau_birkhoffs_1993, watrous_theory_2018}. Despite this, allowing for classical communication can increase the TEF. An example is given by an SPT cluster state in 2d protected by a $\mathbb{Z}_2$ 0-form $\times ~\mathbb{Z}_2$ 1-form symmetry. This is a short-range entangled state with zero $\gamma_F$. However, it's well-known that, under single-site measurement and single-site unitary operations conditioned on the measurement results, it can be turned into a toric-code ground state \cite{rbh_toric_2005}, which exhibits a $\mathbb{Z}_2$ topological order with a non-zero $\gamma_F$.

With the preceding definitions in hand, we first prove a lower bound for the TEE of states with strong-weak anomalous symmetries, and then use it to show a lower bound for the TEF in the case of purely strong symmetries.

\begin{theorem}\label{thm:CMI_bound-sw}
    Suppose $\rho$ a state strongly (weakly) symmetric under one-form symmetries corresponding to anyons labels in $\mathcal{S}_\rho$ (in $\mathcal{W}_\rho \supset \mathcal{S}_\rho$), where the weak symmetries are on-site up to an LU $U$: $\forall b \in \mathcal{W}_\rho,\ W_b(\gamma) = U^ \dagger \prod_{x \in \gamma} w^{(b)}_x U$. Then the TEE $\gamma$ is bounded from below by $\frac{1}{2}\log n$, where $n = \rk ([S_{ab}]_{a \in \mathcal{S}_\rho, b\in \mathcal{W}_\rho})$ is the number of weak symmetry anyons $b \in \mathcal{W}_\rho$ having unique braiding phases with the strong symmetry anyons $a \in \mathcal{S}_\rho$.
\end{theorem}
The proof follows essentially the same argument presented in \cite{kim_universal_2023, levin_physical_2024} for the lower bound to the TEE of the toric code phase or more general pure-state phases described by TQFT. The crucial difference compared to the pure state argument is that, for mixed states, the open strings of weak anyons $W_b', b \in \mathcal{W}_\rho,$ are used to create particles in the hole of the annulus, while the strong symmetry loops $W_a, a \in \mathcal{W}_\rho$, detect them by winding around when they braid nontrivially, $S_{ab} \neq 1$.

\begin{figure}[t]
    \centering
    \begin{tikzpicture}
    \def\R{3}
    \def\r{1.2}
    \def\Rmid{0.5*\R+0.5*\r}
    \def\Rprime{0.66*\R+0.33*\r}
    \def\rprime{0.66*\r+0.33*\R}
    \def\RA{(0.85*\R+0.15*\r)}
    \def\RWa{0.55*\R}
    \def\centerarc[#1](#2)(#3:#4:#5)
    { \draw[#1] ($(#2)+({(#5)*cos(#3)},{(#5)*sin(#3)})$) arc (#3:#4:#5); }
    
    \fill[gray!40] (0,0) circle (\R);
    \fill[white] (0,0) circle (\r);
    \draw (0,0) circle (\R);
    \draw (0,0) circle (\r);
    
    \foreach \angle in {60, 120, 240, 300} {
        \draw[black] ({\r*cos(\angle)},{\r*sin(\angle)}) -- ({\R*cos(\angle)},{\R*sin(\angle)});
    }
    
    \node[font=\Large] at ({\RA*cos(140)}, {\RA*sin(140)}) {A};
    \node at ({-(\Rmid)}, 0) {A'};
    \node[font=\Large] at (0, {-(\Rmid)}) {B};
    \node[font=\Large] at (0, \Rmid) {B};
    \node[font=\Large] at (\Rmid, 0) {C};
    
    \centerarc[dashed](0,0)(120:240:\Rprime)
    \centerarc[dashed](0,0)(120:240:\rprime)
    
    \draw[red, very thick] (-3.5,1) .. controls (-2,-1) and (-1,2) .. (0,0);
    \fill[red] (-3.5,1) circle (0.1) node[above] {$W_b'$};
    \fill[red] (0,0) circle (0.1);
    
    \draw[blue, very thick] (0,0) ellipse ({1.1*\Rprime} and {\rprime});
    \node[blue] at (\RWa, \RWa) {$W_a$};
    
\end{tikzpicture}
    \caption{Configuration of regions and Wilson lines used in the proof of  \ref{thm:CMI_bound-sw}. Importantly, the strong one-form symmetry $W_a$ inside the annulus $R=ABC$ detects the action of the open string operator $W_b'$, if $S_{ab} \neq 1$.}
    \label{fig:TEF_proof}
\end{figure}
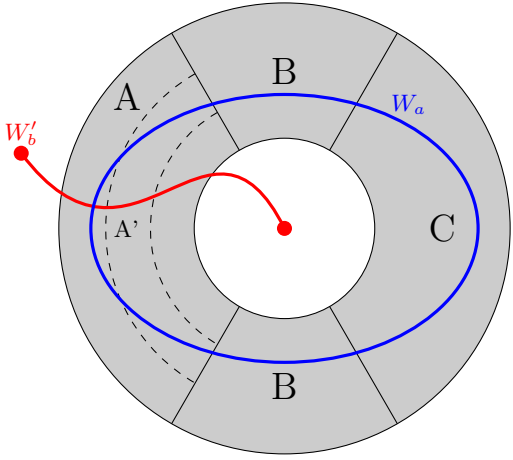

\begin{proof}

Consider an annular region $R$ with a linear size much larger than the depth of $U$ and divided into three regions $ABC$ in the Levin-Wen configuration (See Fig. \ref{fig:TEF_proof}). We will construct a family of states $\{ \rho_i \}$ in $R$ satisfying \cite{sutter_universal_2016}:
\begin{enumerate}
    \item \textit{Orthogonal supports}, $F(\rho_i, \rho_j) = \delta_{ij}$ ($F$ being the fidelity).
    \item \textit{Locally indistinguishable}, $\Tr_{\bar{G}} \rho_i = \Tr_{\bar{G}} \rho_j$ for simply connected regions $G$ in the annulus. 
    \item \textit{Homentropic}, $S(\rho_i) = S(\rho_j)$.
\end{enumerate}

We construct this family of states by first choosing a set $\mathcal{W}^{(n)} \subseteq \mathcal{W}_\rho$ of $n = \rk ([S_{ab}]_{a \in \mathcal{S}_\rho, b\in \mathcal{W}_\rho})$ weak anyons labels so that the phase vectors $\{[S_{ab}]_{a\in \mathcal{S}_\rho} \mid b \in \mathcal{W}^{(n)}\}$ are linearly independent. Then, by applying a weak $b \in \mathcal{W}^{(n)}$ anyon open string $W_b'$ that terminates inside the inner circle of the annulus to the state $\rho$, and then restricting the resulting state to the annulus $R$, we have a family of states $\rho_b \defeq \Tr_{\comp{R}}[W_b' \rho (W_b')^\dagger]$ with the desired properties, for which we argue below.

For the orthogonality property, first note that the charge of the state $\rho_b$ under a strong symmetry loop $W_a(\gamma)$ that goes around the annulus is the braiding phase $S_{ab}$, since
\begin{equation}
    W_a \rho_b = S_{ab} \Tr_{\comp{R}}[W_b' W_a \rho (W_b')^\dagger] =  S_{ab} \rho_b.
\end{equation}
From the definition of $\mathcal{W}^{(n)}$, for any two weak anyons $b, c \in \mathcal{W}^{(n)}$, there exists at least one strong anyon $a \in \mathcal{S}_\rho$ such that $S_{ab} \neq S_{ac}$. Since the unequal strong symmetry charges $S_{ab}$ and $S_{ac}$ are inherited to all states in the decompositions of $\rho_b$ and $\rho_c$, respectively, then their supports are necessarily orthogonal.

For the local indistinguishability of $\rho_b$ in contractible regions $G$ inside the annulus, we use the fact that the weak open string $W_b'(\gamma)$ that creates a $b$ anyon inside the annulus can be deformed to another open string $W_b'(\tilde\gamma)$ whose support does not intersect $G$, while having the same action on $\rho$. Hence,
\begin{equation}
    \Tr_{\comp{G}}[\rho_b] = \Tr_{\comp{G}}[W_b'(\tilde\gamma) \rho W_b'(\tilde\gamma)^\dagger] = \Tr_{\comp{G}}[\rho]
\end{equation}

Finally, to argue that all states $\rho_b$ have the same entropy, we will first make the seemingly artificial assumption that $U$ acts trivially on $A$, which we will later lift to complete the proof. In such case, $W_b'(\gamma) = U^\dagger W_{0,b}'(\gamma) U$ acts on-site on $A$ if $\gamma$ is a curve passing through $A$ (See Fig. \ref{fig:TEF_proof}). Because of the on-site action, the entropy $S(\rho_b)$ of $\rho_b$ in $R = ABC$ is the same as the original entropy of $\rho_e \equiv \rho$, in which no string is applied. 

With the three properties above, we can bound the TEE by the Shannon entropy $H(p)$ of an arbitrary probability distribution $(p_c)_{c\in \mathcal{W}^{(n)}}$ over the states $\rho_c$. To that end, we consider the ensemble $\lambda \defeq \sum_c p_c \rho_c$ and compute:
\begin{align}
    I_{A:C|B}(\rho) & = I_{A:C|B}(\lambda) + S(\lambda) - S(\rho) \\
    & \geq S(\lambda) - S(\rho) \\
    & \geq H(p) + \sum_b p_b [S(\rho_b) - S(\rho)] \\ 
    & = H(p), \label{eq:fine_tuned_TEE_proof}
\end{align}
where in the first line we used that $\lambda$ is locally indistinguishable from $\rho$, in the second line we used strong subadditivity, in the third line we used that the $\{ \rho_b\}$ states are orthogonal, and finally, we used that $W_b'$ acts on-site on $R$ for $S(\rho_b) = S(\rho)$. Since the probability distribution $p$ is arbitrary, we can take the uniform $p_c = 1/n$, giving $I_{A:C|B}(\rho) \geq \log n$, as expected.

Now, let us lift the assumption that $U$ acts trivially in $A$ by considering an annular region $R'=A'BC$ similar to the original one $R$, but with $A$ contracted radially to a smaller region $A' \subset A$ (See Fig. \ref{fig:TEF_proof}). 
We then gather the gates from $U$ that act deep inside $A$ in a unitary $V$ so that $U' \defeq U V^\dagger$ acts trivially on a curve $\gamma$ passing through $A'$. We also define a new state $\rho' \defeq V \rho V^\dagger$, which is weakly symmetric under $W_b'(\gamma) \defeq (U')^\dagger \prod_{x \in \gamma} w_x^{(b)} U'$. Since the entanglement entropy of a subregion does not change when a unitary acts strictly inside it or on its complement, then the CMI of $\rho$ is the same as the CMI of $\rho'$:
\begin{equation}\label{eq:remove_unentangling_gates}
    I_{A:C|B}(\rho) = I_{A:C|B}(\rho').
\end{equation}
Now, we reduce $A$ to $A'$. Since $I_{A:C|B} - I_{A':C|B} = I_{A\setminus A' : C | A', B}$, then by strong subadditivity, we have 
\begin{equation}\label{eq:contract_region_SSA}
    I_{A:C|B}(\rho') \geq I_{A':C|B}(\rho').
\end{equation}
Since $W'_b(\gamma)$ acts on-site in $A'$, the conditions for Eq. \eqref{eq:fine_tuned_TEE_proof} apply for $\rho'$, and we have 
\begin{equation}\label{eq:contracted_TEE_proof}
    I_{A':C|B}(\rho') \geq \log n.
\end{equation}
Finally, by chaining together Eqs. \eqref{eq:remove_unentangling_gates}, \eqref{eq:contract_region_SSA} and \eqref{eq:contracted_TEE_proof}, we arrive at the desired lower bound for the TEE $\gamma$ of $\rho$: $\gamma = \frac{1}{2} I_{A:C|B}(\rho)\geq \frac{1}{2} \log n$.
\end{proof}

Due to the inheritance property of strong symmetries, we can leverage Theorem \ref{thm:CMI_bound-sw} to bound the TEF of states with anomalous strong symmetries:
\begin{corollary}\label{cor:CMI_bound-ss}
    Suppose $\rho_0$ a state strongly symmetric under an anomalous one-form symmetry with on-site action $W^{(0)}_a(\gamma) = \prod_{x \in \gamma} w^{(a)}_x \in G_{\rho_0}$, corresponding to a (2+1)-D topological order with abelian anyons $a \in \mathcal{S}_\rho$. then the TEF $\gamma_F$ of any other state $\rho$ in the same mixed state phase of $\rho_0$ is bounded from below by $\frac{1}{2}\log n$, where $n = \rk ([S_{ab}]_{a,b \in \mathcal{S}_\rho})$ is the number of strong symmetry anyons $b \in \mathcal{W}_\rho$ having unique braiding phases with other strong symmetry anyons $a \in \mathcal{S}_\rho$, including possibly themselves.
\end{corollary}
\begin{proof}
    Let $\rho$ be a state in the same mixed-state phase of $\rho_0$. Each state $\rho_i$ of any decomposition of $\rho = \sum_i p_i \rho_i$ inherits the anomalous symmetries $W^{(0)}_a$ of $\rho_0$ via symmetry pullback. Such states satisfy the conditions of Theorem \ref{thm:CMI_bound-sw} with $n = \rk ([S_{ab}]_{a,b \in \mathcal{S}_\rho})$. Hence,
    \begin{align*}
        2 \gamma_F & = \min_{\{p_i, \rho_i\}} \sum_i p_i I_{A:C|B}(\Tr_E[\rho_i]) \\
        & \geq \min_{\{p_i, \rho_i \}} \sum_i p_i \log n \\
        & \geq \log n.
    \end{align*}
\end{proof}

\section{Example: toric code under local noise}
\label{sec:example_toric_code}

We now apply the formalism of (strong and weak) anomalous higher-form symmetries described in section \ref{sec:anomaly_def} to the toric code state under different decoherence noises. The choice of the toric code is due to its analytical simplicity and extensive understanding in the literature, including recent works on mixed states phases of matter. We will argue here that, not only the long-range entanglement of the toric code under weak Pauli-$X$ and $Z$ noise is guaranteed by the 1-form anomaly, but also that the value of its TEF is exactly known by combining the results of section \ref{sec:long-range_TEF} with the explicit decompositions presented in \cite{wang_analog_2024}. The toric code under $ZX$ dephasing and the intrinsically mixed-state TO it generates at high noise strength will also be discussed.

First, we briefly review the definition and basic properties of the (two-dimensional square-lattice) toric code \cite{kitaev_faulttolerant_2003}. The toric code is the subspace of states stabilized by the star $A_v$ and plaquette $B_p$ operators defined as

\def\starterm{\begin{tikzpicture}[scale=0.2, baseline=-0.5ex]
    \draw[] (0, -3) -- (0, 3); 
    \draw[] (-3, 0) -- (3, 0); 
    
    \node[red] at (0, 1.5) {Z}; 
    \node[red] at (1.5, 0) {Z};  
    \node[red] at (0, -1.5) {Z}; 
    \node[red] at (-1.5, 0) {Z};  
\end{tikzpicture}}

\def\plaqterm{\begin{tikzpicture}[scale=0.2,baseline=-0.5ex]
    \draw[] (-1.5, -1.5) -- (-1.5, 1.5);
    \draw[] (-1.5, 1.5) -- (1.5, 1.5);
    \draw[] (1.5, 1.5) -- (1.5, -1.5);
    \draw[] (1.5, -1.5) -- (-1.5, -1.5);
    
    \node[blue] at (0, 1.5) {X}; 
    \node[blue] at (1.5, 0) {X};  
    \node[blue] at (0, -1.5) {X}; 
    \node[blue] at (-1.5, 0) {X};  
\end{tikzpicture}}

\begin{equation}
    A_v = \starterm, \qquad B_p = \plaqterm.
\end{equation}

Equivalently, it is the ground state subspace of the stabilizer Hamiltonian 
\begin{equation}
    \label{eq:TC_Hamiltonian}
    H_{TC} = - \sum_v A_v - \sum_p B_p.    
\end{equation}
Importantly, it is degenerate when $H_{TC}$ is embedded in a base space with nontrivial topology, such as a torus. However, for now, we will focus on local properties of the toric code, so one can assume $H_{TC}$ is defined on a plane. In that case, the ground state is unique, and we denote it by $\ket{\text{TC}}$. From the stabilizers constraints, it follows that $\ket{\text{TC}}$ is the superposition of all loops in the $Z$ basis:
\begin{equation}
    \ket{\text{TC}} \propto \sum_{\text{loops } \ell} \ket{\ell}
\end{equation}

The one-form symmetries of $H_{TC}$ are generated by loops of Pauli-$X$ operators in the direct lattice, and by loops of Pauli-$Z$ operators in the dual lattice. When opened, the two string operators create so-called electric ($e$) and magnetic ($m$) anyons, respectively, which exhibit nontrivial mutual braiding $S_{em} = -1$ but trivial self-braiding (and thus bosonic statistics). When brought together, the $e$ and $m$ anyons form a fermionic anyon $f = e m$.

One of the most remarkable properties of the toric code is its stability to perturbations. These can assume many forms, but perturbations via local noise channels are the ones relevant to the present discussion. In particular, if a noise channel $\E$ is applied to the toric code ground state, and the resulting state $\E(\ketbra{\text{TC}}{\text{TC}})$ is in the same long-range entanglement phase, then by invariance of the anomaly under symmetry pullback, we know that $\E(\ketbra{\text{TC}}{\text{TC}})$ still has strong one-form symmetries corresponding to $e$, $m$ and $f$ particles, along with their expected braiding and self-statistics. Hence, by Theorem \ref{thm:LRBE}, $\E(\ketbra{\text{TC}}{\text{TC}})$ possess long-range bipartite entanglement. 

Note that the line of reasoning above does not prove that the toric code phase is stable against sufficiently weak perturbations, but rather explains the invariance of its topological properties assuming the perturbed state is two-way connected to the clean state. To say more about the noisy state, we will now focus on two particular noise models: dephasing under Pauli-X and Z operators (Sec. \ref{sec:pauli-x_z_dephasing}), and ZX dephasing (Sec. \ref{sec:zx_dephasing}). The former has been thoroughly studied in the literature from the point of view of loss of quantum memory \cite{dennis_topological_2002}, and the latter sparked recent interest as it gives rise to an intrinsic mixed-state phase at high noise strength \cite{wang_intrinsic_2025, ellison_classification_2025, sohal_noisy_2025, li_how_2024}. In both cases, our main contribution will be to reinterpret their phase diagrams in light of one-form anomalies, the symmetry pullback mechanism, and the TEF.

\subsection{Pauli-X and Z dephasing}
\label{sec:pauli-x_z_dephasing}

The most well-studied noise model is dephasing against Pauli-$Z$ and $X$ operators, where, independently on each site, the $X$ unitary is applied with probability $p_X$ and the $Z$, with probability $p_Z$, resulting in the channel
\begin{align}
    \E(p_X, p_Z) & \defeq \bigotimes_{e} \E_e(X, p_X) \circ \E_e(Z, p_Z),
\end{align}
where $\E_e(U, p) [\rho] \defeq (1-p) \rho + p U_e \rho U_e^\dagger$, is a mixed-unitary channel acting on the edge $e$. When acting on the toric code ground state $\ket{\text{TC}}$, it results in
\begin{equation}
    \rho_{p_X, p_Z} \defeq \E(p_X, p_Z)[\ketbra{\text{TC}}{\text{TC}}].
\end{equation}
Any probability $0 \leq p_{X},p_Z \leq 1$ can be considered, but we will limit ourselves to $0 \leq p_{X},p_Z \leq 1/2$, since the other cases are equivalent up to global Pauli rotations.

\subsubsection{Phase diagram}

We will now identify the phase diagram of the states $\rho_{p_X,p_Z}$ in the $(p_X, p_Z)$ phase space, culminating in Fig. \ref{fig:phase_diagram}.

\begin{figure}[t]
    \centering
    \begin{tikzpicture}

\def\xmax{6}
\def\ymax{6}
\def\pc{2.2}

\fill[green!20] (0,\pc) rectangle (\xmax,\ymax); 
\fill[green!20] (\pc,0) rectangle (\xmax,\pc);   
\fill[green!20] (\pc,\pc) rectangle (\xmax,\ymax); 

\fill[blue!20] (0,0) rectangle (\pc,\pc);

\draw[thick, ->] (0,0) -- (\xmax+0.4,0) node[right] {$p_x$};
\draw[thick, ->] (0,0) -- (0,\ymax+0.4) node[above] {$p_z$};

\draw[thick] (\xmax, 0.1) -- (\xmax, -0.1) node[below] {$1/2$};

\draw[very thick] (0,\pc) -- (\pc,\pc);
\draw[very thick, dashed] (\pc,\pc) -- (\xmax,\pc);

\draw[very thick] (\pc,0) -- (\pc,\pc);
\draw[very thick, dashed] (\pc,\pc) -- (\pc,\ymax);

\node[below] at (\pc,0) {$p_c$};
\node[left] at (0,\pc) {$p_c$};

\fill (0,0) circle (2pt) node[below left] {$\ket{\text{TC}}$};
\fill (0,\ymax) circle (2pt) node[left] {$\rho_{\text{loop}}$};
\fill (\xmax,\ymax) circle (2pt) node[above right] {$\one$};

\node[green!70!black, font=\Large] at (0.5* \pc + 0.5*\xmax,0.5* \pc + 0.5*\xmax) {SRE};
\node[black, font=\small] at (0.5* \pc + 0.5*\xmax,0.5* \pc + 0.5*\xmax - 0.5) {TEF = 0};

\node[blue!70!black, font=\Large] at (0.5* \pc, 0.5* \pc) {TC};
\node[black, font=\small] at (0.5* \pc, 0.5* \pc - 0.5) {TEF = $\log 2$};

\draw[red, thick, ->] (\pc, \ymax) -- (\pc + 0.3, \ymax + 0.5) node[right] {$\xi$};

\begin{scope}[shift={(\pc,\ymax-0.4)}]
\draw[red, thick, smooth, domain=-0.5:0.5, samples=20] 
        plot (\x, {0.6*exp(-(4*\x)^2 )});
\end{scope}


    
    
    
    
\end{tikzpicture}
    \caption{Phase diagram of the toric code state under Pauli-$X$ and $Z$ dephasing. in blue, the toric code phase (TC), and in green, the short-range entangled (SRE) phase. The Markov length $\xi$ diverges at the dashed and solid lines inside the diagram. However, only the solid ones signal a transition from long-range to short-range entanglement.
    }
    \label{fig:phase_diagram}
\end{figure}
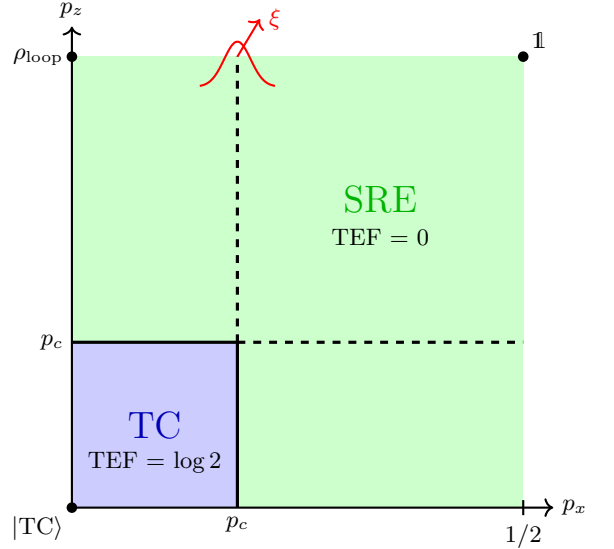

\begin{figure}
    \centering
    \begin{tikzpicture}
    \definecolor{colorA}{HTML}{55AFE0}
    \definecolor{colorB}{HTML}{FFAE00}
    \definecolor{colorC}{HTML}{009E73}

    \fill[colorC, rounded corners] (-1.5,-1.5) rectangle (1.5,1.5);
    
    \node[draw, circle, fill=colorB, minimum size=2.5cm] (B) at (0,0) {};
    
    \node[draw, circle, fill=colorA, minimum size=1.2cm] (A) at (0,0) {A};
    
    \draw[<->, thick] (A) -- node[midway, above left=-1mm] {d} ++(0.9,0.9);
    
    \node at (-0.65,0.65) {B};
    \node at (-1.2,1.2) {C};
\end{tikzpicture}
    \caption{Configuration of regions $A,B$ and $C$ for the calculation of the ``Markov length'' CMI $I_{A:C|B} \sim \exp(-d/\xi)$. Crucially, $B$ separated $A$ and $C$ by a distance $d$.}
    \label{fig:CMI_markov}
\end{figure}
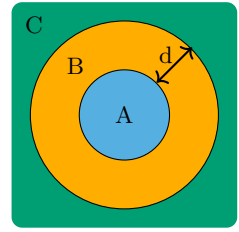

Due to the topological nature of the toric code, we expect that, up to finite decoherence probabilities $p_X, p_Z$, the noisy toric code state $\rho_{p_X, p_Z} \defeq \E(p_X,p_Z) (\ketbra{\text{TC}}{\text{TC}})$ should remain in the same mixed-state phase. We can guarantee this by employing the methods described in \cite{sang_stability_2024}. There, it is proven that a finite-time local Lindbladian evolution $\rho(t) = e^{\int \mathcal{L}(t) dt} \rho_0$ can be reversed if the CMI $I_{A:C|B}$ between regions $A$ and $C$ separated by $B$ with width $d$ (See Fig. \ref{fig:CMI_markov}) decays exponentially throughout the evolution: $I_{A:C|B}[\rho(t)] = O(\exp(-d/\xi))$. The correlation length $\xi$ is called the \textit{Markov length}, and it further shown that, for the $Z$-dephased toric code (so $p_X = 0$), it diverges only at $p_c \approx 0.11$. We can thus conclude that there are at most two phases of states $\rho_{0, p_Z}$, $0 \leq p_Z \leq 1/2$, possibly separated at $p_c$. Since the ``loop soup'' state at $p_Z = 1/2$ is fully separable:
\begin{equation}\label{eq:loop_soup_state}
     \rho_{0, 1/2} = \rho_{\text{loop}} \propto \sum_{\text{loops } \ell} \ketbra{\ell}{\ell},
\end{equation}
then the entire phase at $p > p_c$ is also trivial according to Def. \ref{def:phase_of_matter}.

Furthermore, we can characterize the two phases from the point of view of anomaly. For $0 \leq p_Z < p_c$, the anomalous strong 1-form symmetries of the toric code ground state, corresponding to the anyon data $\{1, e, m, f\}$, are pulled back to the entire phase. From Theorem \ref{thm:LRBE}, we can deduce from the strong-strong anomaly that the toric code phase is long-range bipartite entangled. On the other hand, for $p_c < p_Z \leq 1/2$, the $W_e$ and $W_f$ symmetries become weak, and only $W_m$ (and, trivially, $W_1 = \one$) remains strong.

We can extend this result to the entire phase diagram $0 \leq p_x, p_Z \leq 1/2$ by noticing that the $Z$ and $X$ errors decouple, and the Markov CMI satisfies
\begin{equation}
    \label{eq:CMI_decoupling}
    I_{A:C|B}(\rho_{p_X, p_Z}) = I_{A:C|B}(\rho_{0, p_Z}) + I_{A:C|B}(\rho_{p_X, 0}).
\end{equation}
See Appendix \ref{appendix:noisy_TC_decoupling} for a detailed derivation. From the equation above, we know that the Markov length diverges exactly at $p_X = p_c$ or $p_Z = p_c$, separating the phase space into four disjoint regions (See Fig. \ref{fig:phase_diagram}). Again, this can indicate up to four different phases, but in this case there are only two: the toric code phase and the trivial SRE phase. Indeed, the phase at $p_c < p_X, p_Z \leq 1/2$ is also trivial because $\rho_{1/2,1/2} \propto \one$ is the maximally mixed state.

One might ask, then, if the transition lines separating the two trivial regions (dashed in Fig. \ref{fig:phase_diagram}) signal any meaningful change in the system that is not related to long-range entanglement. Indeed, we will argue in \ref{sec:classical_and_quantum_memories} that the transition is between a system displaying topological, but classical memory to one without a topological memory. Although this transition has already been pointed out in the literature \cite{bao_mixedstate_2023, wang_intrinsic_2025, liReplicaTopologicalOrder2024, zhang_strongweak_2024}, our contribution is to connect the classical memory to the anomaly between strong and weak one-form symmetries. Furthermore this shows the need for a more refined classification that distinguishes phases exhibiting quantum, classical and no memory, which will be further explored in another paper \cite{sang_mixedstate_2025}. 

\subsubsection{TEF}

In the last section, we determined the entire phase diagram of the toric code under $X$ and $Z$ noise. Here, we will use similar techniques to show how the TEF can be calculated.

We start with the toric code ground state $\ket{\text{TC}}$ at $p_X=p_Z=0$. Since it is a pure state, the TEF equals the TEE, which can be explicitly calculated to be $\log 2$. This value remains constant in the entire TC phase because it cannot increase, as the noise is on-site (See property \ref{item:monotonic_mixed_unitary} in Sec. \ref{sec:long-range_TEF}), nor decrease, as pulling back the strong $e$ and $m$ one-form symmetries from $\ket{\text{TC}}$ implies a lower bound of $\log 2$ due to their anomaly (See Corollary $\ref{cor:CMI_bound-ss}$). 

Outside the TC phase, our lower bound for the TEF gives $0$, since there is no guarantee of an anomaly between two strong symmetries. At $p_X = 0$ line, for example, the one-form symmetry of the $Z$ operators remains strong throughout, but the $X$ one-form symmetry becomes weak for any $p_Z > 0$. Moreover, showing that the TEF for $p_Z > p_c$ is exactly zero would mean that not even emergent strong symmetries with nontrivial braiding can be constructed.

To confirm that the TEF is indeed zero, we must show a particular decomposition of $\rho_{p_X, p_Z}$ into states with zero TEE. This was pursued in Refs. \cite{wang_analog_2024, chen2024separability} for the case $p_X = 0$ by considering the decomposition $\rho = \sum_{z} \ketbra{\Psi_z}{\Psi_z}$, where $\ket{\Psi_z} = \sqrt{\rho} \ket{z}$ and $\ket{z}$ is a $Z$-basis product state. In \cite{wang_analog_2024}, the expectation value of a non-contractible Pauli-$Z$ loop operator and the Rényi-2 version of the TEE for each state $\ket{\Psi_z}$ were related to observables in the RBIM, strongly suggesting a transition from a topological phase at $p < p_c$ to a trivial phase at $p > c$. In particular, the Rényi-2 TEE showed a jump from $\log(2)$ to $0$, independently of the ensemble state $\ket{\Psi_z}$. Later, this was confirmed by numerical studies \cite{chen2024separability}.  

More generally, the preceding calculation of the TEF for $p_X = 0$ and $p_Z > p_c$ determines the TEF for the whole trivial phase, i.e. even for $p_X > 0$. That is because the state $\rho_{p_X, p_Z}$ can be reached from $\rho_{0, p_Z}$ by applying the on-site Pauli-$X$ noise $\bigotimes_e \E_e(X, p_X)$: since the TEF cannot increase under on-site Pauli noise, it remains zero for $(p_X, p_Z) \in [0,1/2] \times (p_c, 1/2]$. By exchanging $X$ with $Z$ in the reasoning above, we reach the same conclusion for the entire trivial phase.

We end this section by noting that, even though the TEF of the trivial phase is zero, the TEE for $p_X=0$ (or $p_Z=0$) is lower bounded by $\log 2$ due to the braiding between the strong $W_m$ ($W_e$) and the weak $W_e$ ($W_m$) symmetries (See Theorem \ref{thm:CMI_bound-sw}).

\subsection{ZX dephasing - fermionic intrinsically mixed-state topological order}
\label{sec:zx_dephasing}

In the preceding section, we studied the phase diagram of the toric code under Pauli-X and Z dephasing. There, at a sufficiently high noise strength, the system transitions from a topological ordered phase to a trivial, SRE phase, due to the decoherence of $e$ and/or $m$ anyons. Recently \cite{wang_intrinsic_2025}, a new pattern of decoherence was proposed that changes the picture above. If, instead, the fermionic \emph{strong} one-symmetry is preserved, then long-range entanglement is still guaranteed from the nontrivial self-statistics. As we will see, this is accomplished by dephasing with a two-body operator $Z_e X_{e+\delta}$, where $\delta$ can be taken to be $\delta = (-1/2,1/2)$. Namely, the noise channel is
\begin{align}
    \label{eq:ZX_decoherence_channel}
    \E^{ZX}(p) & \defeq \prod_{e} \E^{ZX}_e (p),
\end{align}
where $\E_e^{ZX}(p) [\rho] \defeq (1-p) \rho + p Z_e X_{e+\delta} \rho X_{e+\delta} Z_{e}$. When acting on the toric code ground state $\ket{\text{TC}}$, it results in
\begin{equation}
    \rho^{ZX}_{p} \defeq \E^{ZX}(p)[\ketbra{\text{TC}}{\text{TC}}].
\end{equation}
We will only consider $p \leq 1/2$, since other values are equivalent up to a global Pauli-$Y$ rotation.

As alluded to earlier, the $ZX$ dephasing maintains the \textit{strong} fermionic one-form symmetry $W_f(\gamma) \defeq \prod_{e \in \gamma} X_e Z_{e+\delta}$, since $[W_f(\gamma), Z_e X_{e+\delta}]=0$ for all edges $e$. As expected, it can be checked explicitly that $W_f$ creates excitations with fermionic self-statistics, as defined in Sec. \ref{sec:anomaly-top_twist}. Thus, from Theorem \ref{thm:LRBE}, $\rho^{ZX}_p$ has long-range bipartite entanglement for any $p \in [0,1/2]$.

There is, however, a phase transition at $p_c \approx 0.11$ that separates the toric code phase at $p < p_c$ from an intrinsically mixed-state topological ordered (imTO) phase at $p > p_c$. First, let us motivate why a transition out of the toric code phase is expected, and then we will characterize the $p > p_c$ phase.

By mapping the coherent information to observables of the RBIM, Ref. \cite{wang_intrinsic_2025} was able to show a transition from a system supporting a quantum memory to one at $p > p_c$ having lost the quantum memory. The Markov length is also expected to diverge only at $p_c$\footnote{More specifically, by viewing the toric code as being stabilized by $A_v$ and $C_v = A_v B_{v-\delta}$, the $ZX$ noise only decoheres the $A_v$ stabilizers, and in the same way a $X$-dephasing noise of the same strength $p_X=p$ would. This gives rise to the same Markov CMI behavior.}. These statements already give strong evidence for a phase transition at $p_c$. 

We can learn more about the phase for $p > p_c$ by analyzing the maximally decohered state at $p = 1/2$. It can be checked that it is the maximally mixed state with strong fermionic one-form symmetry $W_f$, and thus stabilized only by the star-plaquette stabilizers $C_v = A_v B_{v-\delta}$. In particular, it cannot have other strong symmetries braiding nontrivially with $f$ (such as the ones of $e$ and $m$ anyons), which remains true for all states in the $p > p_c$ phase due to symmetry pullback. This has been termed an intrinsically mixed-state topological order because pure and gapped ground states are expected to be described by a \textit{modular} TQFT, which prohibits transparent anyons that do not braid with any other anyon (except for the identity particle) \cite{simon_topological_2023}.

In the imTO phase, the TEE is lower bounded by $\frac{1}{2} \log 2$, due to the braiding of the strong symmetry $W_f$ with the weak symmetries, $W_e$ and $W_m$ (See Theorem \ref{thm:CMI_bound-sw}). The TEF, on the other hand, has no positive lower bound for the TEF coming from corollary \ref{cor:CMI_bound-ss}, since the fermion $f$ braids trivially with itself. If one restricts the decompositions of the mixed convex roof $\gamma_F = \min_{\{p_i, \rho_i\}} \sum_i p_i \gamma(\rho_i)$ to gapped ground states $\rho_i = \ketbra{\text{GS}}{\text{GS}}$ described by modular TQFTs, then the restricted TEF $\gamma_{F,\text{gapped}}$ would be lower bounded by $\log 2$. This is because of the modularity condition discussed earlier. By strong symmetry inheritance, all such states would be symmetric under the fermionic one-form symmetry $W_f$, which by modularity would imply the existence of another anyon $a$ with which $f$ braids nontrivially, thus giving a TEE $\gamma(\rho_i) \geq \log 2$. 

We cannot, however, rule out the possibility of more exotic states (e.g. gapless states) having the strong fermionic one-form symmetry but zero TEE. Thus, we leave the determination of the TEF of the imTO phase as an open question.

A computable quantity that is believed to capture topological order in mixed states is the topological entanglement negativity (TEN) \cite{Lu_finite_T_TO_2020}. For states $\rho$ with area-law entanglement, it is the universal constant $\gamma_N$ appearing in the entanglement negativity in a region $A$, $E_N(\rho) = \alpha | \partial A| - \gamma_N$. The TEN of the state $\rho_{p=1/2}$ can be computed to be $\frac{1}{2}\log 2$ or $\log 2$, depending on the region $A$. See Appendix \ref{appendix:negativity} for the detailed calculation for the $ZX$-dephased toric code and Kitaev's honeycomb model.

\subsection{Classical and quantum memories}
\label{sec:classical_and_quantum_memories}

Here, we briefly mention the connection between topological quantum and classical memories and one-form anomalies. We do not aspire to tackle the connection in its full generality, but rather to showcase it via the examples discussed earlier in Secs. \ref{sec:pauli-x_z_dephasing} and \ref{sec:zx_dephasing}, while hinting at a general structure.

\subsubsection{Quantum memory and strong-strong anomaly}
\label{sec:quantum_memory}

When wrapped around a torus, the ground state subspace of the toric code Hamiltonian (Eq. \ref{eq:TC_Hamiltonian}) is 4-dimensional. This degeneracy can be found by counting the number of independent $A_s$ and $B_p$ stabilizers, but, more importantly to us, it can also be viewed as a consequence of $e$ and $m$ one-form symmetries on non-contractible loops around the torus. More precisely, if $W_e(\ell_{h(v)})$ and $W_m(\ell_{h(v)})$ are the one-form symmetries of the $e$ and $m$ anyons, respectively, over a non-contractible loop $\ell_{h(v)}$ in the horizontal (vertical) direction, then braiding implies
\begin{align}
    \label{eq:non_contractible_algebra}
    \begin{split}
    \{W_e(\ell_h), W_m(\ell_v) \} & = \{W_e(\ell_v), W_m(\ell_h) \} = 0, \\
    [W_e(\ell_h), W_e(\ell_v)] & = [W_m(\ell_h), W_m(\ell_v)] = 0, \\
    [W_e(\ell_h), W_m(\ell_h)] & = [W_e(\ell_v), W_m(\ell_v)] = 0.
    \end{split}
\end{align}
These are equivalent to the algebraic relations followed by Pauli-$X$ and $Z$ logical operators on two qubits. Indeed, one possible assignment is 
\begin{align}
    \label{eq:Pauli-Z_X_assignment}
    \begin{split}
    X_1 \defeq W_e(\ell_h)&, X_2 \defeq W_e(\ell_v), \\
    Z_1 \defeq W_m(\ell_v)&, Z_2 \defeq W_m(\ell_h). 
    \end{split}
\end{align}

Up to this point, the discussion did not require the one-form anomalous symmetries (over contractible loops) to be strong symmetries of a particular state. If this happens for both symmetries with nontrivial-braiding -- which we shall term strong-strong anomaly -- then the non-contractible loop operators are deformable. This further implies that the logical states are locally indistinguishable. The argument goes as follows: suppose $\ket{\psi_{00}}$ is a toric code state satisfying $Z_1 \ket{\psi_{00}}= Z_2 \ket{\psi_{00}} = \ket{\psi_{00}}$. By acting on $\ket{\psi_{00}}$ with $X_1$ and $X_2$, we can get the three other eigenstates of $Z_1$ and $Z_2$, such as $\ket{\psi_{10}} \defeq X_1 \ket{\psi_{00}}$. To see that $\ket{\psi_{10}}$ and $\ket{\psi_{00}}$ are locally indistinguishable, before taking the reduced density matrix of $\ket{\psi_{10}}$ over a small region $A$, we deform $X_1$ to not pass through $A$. In this way, 
\begin{align}
\Tr_{A^c}\ketbra{\psi_{10}}{\psi_{10}} & = \Tr_{A^c}[X_1 \ketbra{\psi_{00}}{\psi_{00}} X_1^\dagger] \label{eq:CM_LI_proof_1} \\
& = \Tr_{A^c}\ketbra{\psi_{00}}{\psi_{00}}. \label{eq:CM_LI_proof_2} 
\end{align}

The same is true for all other eigenstates constructed from $\ket{\psi_{00}}$ of any other logical unitary, and also for mixtures of those. For example, $\ket{\phi_{++}} = \frac{1}{2} \sum_{a,b=0}^1 \ket{\psi_{ab}}$, where $\ket{\psi_{ab}} = X_1^a X_2^b \ket{\psi_{00}}$, is an eigenvector of $X_1$ and $X_2$ with eigenvalue $+1$. Since $X_1 = W_e(\ell_h)$ and $X_2 = W_e(\ell_v)$ are deformable due to the strong symmetry of the $m$ anyon, then $\ket{\phi_{++}}$ is locally indistinguishable from the other eigenstates of $X_1$ and $X_2$. Furthermore, $\ket{\phi_{++}}$ is locally indistinguishable to $\ket{\psi_{00}}$ due to the strong symmetry of the $e$ anyon:
\begin{align}
    \Tr_{A^c}\ketbra{\phi_{++}}{\phi_{++}} & = \frac{1}{4} \sum_{\substack{a,b,\\a',b'}} \Tr_{A^c} [X_1^a X_2^b \ketbra{\psi_{00}}{\psi_{00}} X_2^{b'} X_1^{a'}] \\
    & = \frac{1}{4} \sum_{a,b} \Tr_{A^c} [X_1^a X_2^b \ketbra{\psi_{00}}{\psi_{00}} X_2^{b} X_1^{a}] \\
    & = \Tr_{A^c} \ketbra{\psi_{00}}{\psi_{00}}
\end{align}
where we have used that $X_1$ and $X_2$ can be deformed to have support outside $A$, and that the off-diagonal terms in the sum are zero due to the strong $Z_1$ or $Z_2$ symmetries (supported outside $A$ as well).

Having argued for locally indistinguishability, we now turn out attention to quantum memory. Famously, logical quantum information encoded in the toric code can be decoded back after the application of sufficiently weak and local noise channels \cite{dennis_topological_2002}. Because of this, we say it forms a topological quantum memory, and the traditional point of view is through error correction. Here, we present another interpretation based on strong anomalous symmetries of even the noisy states, if they are in the same LRE phase. In this interpretation, the logical information can be decoded because there exist strong symmetries coming from symmetry pullback, which form a basis for the space of all logical operators and whose expectation values are the same as the ones for the clean system. Thus, in principle, the two-qubit quantum state can be read out by measuring their expectation values.

Even though the arguments were presented just for the anomaly of the toric code, we claim that it generalizes to any abelian topological order. The algebraic relations between non-contractible loop operators, such the ones in Eq. \eqref{eq:non_contractible_algebra}, will change due to different braiding phases, but the fact that they form a representation for a logical Hilbert space of locally indistinguishable states that preserve information in the same phase of matter remains.

\subsubsection{Classical memory and strong-weak anomaly}
\label{sec:classical_memory}

If one of the one-form symmetries is weak, how does the picture above change? Here, we argue that there is only a topological \textit{classical} memory, instead of a quantum memory. More precisely, only a discrete set of mutually orthogonal states (and their mixtures) are guaranteed to be locally indistinguishable, instead of an entire subspace. Moreover, we show that, because weak symmetries cannot in general be pulled back via SLCs, then having a classical memory is also not a property of the whole long-range entanglement phase of matter. At the same time, we briefly discuss how the phase invariance property of the memory can be recovered under a more refined notion of phase equivalence.

For concreteness, let us focus on the ``loop soup'' state $\rho_{\text{loop}} \propto \sum_{\text{loops }\ell} \ketbra{\ell}{\ell}$, where the sum is over contractible loops only (See Eq. \ref{eq:loop_soup_state}). The magnetic one-form symmetry operators $W_m(\gamma) = \prod_{i \in \gamma} Z_i$ over contractible and non-contractible loops $\gamma$ in the dual lattice form the strong symmetry group of $\rho_{\text{loop}}$, whereas $W_e(\tau)$ is only a weak symmetry (for contractible loops $\tau$).

Similarly to the previous section, we can define states with different strong symmetry charges of $Z_1 = W_m(\ell_v)$ and $Z_2 = W_m(\ell_h)$ by acting with $X_1 = W_e(\ell_h)$ and $X_2=W_e(\ell_v)$. For example, $\rho_{\text{loop}}^{01} \defeq X_2 \rho_{\text{loop}} X_2$ has $Z_2 \rho_{\text{loop}}^{01} = - \rho_{\text{loop}}^{01}$. The proof that $\rho_{\text{loop}}^{(10)}$ and $\rho_{\text{loop}}^{(00)} = \rho_{\text{loop}}$ are locally indistinguishable is the same as in Eqs. \eqref{eq:CM_LI_proof_1} and \eqref{eq:CM_LI_proof_2}, since the action of $X_2$ by conjugation is deformable due to the weak symmetry of contractible $W_e$ loop operators.

The difference to strong-strong anomaly is that the states $\rho_{\text{loop}}^{(ab)} = X_1^a X_2^b \rho_{\text{loop}} X_2^b X_1^b$, $a,b \in \{0,1\}$, and their convex combinations are the only states locally indistinguishable to $\rho_{\text{loop}}$. For example, if one tries to construct an eigenstate of $X_1$ and $X_2$ by projecting $\rho_{\text{loop}}$ onto the symmetric subspace $X_1 = X_2 = +1$:
\begin{align}
    \rho_{\text{loop}}^{(++)} \propto (\one + X_1) (\one + X_2) \rho_{\text{loop}} (\one + X_2) (\one + X_1),
\end{align}
one does not arrive at a locally indistinguishable state, since the non-contractible curves of the $X_1$ and $X_2$ projectors above cannot be deformed.

Also unlike the strong-strong anomaly case, the classical memory above is not invariant throughout the long-range entanglement phase. Indeed, the loop soup state is in the trivial phase, and thus two-way connected to, e.g., the product state $\ket{00\cdots 0}$, which has no other state locally indistinguishable to it. 

This can be explained as the failure of the weak $W_e$ symmetry to be pulled back from $\rho_{\text{loop}}$ to $\ket{00\cdots0}$. In a future work \cite{sang_mixedstate_2025},
it will be shown that the weak symmetries survive under restricted class of channels in a generalized sense. Namely, they call also be ``pulled back'' under Lindbladian evolutions that maintain a finite Markov length. Crucially, the braiding between strong and weak symmetries will also be shown to be invariant under this generalized symmetry pullback. One of the consequences is that the phase diagram of the toric code under Pauli-$X$ and $Z$ dephasing (See Fig. \ref{fig:phase_diagram}) will have a finer characterization, as the phase of both $\rho_{0,1/2} = \rho_{\text{loop}}$ and $\rho_{1/2,0}$ will separate from the totally trivial phase of $\rho_{1/2,1/2}=\one$. 

\section{Generalization to higher-form symmetries in higher dimensions}
\label{sec:higher-form_higher-dim}

In Sec. \ref{sec:long-range_bip_ent}, we saw how the anomaly of one-form symmetries in (2+1)-D systems constrains the strongly symmetric states to be long-range bipartite entangled. In a previous work \cite{lessa_mixedstate_2024}, a similar result was found: the anomaly of ordinary $0$-form symmetries in $(d+1)$-D systems imply long-range $(d+2)$-partite entanglement. The goal of this section is to unify and generalize these results. For concreteness, we consider a class of anomaly in bosonic systems, with $\Z_{n_1}^{(p_1)}\times\Z_{n_2}^{(p_2)}\times \cdots \times \Z_{n_M}^{(p_M)}$ symmetry, where $p_i\geq0$ labels the form of the $i$'th symmetry. The 't Hooft anomaly can be characterized as a response term in the partition function $Z=e^{iS[\{a_i\}]}$, where $\{a_i\}$ represents the background gauge field of the symmetry that lives in a $(d+1)+1$ dimensional bulk -- the anomalous theory can be viewed as the boundary of this symmetry-protected topological (SPT) bulk state. In particular, we consider the following type of bulk term:
    \begin{equation}
    \label{eq:cup}
        S = \Theta\int_{\mathcal{M}_{d+2}} a_1^{(p_1)} \cup a_2^{(p_2)} \cup \cdots \cup a_k^{(p_k)},
    \end{equation}
where $a_i^{(p_i)}\in H^{(p_i+1)}(\mathcal{M}_{d+2},\Z_{n_i})$ is a $(p_i+1)$-cocycle corresponding to a $p_i$-form symmetry, and $k = d + 2 - \sum_i p_i \geq 2$. The cup product $\cup$ is a natural product operation for discrete cocycles, analagous to the role of wedge product $\wedge$ for differential forms. $\Theta$ takes a discrete value compatible with the quantization of all the gauge fields involved. The term Eq.~\eqref{eq:cup} describes a large class (but not all) of bosonic anomaly. We shall call such anomaly ``of order $k$'' if $k\geq2$ gauge fields are involved in the response term. For zero-form symmetries, $k=d+2$. For ordinary topological orders whose excitations include gauge charge and gauge flux, $k=2$. As an example, the $2d$ toric code has a one-form $\Z_2^e$ symmetry and a one-form $\Z_2^m$ symmetry. So the response term is simply $S=\pi\int_{\mathcal{M}_4}a^{(1)}_e\cup a^{(1)}_m$.
Our main statement in this Section is the following conjecture:
\begin{conjecture}\label{conj:higher-dim_higher-form}
    Consider a $(d+1)$-D phase of matter with an anomaly of order $k$ described by Eq.~\eqref{eq:cup}. Then, any strongly symmetric mixed state in this phase has long-range $k$-partite entanglement.
\end{conjecture}
The above conjecture naturally unifies the main results in this work ($k=2$) and Ref.~\cite{lessa_mixedstate_2024} ($k=d+2$). For many (but not all) anomalies of the form Eq.~\eqref{eq:cup}, the argument for this conjecture proceeds by dimension reduction as follows. First, arrange $p_i$ from largest to smallest, $p_1 \geq p_2 \geq \cdots \geq p_k$, and consider a $k$-partition $\mathcal{P}$. Then, restrict the action $S$ to a submanifold $\mathcal{N}$ of dimension $d + 2 - p_1$ intersecting all regions of the partition $\mathcal{P}$. On $\mathcal{N}$, the first higher-form symmetry can be made to act globally, i.e. as a 0-form symmetry, and the action reduces to
\begin{equation}
    S_{\mathcal{N}} = \Theta \int_{\mathcal{N}} a_1^{(0)} \cup a_2^{(p_2)} \cup \cdots \cup a_k^{(p_k)},
\end{equation}
where we have used that the other $p$-form symmetries on $\mathcal{M}_{d+2}$ (i.e., $a_2^{(p_2)}, \ldots, a_k^{(p_k)}$) are still $p$-form for generic intersections between the $(d-p)$-dimensional submanifolds where they act and $\mathcal{N}$. 

Continue this procedure until every symmetry is reduced to zero-form, in which case the spacetime effective action of the gauged SPT bulk is reduced to a manifold of dimension $d_{\text{eff}} + 2$, where $d_{\text{eff}} = d - \sum_i p_i$, that is still $k$-partioned. Now, we can use the result of multipartite entanglement of zero-form anomaly treated in \cite{lessa_mixedstate_2024}. There, it was proven for $d_{\text{eff}} \leq 1$ and argued for $d_{\text{eff}} > 1$ that any strongly symmetric state is $k$-partite entangled for certain $k$-partitions, with $k = d_{\text{eff}} + 2 = d + 2 - \sum_i p_i$. 

One of the difficulties of making this argument more precise is to be able to reduce not only the effective bulk action, but also the symmetry operators and their symmetric states to lower-dimensional submanifolds, while maintaining the strong symmetry and its anomaly. In \cite{lessa_mixedstate_2024}, a similar dimension reduction procedure was achieved in (1+1)-D by utilizing technical results regarding the restriction of LUs. In higher dimensions, the same reduction was accomplished at the cost of restricting the class of symmetry actions. In both cases, there are also subtle dependencies on the partitions, which further complicates the general picture above. Another complication comes from the fact that some nontrivial anomalies become trivial upon dimensional reduction -- for example, the fermionic $1$-form anomaly $\pi\int_{\mathcal{M}_4}a^{(1)}\cup a^{(1)}$, which is nontrivial, becomes $\pi\int_{\mathcal{M}_2}a^{(0)}\cup a^{(0)}$, which is trivial, upon dimension reduction. This means that we need different arguments for such anomaly -- in fact, this is exactly why we need a separate argument (based on exchange instead of braiding) for the fermionic $1$-form symmetry in Sec.~\ref{sec:long-range_bip_ent}. Nevertheless, it is plausible that Conjecture \ref{conj:higher-dim_higher-form} holds for many cases of interest.

Apart from the scenario of $p_i = 1$ in (2+1)-D, treated in the present work in detail, other anomaly patterns can be illustrated via explicit models. For example, conventional SSB in $(d+1)$-D can be viewed as a mutual anomaly between a 0-form symmetry $g\in G \mapsto U_g$ (e.g. $\prod_i X_i$ for $G=\Z_2$) and a $d$-form symmetry $\mathcal{O}_i \mathcal{O}_j$ (e.g. $Z_i Z_j$)\footnote{A $d$-form symmetry can be viewed as a map from $S^0 = \{-1, +1\}$ to a pair of (unitary) operators $(\mathcal{O}_i, \mathcal{O}_j)$ at sites $i$ and $j$.}, in the case where the two-point correlation function is saturated at $\langle \mathcal{O}_i \mathcal{O}_j \rangle = 1$. In this case, the mixed anomaly is the fact that a single $\mathcal{O}_i$ is a charged operator under $U$, and the long-range bipartite entanglement (as $k = 2$) comes from the long-range order of $\langle \mathcal{O}_i \mathcal{O}_j\rangle - \langle \mathcal{O}_i\rangle \langle \mathcal{O}_j \rangle = 1$. In Sec. \ref{sec:set} below, we describe another example in (2+1)-D of a mixed anomaly between two $\Z_2$ 0-form symmetries and one $\Z_2$ 1-form symmetry. We explicitly show via a dimension reduction argument that its strongly symmetric states are long-range tripartite entangled, thus providing further evidence for the conjecture above.

\subsection{\texorpdfstring{$\mathbb{Z}_2^{(0)} \times \mathbb{Z}_2^{(0)} \times \mathbb{Z}_2^{(1)}$ anomaly in $(2+1)$-D with long-range tripartite entanglement.}{Z2 0-form times Z2 0-form times Z2 1-form anomaly in (2+1)-D with long-range tripartite entanglement.}}\label{sec:set}

Here, we present a mixed anomalous phase between 0-form and 1-form symmetries in $(2+1)$-D that satisfies Conjecture \ref{conj:higher-dim_higher-form} explicitly. It has a mutual anomaly between two $\Z_2$ 0-form symmetries and one $\Z_2$ 1-form symmetry, and can be represented in a bipartite lattice $\Lambda = \Lambda_\circ \sqcup \Lambda_\bullet$ as follows: one $0$-form $\Z_2$ symmetry acts as the spin flip operator $X_{\circ} \defeq \prod_{i \in \Lambda_\circ} X_i$ on one sublattice, while the other $\Z_2$ acts on the other sublattice $X_{\bullet} \defeq \prod_{i \in \Lambda_\bullet} X_i$, and the 1-form $\Z_2$ string operators are generated by control-Z gates on neighboring sites, as $W_{CZ}(\gamma) \defeq \prod_{\langle i j \rangle \in \gamma} CZ_{ij}$. Thus, the effective response action of the (3+1)-D SPT bulk is $S = \pi \int_{\mathcal{M}_4} a_\circ^{(0)} \cup a_{\bullet}^{(0)} \cup a_{CZ}^{(1)}$.

\subsubsection{Proof of long-range tripartite entanglement}

\begin{figure}[t]
    \centering
    \def\svgwidth{0.9\linewidth}
    {\large 
\begingroup%
  \makeatletter%
  \providecommand\color[2][]{%
    \errmessage{(Inkscape) Color is used for the text in Inkscape, but the package 'color.sty' is not loaded}%
    \renewcommand\color[2][]{}%
  }%
  \providecommand\transparent[1]{%
    \errmessage{(Inkscape) Transparency is used (non-zero) for the text in Inkscape, but the package 'transparent.sty' is not loaded}%
    \renewcommand\transparent[1]{}%
  }%
  \providecommand\rotatebox[2]{#2}%
  \newcommand*\fsize{\dimexpr\f@size pt\relax}%
  \newcommand*\lineheight[1]{\fontsize{\fsize}{#1\fsize}\selectfont}%
  \ifx\svgwidth\undefined%
    \setlength{\unitlength}{270.84216261bp}%
    \ifx\svgscale\undefined%
      \relax%
    \else%
      \setlength{\unitlength}{\unitlength * \real{\svgscale}}%
    \fi%
  \else%
    \setlength{\unitlength}{\svgwidth}%
  \fi%
  \global\let\svgwidth\undefined%
  \global\let\svgscale\undefined%
  \makeatother%
  \begin{picture}(1,0.97689896)%
    \lineheight{1}%
    \setlength\tabcolsep{0pt}%
    \put(0,0){\includegraphics[width=\unitlength,page=1]{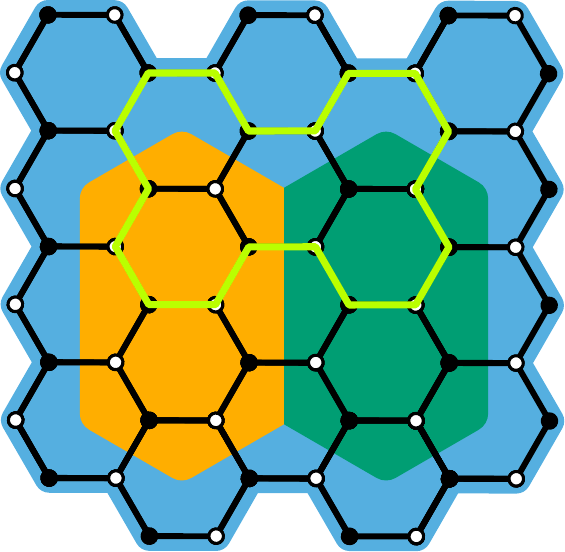}}%
    \put(0.5006207,0.80875421){\color[rgb]{0.72941176,1,0}\makebox(0,0)[t]{\lineheight{1.25}\smash{\begin{tabular}[t]{c}$W_{CZ}(\ell)$\end{tabular}}}}%
    \put(0.32425974,0.31459158){\makebox(0,0)[t]{\lineheight{1.25}\smash{\begin{tabular}[t]{c}\textbf{B}\end{tabular}}}}%
    \put(0.67787779,0.31489293){\makebox(0,0)[t]{\lineheight{1.25}\smash{\begin{tabular}[t]{c}\textbf{C}\end{tabular}}}}%
    \put(0.14425365,0.82756166){\makebox(0,0)[t]{\lineheight{1.25}\smash{\begin{tabular}[t]{c}\textbf{A}\end{tabular}}}}%
    \put(0.43972573,0.98964862){\makebox(0,0)[t]{\lineheight{1.25}\smash{\begin{tabular}[t]{c}$X_\bullet$\end{tabular}}}}%
    \put(0.55800777,0.98964862){\makebox(0,0)[t]{\lineheight{1.25}\smash{\begin{tabular}[t]{c}$X_\circ$\end{tabular}}}}%
  \end{picture}%
\endgroup%
}
    \caption{Illustration for the proof of tripartite entanglement for the $\mathbb{Z}_2^{(0)} \times \mathbb{Z}_2^{(0)} \times \mathbb{Z}_2^{(1)}$ described in Sec. \ref{sec:set}. There, the space is partitioned into regions $ABC$, and the $\Z_2^{(1)}$ one-form symmetry $W_{CZ}(\ell)$ intersects all of them. The mutual anomaly with the two $\Z_2^{(0)}$ symmetries, $X_{\circ}$ and $X_\bullet$, which act in the white ($\circ$) and black ($\bullet$) sublattices, respectively, can be seen from their commutation relations when restricting $W_{CZ}(\ell)$ to $BC$. }
    \label{fig:SET_trip-ent_proof}
\end{figure}
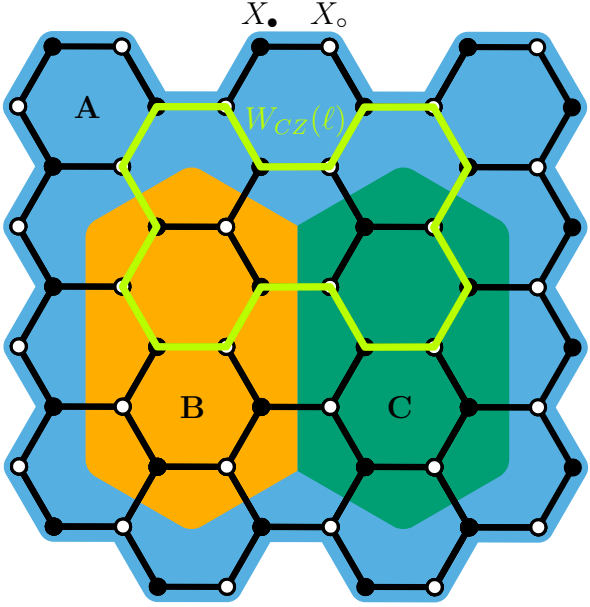

Using that $CZ_{ij} X_i CZ_{ij}= X_iZ_j$, we can check that $X_{\circ}$ and $X_{\bullet}$ commute with $W_{CZ}(\gamma)$ for any loop $\gamma$. However, when $\gamma$ is a open curve, $W_{CZ}(\gamma)$ fails to commute with both $X_{\circ}$ and $X_\bullet$ at the boundary of $\gamma$:
\begin{align}
    X_{\circ} W_{CZ}(\gamma) X_\circ & = Z^{\bullet}_{\partial \gamma} W_{CZ}(\gamma) \label{eq:anomalous_algebra_circ} \\
    X_{\bullet} W_{CZ}(\gamma) X_\bullet & = Z^{\circ}_{\partial \gamma} W_{CZ}(\gamma).
\end{align}
Importantly, the boundary operator $Z^{\circ}_{\partial \gamma}$ ($Z^{\bullet}_{\partial \gamma}$) is a product of Pauli-$Z$ gates near each endpoint of $\gamma$ that anticommute with the $\Z_2$ 0-form symmetry $X_\circ$ (with $X_\bullet$) of the opposite sublattice. This local obstruction to commutation at the boundary of restricted symmetries is the manifestation of the mutual anomaly between the three symmetries.

As predicted by Conjecture \ref{conj:higher-dim_higher-form}, any strongly symmetric state $\rho$ under this anomalous symmetry has to be $k=3$-partite entangled. We will prove this by following the dimension reduction argument presented above. First, by contradiction, we assume a strongly symmetric tripartite separable state $\rho$ under a tripartition $A|B|C$ depicted in Fig. \ref{fig:SET_trip-ent_proof}. By strong symmetry inheritance, we can assume $\rho$ is a pure state equal to $\ket{A}\ket{B}\ket{C}$. Then, we restrict it to a loop $\ell$ intersecting all three regions, resulting in an effective one-dimensional state $\rho_{\ell} = \rho_\ell^{A} \otimes \rho_\ell^B \otimes \rho_\ell^C$ that has no correlations between $A$, $B$ and $C$. Since $W_{CZ}(\ell)$ acts only on $\ell$, it remains a strong (0-form) symmetry of $\rho_\ell$, but $X_\circ$ and $X_\bullet$ turn into weak symmetries that act, say, on even and odd sites. Importantly though, the mixed anomaly between the strong and weak symmetries implies that $\rho_\ell$ cannot be tripartite uncorrelated. Otherwise, the $X_\text{even} \defeq X_\circ |_\ell$ symmetry restricted to, say, $BC$ would still be a weak symmetry of $\rho_\ell^B \otimes \rho_\ell^C$. This, together with the mixed anomaly with the $W_{CZ}(\ell)$ strong symmetry, would imply 
\begin{align}
    1 & = \Tr[\rho_\ell X_{\circ}|_{BC} W_{CZ}(\ell) X_{\circ}|_{BC} W_{CZ}(\ell)] \\
    & = \Tr[\rho_\ell Z_b Z_c] \\
    & = \langle Z_b Z_c \rangle_c,
\end{align}
where in the last line we used that $\langle Z_i \rangle = 0$ due to the weak $X_{\text{even}}$ and $X_{\text{odd}}$ symmetries, and $b$ and $c$ are \emph{odd} sites near the boundaries of $B$ and $C$ with $A$, respectively. The nontrivial connected correlation function $\langle Z_bZ_c \rangle_c = 1$ is in contradiction with the absence of correlations between regions $A$, $B$ and $C$ \footnote{It may happen that one particular choice of which on-site symmetry ($X_{\text{even}}$ or $X_{\text{odd}}$) to restrict and to which region leaves endpoints $b$ and $c$ inside a single region. If that happens, a different choice will separate them, and thus our conclusion remains.}, QED.

Finally, to argue long-range entanglement, we use that the anomaly remains invariant under symmetry pullback. For one-form anomalies in 2d, this means braiding and self-statistics, and here, it is the fact that $X_\circ W_{CZ}(\gamma)X_{\circ} W_{CZ}(\gamma)$ is the product of operators at $\partial\gamma$ that are individually odd under $X_\bullet$, and vice-versa by exchanging the sublattices. Then, following a similar argument as above, we conclude that $\rho$ exhibits long-range tripartite entanglement\footnote{For the tripartite entanglement argument, we used the fact that $X_\circ$ and $X_\bullet$ turn into weak symmetries of the reduced density matrix on $\ell$. This can fail for the pulled-back symmetries, but we can still reach the same conclusion but working with the 2d regions $A,B$ and $C$ directly, without reducing to the $1d$ loop $\ell$.}.

\subsubsection{Bipartite separability and other symmetric states}

Even though we showed that the mixed anomaly between $X_\circ$, $X_\bullet$ and $W_{CZ}(\ell)$ implies \textit{tripartite} entanglement, one may ask whether \textit{bipartite} entanglement could also be guaranteed by the anomly. If so, the prediction from Conjecture \ref{conj:higher-dim_higher-form} for these particular symmetries would not be tight. In this section, we answer the above negatively by constructing symmetric pure states separable under a restricted set of bipartitions. Afterwards, we present other pertinent examples.

Specifically, if $A$ is a connected region whose edges to its complement $B = A^c$ are only through boundary sites of the same sublattice (so either $\partial A \subset \Lambda_\circ$ or $\partial A \subset \Lambda_\bullet$), then the following bipartite separable state is symmetric:
\begin{align}
\label{eq:SET_2-sep_state}
    \begin{split}
    \ket{\psi_{\text{2-sep}}} \propto 
        & (\one + \lambda_\circ X_{\circ,A})(\one + \lambda_\bullet X_{\bullet,A}) \ket{\alpha; \alpha_{\partial A} = \mu_A}_A \otimes \\
        & (\one + \lambda_\circ X_{\circ,B})(\one + \lambda_\bullet X_{\bullet,B}) \ket{\beta; \beta_{\partial B} = \mu_B}_B,
    \end{split}
\end{align}
where $\lambda_{\circ,\bullet}^2 = 1$, and $\ket{\alpha; \alpha_{\partial A} = \mu_A}_A$ is any $Z$-basis product state in $A$ having all boundary qubits set to $\ket{\mu_A}$, $\mu_A \in \{0,1\}$ and $W_{CZ}(\ell) = +1$ for all loops $\ell$ in $A$ (similarly for $\ket{\beta; \beta_{\partial B} = \mu_B}_B$). One example of an allowed region $A$ (or $B$) is a triangular region with sides inclined at $0$, $120$ and $240$ degrees with respect to the vertical direction (See Fig. \ref{fig:SET_bip_sep}). 

\begin{figure}[t]
    \centering
    \def\svgwidth{0.9\linewidth}
    {\large 
\begingroup%
  \makeatletter%
  \providecommand\color[2][]{%
    \errmessage{(Inkscape) Color is used for the text in Inkscape, but the package 'color.sty' is not loaded}%
    \renewcommand\color[2][]{}%
  }%
  \providecommand\transparent[1]{%
    \errmessage{(Inkscape) Transparency is used (non-zero) for the text in Inkscape, but the package 'transparent.sty' is not loaded}%
    \renewcommand\transparent[1]{}%
  }%
  \providecommand\rotatebox[2]{#2}%
  \newcommand*\fsize{\dimexpr\f@size pt\relax}%
  \newcommand*\lineheight[1]{\fontsize{\fsize}{#1\fsize}\selectfont}%
  \ifx\svgwidth\undefined%
    \setlength{\unitlength}{270.84216261bp}%
    \ifx\svgscale\undefined%
      \relax%
    \else%
      \setlength{\unitlength}{\unitlength * \real{\svgscale}}%
    \fi%
  \else%
    \setlength{\unitlength}{\svgwidth}%
  \fi%
  \global\let\svgwidth\undefined%
  \global\let\svgscale\undefined%
  \makeatother%
  \begin{picture}(1,0.97689896)%
    \lineheight{1}%
    \setlength\tabcolsep{0pt}%
    \put(0,0){\includegraphics[width=\unitlength,page=1]{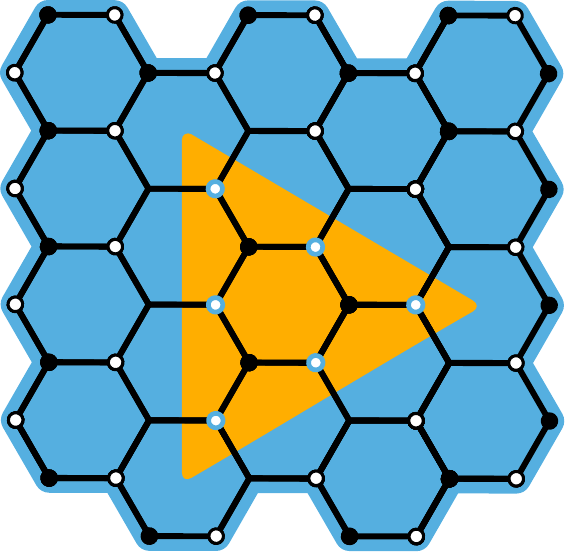}}%
    \put(0.50379048,0.41709399){\makebox(0,0)[t]{\lineheight{1.25}\smash{\begin{tabular}[t]{c}\textbf{B}\end{tabular}}}}%
    \put(0.14425365,0.82756166){\makebox(0,0)[t]{\lineheight{1.25}\smash{\begin{tabular}[t]{c}\textbf{A}\end{tabular}}}}%
    \put(0,0){\includegraphics[width=\unitlength,page=2]{SET_bip_sep.pdf}}%
  \end{picture}%
\endgroup%
}
    \caption{Configuration of a triangular bipartition $A|B$ for which there are symmetric and bipartite separable states $\ket{\psi_{\text{2-sep}}}$ (See Eq.~\ref{eq:SET_2-sep_state}). The boundary qubits of each region, constrained to the $\ket{0}$ state, are highlighted by an outline with the color of the opposite region.}
    \label{fig:SET_bip_sep}
\end{figure}

The state $\ket{\psi_{\text{2-sep}}}$ is symmetric under $X_\circ$ and $X_\bullet$ by construction, and under $W_{CZ}(\ell)$ because it passes through the restricted $\Z_2^{(0)}$ operators (say, $X_{\circ, A}$) as follows:
\begin{equation}
    W_{CZ}(\ell) X_{\circ, A} = Z^{\bullet}_{\ell \cap \partial B} X_{\circ, A} W_{CZ}(\ell),
\end{equation}
where we assumed $\partial A \subset \Lambda_\circ$ (the other case is analogous).
Since $Z^{\bullet}_{\ell \cap \partial B}$ has an even number of $Z$ operators in the same sublattice $\Lambda_\bullet$, then it commutes with both $X_{\bullet,B}$ and $X_{\circ, B}$ and stabilizes $\ket{\beta; \beta_{\partial B} = \mu_B}$.

The existence bipartite-separable but tripartite-entangled anomalous states occurs for the 0-form anomaly in (1+1)-D as well \cite{lessa_mixedstate_2024}. However, a difference between the $\mathbb{Z}_2^{(0)} \times \mathbb{Z}_2^{(0)} \times \mathbb{Z}_2^{(1)}$ anomaly in (2+1)-D and, for example, the $\mathbb{Z}_2^{(0)}$ CZX anomaly in (1+1)-D is that the latter exhibits an orthonormal basis of bipartite separable states, while the former does not, due to the extensive number of constraints coming from the one-form symmetry condition $\forall \ell, W_{CZ}(\ell) = +1$. Indeed, if there were such a basis, the maximally mixed symmetric state\footnote{The subscript ``$\infty$'' in the maximally mixed symmetric state $\rho_{\infty}$ comes from viewing it as the infinite-temperature state in the canonical ensemble.}
\begin{equation}
\label{eq:MMS_Z2_cubed}
\rho_{\infty} \propto (\one + X_\circ) (\one + X_\bullet) \prod_{\hexagon} (\one + W_{CZ}(\hexagon)),
\end{equation}
being the mixture of all symmetric states, would be bipartite separable, at least under the triangular bipartitions. However, we show in Appendix \ref{appendix:ent_neg_MMS} that $\rho_\infty$ is entangled by calculating its negativity.

The final example of a symmetric state under the $\mathbb{Z}_2^{(0)} \times \mathbb{Z}_2^{(0)} \times \mathbb{Z}_2^{(1)}$ symmetry comes from a recent work \cite{zhang_longrange_2024} that studied the properties of the pure state
\begin{equation}
    \ket{\psi} \propto \prod_{\hexagon}(\one + W_{CZ}(\hexagon)) \ket{+}.
\end{equation}
It was named ``SPT soup'' as it can also be written as the superposition of 1d cluster states (with $\mathbb{Z}_2 \times \mathbb{Z}_2$ SPT orders) over all loops in a sea of $\ket{+}$ states. From the discussion above, the anomalous symmetries of $\ket{\psi}$ imply it is long-range tripartite entangled. In fact, an even stronger statement can be made for $\ket{\psi}$: \cite{zhang_longrange_2024} showed that it has $Z_iZ_j$ two-point correlation function decaying algebraic as $\langle \psi |Z_i Z_j| \psi\rangle \sim 1/|i-j|$, if $i$ and $j$ are in the same sublattice.

\section{Outlook}

We end with some discussions on open questions and future directions:

\begin{enumerate}
    \item In Sec. \ref{sec:TO_mixed-state_phases}, we propose a new definition of (long-range entanglement) phase of matter based on stochastic local channels. Compared to the more widely used definition using local channels, SLCs can create long-range correlation. That is why, for example, long-range correlated classical states become trivial under the new definition. However, we are not aware of a \textit{topologically ordered} phase that sensitively depends on whether one chooses SLCs or LCs as the set of free operations. More broadly, is there a (physically relevant) nontrivial phase in the LC definition that becomes trivial in the SLC definition? Moreover, can two nontrivial phases under LCs become one phase under SLCs?

    \item Here, we focused on one-form symmetries of topologically ordered systems. However, the formalism of higher-form symmetries is widely applicable to other systems, such as symmetry-protected topological phases and even gauge theories, such as Maxwell's electromagnetism \cite{gaiotto_generalized_2014, mcgreevy_generalized_2023}. Going further into generalized symmetries, non-invertible symmetries such as the Kramers-Wannier duality transformation are of relevance as well \cite{shao_whats_2024}. Such non-invertible symmetries are also relevant for more general (non-Abelian) topological orders. What can we say about the mixed states with such symmetries? 
    
    \item We also focused on entanglement features induced by higher-form symmetries, which possess a rich theory. However, estimating them either numerically or experimentally remains challenging. For example, the topological entanglement of formation (Def. \ref{def:TEF}) is not only nonlinear in the density matrix, but also involves a minimization over decompositions of the mixed state, which are hard tasks even on a computer. On the other hand, the more accessible linear expectation values of local observables suffer from behaving smoothly even when passing through a mixed-state phase transition \cite{fan_diagnostics_2024}. A middle ground is to consider correlation functions of string operators that are only quadratic in the density matrix. This approach was taken in \cite{zhang_strongweak_2024} by studying measures of strong-to-weak SSB, which distinguish different mixed-state phases.

    \item As discussed in Sec. \ref{sec:zx_dephasing}, the TEF of the fermionic imTO phase could not be determined, since $W_f$ is the unique strong symmetry, and it doe not braid non-trivially with itself. However, by restricting to decompositions of the imTO state into pure states described by modular TQFT, we can lower bound this modified TEF by $\log 2$. Is the same valid for the TEF over all decompositions?
    
    \item Related to the above, are there topologically ordered mixed states with a strong one-form symmetry that does not braid with any other symmetry, weak \textit{and} strong? For example, even though the imTO phase of the ZX-dephased TC has a strong fermionic one-form symmetry that braids with the weak $e$ and $m$ anyons, there is no reason to expect the same behavior for the entire LRE phase, since only the strong symmetry is necessarily pulled back. On the other hand, we are not aware of any explicit example for the question above.
\end{enumerate}

\begin{acknowledgments}
We thank Tyler Ellison, Subhayan Sahu, Tarun Grover for illuminating discussions. L.A.L. acknowledges support from the Natural Sciences and Engineering Research Council of Canada (NSERC) under Discovery Grant No. RGPIN-2020-04688 and No. RGPIN-2018-04380.  This work was also supported by an Ontario Early Researcher Award.  T.-C.L. acknowledges the support of the RQS postdoctoral fellowship through the National Science Foundation (QLCI grant OMA-2120757). S.S. was supported by the SITP postdoctoral fellowship at Stanford University. Research at Perimeter Institute is supported in part by the Government of Canada through the Department of Innovation, Science and Industry Canada and by the Province of Ontario through the Ministry of Colleges and Universities.

\end{acknowledgments}

\bibliography{bibliography}

\begin{thebibliography}{78}%
\makeatletter
\providecommand \@ifxundefined [1]{%
 \@ifx{#1\undefined}
}%
\providecommand \@ifnum [1]{%
 \ifnum #1\expandafter \@firstoftwo
 \else \expandafter \@secondoftwo
 \fi
}%
\providecommand \@ifx [1]{%
 \ifx #1\expandafter \@firstoftwo
 \else \expandafter \@secondoftwo
 \fi
}%
\providecommand \natexlab [1]{#1}%
\providecommand \enquote  [1]{``#1''}%
\providecommand \bibnamefont  [1]{#1}%
\providecommand \bibfnamefont [1]{#1}%
\providecommand \citenamefont [1]{#1}%
\providecommand \href@noop [0]{\@secondoftwo}%
\providecommand \href [0]{\begingroup \@sanitize@url \@href}%
\providecommand \@href[1]{\@@startlink{#1}\@@href}%
\providecommand \@@href[1]{\endgroup#1\@@endlink}%
\providecommand \@sanitize@url [0]{\catcode `\\12\catcode `\$12\catcode `\&12\catcode `\#12\catcode `\^12\catcode `\_12\catcode `\%12\relax}%
\providecommand \@@startlink[1]{}%
\providecommand \@@endlink[0]{}%
\providecommand \url  [0]{\begingroup\@sanitize@url \@url }%
\providecommand \@url [1]{\endgroup\@href {#1}{\urlprefix }}%
\providecommand \urlprefix  [0]{URL }%
\providecommand \Eprint [0]{\href }%
\providecommand \doibase [0]{https://doi.org/}%
\providecommand \selectlanguage [0]{\@gobble}%
\providecommand \bibinfo  [0]{\@secondoftwo}%
\providecommand \bibfield  [0]{\@secondoftwo}%
\providecommand \translation [1]{[#1]}%
\providecommand \BibitemOpen [0]{}%
\providecommand \bibitemStop [0]{}%
\providecommand \bibitemNoStop [0]{.\EOS\space}%
\providecommand \EOS [0]{\spacefactor3000\relax}%
\providecommand \BibitemShut  [1]{\csname bibitem#1\endcsname}%
\let\auto@bib@innerbib\@empty
\bibitem [{\citenamefont {Zeng}\ \emph {et~al.}(2019)\citenamefont {Zeng}, \citenamefont {Chen}, \citenamefont {Zhou},\ and\ \citenamefont {Wen}}]{zeng_quantum_2019}%
  \BibitemOpen
  \bibfield  {author} {\bibinfo {author} {\bibfnamefont {B.}~\bibnamefont {Zeng}}, \bibinfo {author} {\bibfnamefont {X.}~\bibnamefont {Chen}}, \bibinfo {author} {\bibfnamefont {D.-L.}\ \bibnamefont {Zhou}},\ and\ \bibinfo {author} {\bibfnamefont {X.-G.}\ \bibnamefont {Wen}},\ }\href {https://doi.org/10.1007/978-1-4939-9084-9} {\emph {\bibinfo {title} {Quantum Information Meets Quantum Matter: From Quantum Entanglement to Topological Phases of Many-Body Systems}}},\ Quantum {{Science}} and {{Technology}}\ (\bibinfo  {publisher} {Springer-Verlag},\ \bibinfo {address} {New York},\ \bibinfo {year} {2019})\BibitemShut {NoStop}%
\bibitem [{\citenamefont {Wen}(1990)}]{wen_topological_1990}%
  \BibitemOpen
  \bibfield  {author} {\bibinfo {author} {\bibfnamefont {X.~G.}\ \bibnamefont {Wen}},\ }\bibfield  {title} {\bibinfo {title} {Topological orders in rigid states},\ }\href {https://doi.org/10.1142/S0217979290000139} {\bibfield  {journal} {\bibinfo  {journal} {International Journal of Modern Physics B}\ }\textbf {\bibinfo {volume} {04}},\ \bibinfo {pages} {239} (\bibinfo {year} {1990})}\BibitemShut {NoStop}%
\bibitem [{\citenamefont {Chen}\ \emph {et~al.}(2010)\citenamefont {Chen}, \citenamefont {Gu},\ and\ \citenamefont {Wen}}]{chen_local_2010}%
  \BibitemOpen
  \bibfield  {author} {\bibinfo {author} {\bibfnamefont {X.}~\bibnamefont {Chen}}, \bibinfo {author} {\bibfnamefont {Z.-C.}\ \bibnamefont {Gu}},\ and\ \bibinfo {author} {\bibfnamefont {X.-G.}\ \bibnamefont {Wen}},\ }\bibfield  {title} {\bibinfo {title} {Local unitary transformation, long-range quantum entanglement, wave function renormalization, and topological order},\ }\href {https://doi.org/10.1103/PhysRevB.82.155138} {\bibfield  {journal} {\bibinfo  {journal} {Physical Review B}\ }\textbf {\bibinfo {volume} {82}},\ \bibinfo {pages} {155138} (\bibinfo {year} {2010})},\ \Eprint {https://arxiv.org/abs/1004.3835} {arXiv:1004.3835} \BibitemShut {NoStop}%
\bibitem [{\citenamefont {Dennis}\ \emph {et~al.}(2002)\citenamefont {Dennis}, \citenamefont {Kitaev}, \citenamefont {Landahl},\ and\ \citenamefont {Preskill}}]{dennis_topological_2002}%
  \BibitemOpen
  \bibfield  {author} {\bibinfo {author} {\bibfnamefont {E.}~\bibnamefont {Dennis}}, \bibinfo {author} {\bibfnamefont {A.}~\bibnamefont {Kitaev}}, \bibinfo {author} {\bibfnamefont {A.}~\bibnamefont {Landahl}},\ and\ \bibinfo {author} {\bibfnamefont {J.}~\bibnamefont {Preskill}},\ }\bibfield  {title} {\bibinfo {title} {Topological quantum memory},\ }\href {https://doi.org/10.1063/1.1499754} {\bibfield  {journal} {\bibinfo  {journal} {Journal of Mathematical Physics}\ }\textbf {\bibinfo {volume} {43}},\ \bibinfo {pages} {4452} (\bibinfo {year} {2002})},\ \Eprint {https://arxiv.org/abs/quant-ph/0110143} {arXiv:quant-ph/0110143} \BibitemShut {NoStop}%
\bibitem [{\citenamefont {Sang}\ and\ \citenamefont {Hsieh}(2025)}]{sang_stability_2024}%
  \BibitemOpen
  \bibfield  {author} {\bibinfo {author} {\bibfnamefont {S.}~\bibnamefont {Sang}}\ and\ \bibinfo {author} {\bibfnamefont {T.~H.}\ \bibnamefont {Hsieh}},\ }\bibfield  {title} {\bibinfo {title} {Stability of mixed-state quantum phases via finite markov length},\ }\href {https://doi.org/10.1103/PhysRevLett.134.070403} {\bibfield  {journal} {\bibinfo  {journal} {Phys. Rev. Lett.}\ }\textbf {\bibinfo {volume} {134}},\ \bibinfo {pages} {070403} (\bibinfo {year} {2025})}\BibitemShut {NoStop}%
\bibitem [{\citenamefont {Fan}\ \emph {et~al.}(2024)\citenamefont {Fan}, \citenamefont {Bao}, \citenamefont {Altman},\ and\ \citenamefont {Vishwanath}}]{fan_diagnostics_2024}%
  \BibitemOpen
  \bibfield  {author} {\bibinfo {author} {\bibfnamefont {R.}~\bibnamefont {Fan}}, \bibinfo {author} {\bibfnamefont {Y.}~\bibnamefont {Bao}}, \bibinfo {author} {\bibfnamefont {E.}~\bibnamefont {Altman}},\ and\ \bibinfo {author} {\bibfnamefont {A.}~\bibnamefont {Vishwanath}},\ }\bibfield  {title} {\bibinfo {title} {Diagnostics of {{Mixed-State Topological Order}} and {{Breakdown}} of {{Quantum Memory}}},\ }\href {https://doi.org/10.1103/PRXQuantum.5.020343} {\bibfield  {journal} {\bibinfo  {journal} {PRX Quantum}\ }\textbf {\bibinfo {volume} {5}},\ \bibinfo {pages} {020343} (\bibinfo {year} {2024})}\BibitemShut {NoStop}%
\bibitem [{\citenamefont {Bao}\ \emph {et~al.}(2023)\citenamefont {Bao}, \citenamefont {Fan}, \citenamefont {Vishwanath},\ and\ \citenamefont {Altman}}]{bao_mixedstate_2023}%
  \BibitemOpen
  \bibfield  {author} {\bibinfo {author} {\bibfnamefont {Y.}~\bibnamefont {Bao}}, \bibinfo {author} {\bibfnamefont {R.}~\bibnamefont {Fan}}, \bibinfo {author} {\bibfnamefont {A.}~\bibnamefont {Vishwanath}},\ and\ \bibinfo {author} {\bibfnamefont {E.}~\bibnamefont {Altman}},\ }\href {https://doi.org/10.48550/arXiv.2301.05687} {\bibinfo {title} {Mixed-state topological order and the errorfield double formulation of decoherence-induced transitions}} (\bibinfo {year} {2023}),\ \Eprint {https://arxiv.org/abs/2301.05687} {arXiv:2301.05687 [cond-mat, physics:quant-ph]} \BibitemShut {NoStop}%
\bibitem [{\citenamefont {Wang}\ \emph {et~al.}(2025)\citenamefont {Wang}, \citenamefont {Wu},\ and\ \citenamefont {Wang}}]{wang_intrinsic_2025}%
  \BibitemOpen
  \bibfield  {author} {\bibinfo {author} {\bibfnamefont {Z.}~\bibnamefont {Wang}}, \bibinfo {author} {\bibfnamefont {Z.}~\bibnamefont {Wu}},\ and\ \bibinfo {author} {\bibfnamefont {Z.}~\bibnamefont {Wang}},\ }\bibfield  {title} {\bibinfo {title} {Intrinsic {{Mixed-State Topological Order}}},\ }\href {https://doi.org/10.1103/PRXQuantum.6.010314} {\bibfield  {journal} {\bibinfo  {journal} {PRX Quantum}\ }\textbf {\bibinfo {volume} {6}},\ \bibinfo {pages} {010314} (\bibinfo {year} {2025})}\BibitemShut {NoStop}%
\bibitem [{\citenamefont {Ellison}\ and\ \citenamefont {Cheng}(2025)}]{ellison_classification_2025}%
  \BibitemOpen
  \bibfield  {author} {\bibinfo {author} {\bibfnamefont {T.~D.}\ \bibnamefont {Ellison}}\ and\ \bibinfo {author} {\bibfnamefont {M.}~\bibnamefont {Cheng}},\ }\bibfield  {title} {\bibinfo {title} {Toward a {{Classification}} of {{Mixed-State Topological Orders}} in {{Two Dimensions}}},\ }\href {https://doi.org/10.1103/PRXQuantum.6.010315} {\bibfield  {journal} {\bibinfo  {journal} {PRX Quantum}\ }\textbf {\bibinfo {volume} {6}},\ \bibinfo {pages} {010315} (\bibinfo {year} {2025})}\BibitemShut {NoStop}%
\bibitem [{\citenamefont {Sohal}\ and\ \citenamefont {Prem}(2025)}]{sohal_noisy_2025}%
  \BibitemOpen
  \bibfield  {author} {\bibinfo {author} {\bibfnamefont {R.}~\bibnamefont {Sohal}}\ and\ \bibinfo {author} {\bibfnamefont {A.}~\bibnamefont {Prem}},\ }\bibfield  {title} {\bibinfo {title} {Noisy {{Approach}} to {{Intrinsically Mixed-State Topological Order}}},\ }\href {https://doi.org/10.1103/PRXQuantum.6.010313} {\bibfield  {journal} {\bibinfo  {journal} {PRX Quantum}\ }\textbf {\bibinfo {volume} {6}},\ \bibinfo {pages} {010313} (\bibinfo {year} {2025})}\BibitemShut {NoStop}%
\bibitem [{\citenamefont {Gaiotto}\ \emph {et~al.}(2015)\citenamefont {Gaiotto}, \citenamefont {Kapustin}, \citenamefont {Seiberg},\ and\ \citenamefont {Willett}}]{gaiotto_generalized_2014}%
  \BibitemOpen
  \bibfield  {author} {\bibinfo {author} {\bibfnamefont {D.}~\bibnamefont {Gaiotto}}, \bibinfo {author} {\bibfnamefont {A.}~\bibnamefont {Kapustin}}, \bibinfo {author} {\bibfnamefont {N.}~\bibnamefont {Seiberg}},\ and\ \bibinfo {author} {\bibfnamefont {B.}~\bibnamefont {Willett}},\ }\bibfield  {title} {\bibinfo {title} {Generalized global symmetries},\ }\href@noop {} {\bibfield  {journal} {\bibinfo  {journal} {Journal of High Energy Physics}\ }\textbf {\bibinfo {volume} {2015}},\ \bibinfo {pages} {1} (\bibinfo {year} {2015})}\BibitemShut {NoStop}%
\bibitem [{\citenamefont {McGreevy}(2023)}]{mcgreevy_generalized_2023}%
  \BibitemOpen
  \bibfield  {author} {\bibinfo {author} {\bibfnamefont {J.}~\bibnamefont {McGreevy}},\ }\bibfield  {title} {\bibinfo {title} {Generalized {{Symmetries}} in {{Condensed Matter}}},\ }\href {https://doi.org/10.1146/annurev-conmatphys-040721-021029} {\bibfield  {journal} {\bibinfo  {journal} {Annual Review of Condensed Matter Physics}\ }\textbf {\bibinfo {volume} {14}},\ \bibinfo {pages} {57} (\bibinfo {year} {2023})}\BibitemShut {NoStop}%
\bibitem [{\citenamefont {Bu{\v c}a}\ and\ \citenamefont {Prosen}(2012)}]{buca_note_2012}%
  \BibitemOpen
  \bibfield  {author} {\bibinfo {author} {\bibfnamefont {B.}~\bibnamefont {Bu{\v c}a}}\ and\ \bibinfo {author} {\bibfnamefont {T.}~\bibnamefont {Prosen}},\ }\bibfield  {title} {\bibinfo {title} {A note on symmetry reductions of the {{Lindblad}} equation: Transport in constrained open spin chains},\ }\href {https://doi.org/10.1088/1367-2630/14/7/073007} {\bibfield  {journal} {\bibinfo  {journal} {New Journal of Physics}\ }\textbf {\bibinfo {volume} {14}},\ \bibinfo {pages} {073007} (\bibinfo {year} {2012})}\BibitemShut {NoStop}%
\bibitem [{\citenamefont {de~Groot}\ \emph {et~al.}(2022)\citenamefont {de~Groot}, \citenamefont {Turzillo},\ and\ \citenamefont {Schuch}}]{groot_symmetry_2022}%
  \BibitemOpen
  \bibfield  {author} {\bibinfo {author} {\bibfnamefont {C.}~\bibnamefont {de~Groot}}, \bibinfo {author} {\bibfnamefont {A.}~\bibnamefont {Turzillo}},\ and\ \bibinfo {author} {\bibfnamefont {N.}~\bibnamefont {Schuch}},\ }\bibfield  {title} {\bibinfo {title} {Symmetry {{Protected Topological Order}} in {{Open Quantum Systems}}},\ }\href {https://doi.org/10.22331/q-2022-11-10-856} {\bibfield  {journal} {\bibinfo  {journal} {Quantum}\ }\textbf {\bibinfo {volume} {6}},\ \bibinfo {pages} {856} (\bibinfo {year} {2022})}\BibitemShut {NoStop}%
\bibitem [{\citenamefont {Sala}\ \emph {et~al.}(2024)\citenamefont {Sala}, \citenamefont {Gopalakrishnan}, \citenamefont {Oshikawa},\ and\ \citenamefont {You}}]{sala_spontaneous_2024}%
  \BibitemOpen
  \bibfield  {author} {\bibinfo {author} {\bibfnamefont {P.}~\bibnamefont {Sala}}, \bibinfo {author} {\bibfnamefont {S.}~\bibnamefont {Gopalakrishnan}}, \bibinfo {author} {\bibfnamefont {M.}~\bibnamefont {Oshikawa}},\ and\ \bibinfo {author} {\bibfnamefont {Y.}~\bibnamefont {You}},\ }\href {https://doi.org/10.48550/arXiv.2405.02402} {\bibinfo {title} {Spontaneous {{Strong Symmetry Breaking}} in {{Open Systems}}: {{Purification Perspective}}}} (\bibinfo {year} {2024}),\ \Eprint {https://arxiv.org/abs/2405.02402} {arXiv:2405.02402 [cond-mat, physics:hep-th, physics:quant-ph]} \BibitemShut {NoStop}%
\bibitem [{\citenamefont {Lessa}\ \emph {et~al.}(2024{\natexlab{a}})\citenamefont {Lessa}, \citenamefont {Ma}, \citenamefont {Zhang}, \citenamefont {Bi}, \citenamefont {Cheng},\ and\ \citenamefont {Wang}}]{lessa_strongweak_2024}%
  \BibitemOpen
  \bibfield  {author} {\bibinfo {author} {\bibfnamefont {L.~A.}\ \bibnamefont {Lessa}}, \bibinfo {author} {\bibfnamefont {R.}~\bibnamefont {Ma}}, \bibinfo {author} {\bibfnamefont {J.-H.}\ \bibnamefont {Zhang}}, \bibinfo {author} {\bibfnamefont {Z.}~\bibnamefont {Bi}}, \bibinfo {author} {\bibfnamefont {M.}~\bibnamefont {Cheng}},\ and\ \bibinfo {author} {\bibfnamefont {C.}~\bibnamefont {Wang}},\ }\href {https://doi.org/10.48550/arXiv.2405.03639} {\bibinfo {title} {Strong-to-{{Weak Spontaneous Symmetry Breaking}} in {{Mixed Quantum States}}}} (\bibinfo {year} {2024}{\natexlab{a}}),\ \Eprint {https://arxiv.org/abs/2405.03639} {arXiv:2405.03639 [cond-mat, physics:quant-ph]} \BibitemShut {NoStop}%
\bibitem [{\citenamefont {Zhang}\ \emph {et~al.}(2024{\natexlab{a}})\citenamefont {Zhang}, \citenamefont {Xu}, \citenamefont {Zhang}, \citenamefont {Xu}, \citenamefont {Bi},\ and\ \citenamefont {Luo}}]{zhang_strongweak_2024}%
  \BibitemOpen
  \bibfield  {author} {\bibinfo {author} {\bibfnamefont {C.}~\bibnamefont {Zhang}}, \bibinfo {author} {\bibfnamefont {Y.}~\bibnamefont {Xu}}, \bibinfo {author} {\bibfnamefont {J.-H.}\ \bibnamefont {Zhang}}, \bibinfo {author} {\bibfnamefont {C.}~\bibnamefont {Xu}}, \bibinfo {author} {\bibfnamefont {Z.}~\bibnamefont {Bi}},\ and\ \bibinfo {author} {\bibfnamefont {Z.-X.}\ \bibnamefont {Luo}},\ }\href {https://doi.org/10.48550/arXiv.2409.17530} {\bibinfo {title} {Strong-to-weak spontaneous breaking of 1-form symmetry and intrinsically mixed topological order}} (\bibinfo {year} {2024}{\natexlab{a}}),\ \Eprint {https://arxiv.org/abs/2409.17530} {arXiv:2409.17530} \BibitemShut {NoStop}%
\bibitem [{\citenamefont {Ma}\ and\ \citenamefont {Wang}(2023)}]{ma_average_2023}%
  \BibitemOpen
  \bibfield  {author} {\bibinfo {author} {\bibfnamefont {R.}~\bibnamefont {Ma}}\ and\ \bibinfo {author} {\bibfnamefont {C.}~\bibnamefont {Wang}},\ }\bibfield  {title} {\bibinfo {title} {Average {{Symmetry-Protected Topological Phases}}},\ }\href {https://doi.org/10.1103/PhysRevX.13.031016} {\bibfield  {journal} {\bibinfo  {journal} {Physical Review X}\ }\textbf {\bibinfo {volume} {13}},\ \bibinfo {pages} {031016} (\bibinfo {year} {2023})}\BibitemShut {NoStop}%
\bibitem [{\citenamefont {Ma}\ \emph {et~al.}(2023)\citenamefont {Ma}, \citenamefont {Zhang}, \citenamefont {Bi}, \citenamefont {Cheng},\ and\ \citenamefont {Wang}}]{ma_topological_2023}%
  \BibitemOpen
  \bibfield  {author} {\bibinfo {author} {\bibfnamefont {R.}~\bibnamefont {Ma}}, \bibinfo {author} {\bibfnamefont {J.-H.}\ \bibnamefont {Zhang}}, \bibinfo {author} {\bibfnamefont {Z.}~\bibnamefont {Bi}}, \bibinfo {author} {\bibfnamefont {M.}~\bibnamefont {Cheng}},\ and\ \bibinfo {author} {\bibfnamefont {C.}~\bibnamefont {Wang}},\ }\href {https://doi.org/10.48550/arXiv.2305.16399} {\bibinfo {title} {Topological {{Phases}} with {{Average Symmetries}}: The {{Decohered}}, the {{Disordered}}, and the {{Intrinsic}}}} (\bibinfo {year} {2023}),\ \Eprint {https://arxiv.org/abs/2305.16399} {arXiv:2305.16399 [cond-mat, physics:math-ph, physics:quant-ph]} \BibitemShut {NoStop}%
\bibitem [{\citenamefont {Ma}\ and\ \citenamefont {Turzillo}(2024)}]{ma_symmetry_2024}%
  \BibitemOpen
  \bibfield  {author} {\bibinfo {author} {\bibfnamefont {R.}~\bibnamefont {Ma}}\ and\ \bibinfo {author} {\bibfnamefont {A.}~\bibnamefont {Turzillo}},\ }\href {https://doi.org/10.48550/arXiv.2403.13280} {\bibinfo {title} {Symmetry {{Protected Topological Phases}} of {{Mixed States}} in the {{Doubled Space}}}} (\bibinfo {year} {2024}),\ \Eprint {https://arxiv.org/abs/2403.13280} {arxiv:2403.13280 [cond-mat, physics:quant-ph]} \BibitemShut {NoStop}%
\bibitem [{\citenamefont {Lessa}\ \emph {et~al.}(2024{\natexlab{b}})\citenamefont {Lessa}, \citenamefont {Cheng},\ and\ \citenamefont {Wang}}]{lessa_mixedstate_2024}%
  \BibitemOpen
  \bibfield  {author} {\bibinfo {author} {\bibfnamefont {L.~A.}\ \bibnamefont {Lessa}}, \bibinfo {author} {\bibfnamefont {M.}~\bibnamefont {Cheng}},\ and\ \bibinfo {author} {\bibfnamefont {C.}~\bibnamefont {Wang}},\ }\href {https://doi.org/10.48550/arXiv.2401.17357} {\bibinfo {title} {Mixed-state quantum anomaly and multipartite entanglement}} (\bibinfo {year} {2024}{\natexlab{b}}),\ \Eprint {https://arxiv.org/abs/2401.17357} {arxiv:2401.17357 [cond-mat, physics:hep-th, physics:quant-ph]} \BibitemShut {NoStop}%
\bibitem [{\citenamefont {Wang}\ and\ \citenamefont {Li}(2024)}]{wang_anomaly_2024}%
  \BibitemOpen
  \bibfield  {author} {\bibinfo {author} {\bibfnamefont {Z.}~\bibnamefont {Wang}}\ and\ \bibinfo {author} {\bibfnamefont {L.}~\bibnamefont {Li}},\ }\href {https://doi.org/10.48550/arXiv.2403.14533} {\bibinfo {title} {Anomaly in open quantum systems and its implications on mixed-state quantum phases}} (\bibinfo {year} {2024}),\ \Eprint {https://arxiv.org/abs/2403.14533} {arXiv:2403.14533 [cond-mat, physics:math-ph, physics:quant-ph]} \BibitemShut {NoStop}%
\bibitem [{\citenamefont {Coser}\ and\ \citenamefont {{P{\'e}rez-Garc{\'i}a}}(2019)}]{coser_classification_2019a}%
  \BibitemOpen
  \bibfield  {author} {\bibinfo {author} {\bibfnamefont {A.}~\bibnamefont {Coser}}\ and\ \bibinfo {author} {\bibfnamefont {D.}~\bibnamefont {{P{\'e}rez-Garc{\'i}a}}},\ }\bibfield  {title} {\bibinfo {title} {Classification of phases for mixed states via fast dissipative evolution},\ }\href {https://doi.org/10.22331/q-2019-08-12-174} {\bibfield  {journal} {\bibinfo  {journal} {Quantum}\ }\textbf {\bibinfo {volume} {3}},\ \bibinfo {pages} {174} (\bibinfo {year} {2019})}\BibitemShut {NoStop}%
\bibitem [{\citenamefont {Sang}\ \emph {et~al.}(2024)\citenamefont {Sang}, \citenamefont {Zou},\ and\ \citenamefont {Hsieh}}]{sang2024mixed}%
  \BibitemOpen
  \bibfield  {author} {\bibinfo {author} {\bibfnamefont {S.}~\bibnamefont {Sang}}, \bibinfo {author} {\bibfnamefont {Y.}~\bibnamefont {Zou}},\ and\ \bibinfo {author} {\bibfnamefont {T.~H.}\ \bibnamefont {Hsieh}},\ }\bibfield  {title} {\bibinfo {title} {Mixed-state quantum phases: Renormalization and quantum error correction},\ }\href@noop {} {\bibfield  {journal} {\bibinfo  {journal} {Physical Review X}\ }\textbf {\bibinfo {volume} {14}},\ \bibinfo {pages} {031044} (\bibinfo {year} {2024})}\BibitemShut {NoStop}%
\bibitem [{\citenamefont {Hastings}(2011)}]{hastings2011topological}%
  \BibitemOpen
  \bibfield  {author} {\bibinfo {author} {\bibfnamefont {M.~B.}\ \bibnamefont {Hastings}},\ }\bibfield  {title} {\bibinfo {title} {Topological order at nonzero temperature},\ }\href@noop {} {\bibfield  {journal} {\bibinfo  {journal} {Physical review letters}\ }\textbf {\bibinfo {volume} {107}},\ \bibinfo {pages} {210501} (\bibinfo {year} {2011})}\BibitemShut {NoStop}%
\bibitem [{\citenamefont {Chen}\ and\ \citenamefont {Grover}(2024{\natexlab{a}})}]{chen2024separability}%
  \BibitemOpen
  \bibfield  {author} {\bibinfo {author} {\bibfnamefont {Y.-H.}\ \bibnamefont {Chen}}\ and\ \bibinfo {author} {\bibfnamefont {T.}~\bibnamefont {Grover}},\ }\bibfield  {title} {\bibinfo {title} {Separability transitions in topological states induced by local decoherence},\ }\href@noop {} {\bibfield  {journal} {\bibinfo  {journal} {Physical Review Letters}\ }\textbf {\bibinfo {volume} {132}},\ \bibinfo {pages} {170602} (\bibinfo {year} {2024}{\natexlab{a}})}\BibitemShut {NoStop}%
\bibitem [{\citenamefont {Chen}\ and\ \citenamefont {Grover}(2024{\natexlab{b}})}]{chenUnconventionalTopologicalMixedstate2024}%
  \BibitemOpen
  \bibfield  {author} {\bibinfo {author} {\bibfnamefont {Y.-H.}\ \bibnamefont {Chen}}\ and\ \bibinfo {author} {\bibfnamefont {T.}~\bibnamefont {Grover}},\ }\href {https://doi.org/10.48550/arXiv.2403.06553} {\bibinfo {title} {Unconventional topological mixed-state transition and critical phase induced by self-dual coherent errors}} (\bibinfo {year} {2024}{\natexlab{b}}),\ \Eprint {https://arxiv.org/abs/2403.06553} {2403.06553 [cond-mat, physics:quant-ph]} \BibitemShut {NoStop}%
\bibitem [{\citenamefont {Wang}\ \emph {et~al.}(2024)\citenamefont {Wang}, \citenamefont {Song}, \citenamefont {Meng},\ and\ \citenamefont {Grover}}]{wang_analog_2024}%
  \BibitemOpen
  \bibfield  {author} {\bibinfo {author} {\bibfnamefont {T.-T.}\ \bibnamefont {Wang}}, \bibinfo {author} {\bibfnamefont {M.}~\bibnamefont {Song}}, \bibinfo {author} {\bibfnamefont {Z.~Y.}\ \bibnamefont {Meng}},\ and\ \bibinfo {author} {\bibfnamefont {T.}~\bibnamefont {Grover}},\ }\href {https://doi.org/10.48550/arXiv.2407.20500} {\bibinfo {title} {An analog of topological entanglement entropy for mixed states}} (\bibinfo {year} {2024}),\ \Eprint {https://arxiv.org/abs/2407.20500} {arXiv:2407.20500} \BibitemShut {NoStop}%
\bibitem [{\citenamefont {Kitaev}\ and\ \citenamefont {Preskill}(2006)}]{kitaev_topological_2006}%
  \BibitemOpen
  \bibfield  {author} {\bibinfo {author} {\bibfnamefont {A.}~\bibnamefont {Kitaev}}\ and\ \bibinfo {author} {\bibfnamefont {J.}~\bibnamefont {Preskill}},\ }\bibfield  {title} {\bibinfo {title} {Topological {{Entanglement Entropy}}},\ }\href {https://doi.org/10.1103/PhysRevLett.96.110404} {\bibfield  {journal} {\bibinfo  {journal} {Physical Review Letters}\ }\textbf {\bibinfo {volume} {96}},\ \bibinfo {pages} {110404} (\bibinfo {year} {2006})}\BibitemShut {NoStop}%
\bibitem [{\citenamefont {Levin}\ and\ \citenamefont {Wen}(2006)}]{levinDetectingTopologicalOrder2006}%
  \BibitemOpen
  \bibfield  {author} {\bibinfo {author} {\bibfnamefont {M.}~\bibnamefont {Levin}}\ and\ \bibinfo {author} {\bibfnamefont {X.-G.}\ \bibnamefont {Wen}},\ }\bibfield  {title} {\bibinfo {title} {Detecting topological order in a ground state wave function},\ }\href {https://doi.org/10.1103/PhysRevLett.96.110405} {\bibfield  {journal} {\bibinfo  {journal} {Physical Review Letters}\ }\textbf {\bibinfo {volume} {96}},\ \bibinfo {pages} {110405} (\bibinfo {year} {2006})},\ \Eprint {https://arxiv.org/abs/cond-mat/0510613} {cond-mat/0510613} \BibitemShut {NoStop}%
\bibitem [{\citenamefont {Bennett}\ \emph {et~al.}(1996)\citenamefont {Bennett}, \citenamefont {DiVincenzo}, \citenamefont {Smolin},\ and\ \citenamefont {Wootters}}]{bennett_mixedstate_1996}%
  \BibitemOpen
  \bibfield  {author} {\bibinfo {author} {\bibfnamefont {C.~H.}\ \bibnamefont {Bennett}}, \bibinfo {author} {\bibfnamefont {D.~P.}\ \bibnamefont {DiVincenzo}}, \bibinfo {author} {\bibfnamefont {J.~A.}\ \bibnamefont {Smolin}},\ and\ \bibinfo {author} {\bibfnamefont {W.~K.}\ \bibnamefont {Wootters}},\ }\bibfield  {title} {\bibinfo {title} {Mixed-state entanglement and quantum error correction},\ }\href {https://doi.org/10.1103/PhysRevA.54.3824} {\bibfield  {journal} {\bibinfo  {journal} {Physical Review A}\ }\textbf {\bibinfo {volume} {54}},\ \bibinfo {pages} {3824} (\bibinfo {year} {1996})}\BibitemShut {NoStop}%
\bibitem [{\citenamefont {Kim}\ \emph {et~al.}(2023)\citenamefont {Kim}, \citenamefont {Levin}, \citenamefont {Lin}, \citenamefont {Ranard},\ and\ \citenamefont {Shi}}]{kim_universal_2023}%
  \BibitemOpen
  \bibfield  {author} {\bibinfo {author} {\bibfnamefont {I.~H.}\ \bibnamefont {Kim}}, \bibinfo {author} {\bibfnamefont {M.}~\bibnamefont {Levin}}, \bibinfo {author} {\bibfnamefont {T.-C.}\ \bibnamefont {Lin}}, \bibinfo {author} {\bibfnamefont {D.}~\bibnamefont {Ranard}},\ and\ \bibinfo {author} {\bibfnamefont {B.}~\bibnamefont {Shi}},\ }\bibfield  {title} {\bibinfo {title} {Universal lower bound on topological entanglement entropy},\ }\href {https://doi.org/10.1103/PhysRevLett.131.166601} {\bibfield  {journal} {\bibinfo  {journal} {Physical Review Letters}\ }\textbf {\bibinfo {volume} {131}},\ \bibinfo {pages} {166601} (\bibinfo {year} {2023})},\ \Eprint {https://arxiv.org/abs/2302.00689} {arXiv:2302.00689 [quant-ph]} \BibitemShut {NoStop}%
\bibitem [{\citenamefont {Levin}(2024)}]{levin_physical_2024}%
  \BibitemOpen
  \bibfield  {author} {\bibinfo {author} {\bibfnamefont {M.}~\bibnamefont {Levin}},\ }\bibfield  {title} {\bibinfo {title} {Physical proof of the topological entanglement entropy inequality},\ }\href {https://doi.org/10.1103/PhysRevB.110.165154} {\bibfield  {journal} {\bibinfo  {journal} {Phys. Rev. B}\ }\textbf {\bibinfo {volume} {110}},\ \bibinfo {pages} {165154} (\bibinfo {year} {2024})}\BibitemShut {NoStop}%
\bibitem [{\citenamefont {Li}\ \emph {et~al.}(2024)\citenamefont {Li}, \citenamefont {Lee},\ and\ \citenamefont {Yoshida}}]{li_how_2024}%
  \BibitemOpen
  \bibfield  {author} {\bibinfo {author} {\bibfnamefont {Z.}~\bibnamefont {Li}}, \bibinfo {author} {\bibfnamefont {D.}~\bibnamefont {Lee}},\ and\ \bibinfo {author} {\bibfnamefont {B.}~\bibnamefont {Yoshida}},\ }\href {https://doi.org/10.48550/arXiv.2405.07970} {\bibinfo {title} {How much entanglement is needed for emergent anyons and fermions?}} (\bibinfo {year} {2024}),\ \Eprint {https://arxiv.org/abs/2405.07970} {arXiv:2405.07970} \BibitemShut {NoStop}%
\bibitem [{\citenamefont {Uhlmann}(1998)}]{uhlmann_entropy_1998}%
  \BibitemOpen
  \bibfield  {author} {\bibinfo {author} {\bibfnamefont {A.}~\bibnamefont {Uhlmann}},\ }\bibfield  {title} {\bibinfo {title} {Entropy and {{Optimal Decompositions}} of {{States Relative}} to a {{Maximal Commutative Subalgebra}}},\ }\href {https://doi.org/10.1023/A:1009664331611} {\bibfield  {journal} {\bibinfo  {journal} {Open Systems \& Information Dynamics}\ }\textbf {\bibinfo {volume} {5}},\ \bibinfo {pages} {209} (\bibinfo {year} {1998})}\BibitemShut {NoStop}%
\bibitem [{\citenamefont {Horodecki}\ \emph {et~al.}(2009)\citenamefont {Horodecki}, \citenamefont {Horodecki}, \citenamefont {Horodecki},\ and\ \citenamefont {Horodecki}}]{horodecki_quantum_2009}%
  \BibitemOpen
  \bibfield  {author} {\bibinfo {author} {\bibfnamefont {R.}~\bibnamefont {Horodecki}}, \bibinfo {author} {\bibfnamefont {P.}~\bibnamefont {Horodecki}}, \bibinfo {author} {\bibfnamefont {M.}~\bibnamefont {Horodecki}},\ and\ \bibinfo {author} {\bibfnamefont {K.}~\bibnamefont {Horodecki}},\ }\bibfield  {title} {\bibinfo {title} {Quantum entanglement},\ }\href {https://doi.org/10.1103/RevModPhys.81.865} {\bibfield  {journal} {\bibinfo  {journal} {Reviews of Modern Physics}\ }\textbf {\bibinfo {volume} {81}},\ \bibinfo {pages} {865} (\bibinfo {year} {2009})}\BibitemShut {NoStop}%
\bibitem [{\citenamefont {Chitambar}\ and\ \citenamefont {Gour}(2019)}]{chitambar_quantum_2019}%
  \BibitemOpen
  \bibfield  {author} {\bibinfo {author} {\bibfnamefont {E.}~\bibnamefont {Chitambar}}\ and\ \bibinfo {author} {\bibfnamefont {G.}~\bibnamefont {Gour}},\ }\bibfield  {title} {\bibinfo {title} {Quantum resource theories},\ }\href {https://doi.org/10.1103/RevModPhys.91.025001} {\bibfield  {journal} {\bibinfo  {journal} {Reviews of Modern Physics}\ }\textbf {\bibinfo {volume} {91}},\ \bibinfo {pages} {025001} (\bibinfo {year} {2019})},\ \bibinfo {note} {publisher: American Physical Society}\BibitemShut {NoStop}%
\bibitem [{\citenamefont {Buscemi}(2012)}]{buscemi_all_2012}%
  \BibitemOpen
  \bibfield  {author} {\bibinfo {author} {\bibfnamefont {F.}~\bibnamefont {Buscemi}},\ }\bibfield  {title} {\bibinfo {title} {All {{Entangled Quantum States Are Nonlocal}}},\ }\href {https://doi.org/10.1103/PhysRevLett.108.200401} {\bibfield  {journal} {\bibinfo  {journal} {Physical Review Letters}\ }\textbf {\bibinfo {volume} {108}},\ \bibinfo {pages} {200401} (\bibinfo {year} {2012})}\BibitemShut {NoStop}%
\bibitem [{\citenamefont {Schmid}\ \emph {et~al.}(2023)\citenamefont {Schmid}, \citenamefont {Fraser}, \citenamefont {Kunjwal}, \citenamefont {Sainz}, \citenamefont {Wolfe},\ and\ \citenamefont {Spekkens}}]{schmid_understanding_2023}%
  \BibitemOpen
  \bibfield  {author} {\bibinfo {author} {\bibfnamefont {D.}~\bibnamefont {Schmid}}, \bibinfo {author} {\bibfnamefont {T.~C.}\ \bibnamefont {Fraser}}, \bibinfo {author} {\bibfnamefont {R.}~\bibnamefont {Kunjwal}}, \bibinfo {author} {\bibfnamefont {A.~B.}\ \bibnamefont {Sainz}}, \bibinfo {author} {\bibfnamefont {E.}~\bibnamefont {Wolfe}},\ and\ \bibinfo {author} {\bibfnamefont {R.~W.}\ \bibnamefont {Spekkens}},\ }\bibfield  {title} {\bibinfo {title} {Understanding the interplay of entanglement and nonlocality: motivating and developing a new branch of entanglement theory},\ }\bibfield  {journal} {\bibinfo  {journal} {Quantum}\ }\textbf {\bibinfo {volume} {7}},\ \href {https://doi.org/10.22331/q-2023-12-04-1194} {10.22331/q-2023-12-04-1194} (\bibinfo {year} {2023})\BibitemShut {NoStop}%
\bibitem [{\citenamefont {Piroli}\ \emph {et~al.}(2021)\citenamefont {Piroli}, \citenamefont {Styliaris},\ and\ \citenamefont {Cirac}}]{piroli_quantum_2021}%
  \BibitemOpen
  \bibfield  {author} {\bibinfo {author} {\bibfnamefont {L.}~\bibnamefont {Piroli}}, \bibinfo {author} {\bibfnamefont {G.}~\bibnamefont {Styliaris}},\ and\ \bibinfo {author} {\bibfnamefont {J.~I.}\ \bibnamefont {Cirac}},\ }\bibfield  {title} {\bibinfo {title} {Quantum {{Circuits Assisted}} by {{Local Operations}} and {{Classical Communication}}: {{Transformations}} and {{Phases}} of {{Matter}}},\ }\href {https://doi.org/10.1103/PhysRevLett.127.220503} {\bibfield  {journal} {\bibinfo  {journal} {Physical Review Letters}\ }\textbf {\bibinfo {volume} {127}},\ \bibinfo {pages} {220503} (\bibinfo {year} {2021})}\BibitemShut {NoStop}%
\bibitem [{\citenamefont {Friedman}\ \emph {et~al.}(2022)\citenamefont {Friedman}, \citenamefont {Yin}, \citenamefont {Hong},\ and\ \citenamefont {Lucas}}]{friedman_locality_2022}%
  \BibitemOpen
  \bibfield  {author} {\bibinfo {author} {\bibfnamefont {A.~J.}\ \bibnamefont {Friedman}}, \bibinfo {author} {\bibfnamefont {C.}~\bibnamefont {Yin}}, \bibinfo {author} {\bibfnamefont {Y.}~\bibnamefont {Hong}},\ and\ \bibinfo {author} {\bibfnamefont {A.}~\bibnamefont {Lucas}},\ }\href {https://doi.org/10.48550/arXiv.2206.09929} {\bibinfo {title} {Locality and error correction in quantum dynamics with measurement}} (\bibinfo {year} {2022}),\ \Eprint {https://arxiv.org/abs/2206.09929} {arXiv:2206.09929 [cond-mat, physics:math-ph, physics:quant-ph]} \BibitemShut {NoStop}%
\bibitem [{\citenamefont {Lee}\ \emph {et~al.}(2022)\citenamefont {Lee}, \citenamefont {Ji}, \citenamefont {Bi},\ and\ \citenamefont {Fisher}}]{lee_measurementprepared_2022}%
  \BibitemOpen
  \bibfield  {author} {\bibinfo {author} {\bibfnamefont {J.~Y.}\ \bibnamefont {Lee}}, \bibinfo {author} {\bibfnamefont {W.}~\bibnamefont {Ji}}, \bibinfo {author} {\bibfnamefont {Z.}~\bibnamefont {Bi}},\ and\ \bibinfo {author} {\bibfnamefont {M.~P.~A.}\ \bibnamefont {Fisher}},\ }\href@noop {} {\bibinfo {title} {Measurement-{{Prepared Quantum Criticality}}: From {{Ising}} model to gauge theory, and beyond}} (\bibinfo {year} {2022}),\ \Eprint {https://arxiv.org/abs/2208.11699} {arXiv:2208.11699 [cond-mat, physics:quant-ph]} \BibitemShut {NoStop}%
\bibitem [{\citenamefont {Lu}\ \emph {et~al.}(2022)\citenamefont {Lu}, \citenamefont {Lessa}, \citenamefont {Kim},\ and\ \citenamefont {Hsieh}}]{lu_measurement_2022}%
  \BibitemOpen
  \bibfield  {author} {\bibinfo {author} {\bibfnamefont {T.-C.}\ \bibnamefont {Lu}}, \bibinfo {author} {\bibfnamefont {L.~A.}\ \bibnamefont {Lessa}}, \bibinfo {author} {\bibfnamefont {I.~H.}\ \bibnamefont {Kim}},\ and\ \bibinfo {author} {\bibfnamefont {T.~H.}\ \bibnamefont {Hsieh}},\ }\bibfield  {title} {\bibinfo {title} {Measurement as a {Shortcut} to {Long}-{Range} {Entangled} {Quantum} {Matter}},\ }\bibfield  {journal} {\bibinfo  {journal} {PRX Quantum}\ }\textbf {\bibinfo {volume} {3}},\ \href {https://doi.org/10.1103/PRXQuantum.3.040337} {10.1103/PRXQuantum.3.040337} (\bibinfo {year} {2022})\BibitemShut {NoStop}%
\bibitem [{\citenamefont {Tantivasadakarn}\ \emph {et~al.}(2023)\citenamefont {Tantivasadakarn}, \citenamefont {Vishwanath},\ and\ \citenamefont {Verresen}}]{tantivasadakarn_hierarchy_2023}%
  \BibitemOpen
  \bibfield  {author} {\bibinfo {author} {\bibfnamefont {N.}~\bibnamefont {Tantivasadakarn}}, \bibinfo {author} {\bibfnamefont {A.}~\bibnamefont {Vishwanath}},\ and\ \bibinfo {author} {\bibfnamefont {R.}~\bibnamefont {Verresen}},\ }\bibfield  {title} {\bibinfo {title} {Hierarchy of {Topological} {Order} {From} {Finite}-{Depth} {Unitaries}, {Measurement}, and {Feedforward}},\ }\bibfield  {journal} {\bibinfo  {journal} {PRX Quantum}\ }\textbf {\bibinfo {volume} {4}},\ \href {https://doi.org/10.1103/PRXQuantum.4.020339} {10.1103/PRXQuantum.4.020339} (\bibinfo {year} {2023})\BibitemShut {NoStop}%
\bibitem [{\citenamefont {Foss-Feig}\ \emph {et~al.}(2023)\citenamefont {Foss-Feig}, \citenamefont {Tikku}, \citenamefont {Lu}, \citenamefont {Mayer}, \citenamefont {Iqbal}, \citenamefont {Gatterman}, \citenamefont {Gerber}, \citenamefont {Gilmore}, \citenamefont {Gresh}, \citenamefont {Hankin}, \citenamefont {Hewitt}, \citenamefont {Horst}, \citenamefont {Matheny}, \citenamefont {Mengle}, \citenamefont {Neyenhuis}, \citenamefont {Dreyer}, \citenamefont {Hayes}, \citenamefont {Hsieh},\ and\ \citenamefont {Kim}}]{foss-feig_experimental_2023}%
  \BibitemOpen
  \bibfield  {author} {\bibinfo {author} {\bibfnamefont {M.}~\bibnamefont {Foss-Feig}}, \bibinfo {author} {\bibfnamefont {A.}~\bibnamefont {Tikku}}, \bibinfo {author} {\bibfnamefont {T.-C.}\ \bibnamefont {Lu}}, \bibinfo {author} {\bibfnamefont {K.}~\bibnamefont {Mayer}}, \bibinfo {author} {\bibfnamefont {M.}~\bibnamefont {Iqbal}}, \bibinfo {author} {\bibfnamefont {T.~M.}\ \bibnamefont {Gatterman}}, \bibinfo {author} {\bibfnamefont {J.~A.}\ \bibnamefont {Gerber}}, \bibinfo {author} {\bibfnamefont {K.}~\bibnamefont {Gilmore}}, \bibinfo {author} {\bibfnamefont {D.}~\bibnamefont {Gresh}}, \bibinfo {author} {\bibfnamefont {A.}~\bibnamefont {Hankin}}, \bibinfo {author} {\bibfnamefont {N.}~\bibnamefont {Hewitt}}, \bibinfo {author} {\bibfnamefont {C.~V.}\ \bibnamefont {Horst}}, \bibinfo {author} {\bibfnamefont {M.}~\bibnamefont {Matheny}}, \bibinfo {author} {\bibfnamefont {T.}~\bibnamefont {Mengle}}, \bibinfo {author} {\bibfnamefont {B.}~\bibnamefont {Neyenhuis}}, \bibinfo {author} {\bibfnamefont {H.}~\bibnamefont
  {Dreyer}}, \bibinfo {author} {\bibfnamefont {D.}~\bibnamefont {Hayes}}, \bibinfo {author} {\bibfnamefont {T.~H.}\ \bibnamefont {Hsieh}},\ and\ \bibinfo {author} {\bibfnamefont {I.~H.}\ \bibnamefont {Kim}},\ }\href {https://doi.org/10.48550/arXiv.2302.03029} {\bibinfo {title} {Experimental demonstration of the advantage of adaptive quantum circuits}} (\bibinfo {year} {2023})\BibitemShut {NoStop}%
\bibitem [{\citenamefont {Tantivasadakarn}\ \emph {et~al.}(2024)\citenamefont {Tantivasadakarn}, \citenamefont {Thorngren}, \citenamefont {Vishwanath},\ and\ \citenamefont {Verresen}}]{tantivasadakarn_longrange_2024}%
  \BibitemOpen
  \bibfield  {author} {\bibinfo {author} {\bibfnamefont {N.}~\bibnamefont {Tantivasadakarn}}, \bibinfo {author} {\bibfnamefont {R.}~\bibnamefont {Thorngren}}, \bibinfo {author} {\bibfnamefont {A.}~\bibnamefont {Vishwanath}},\ and\ \bibinfo {author} {\bibfnamefont {R.}~\bibnamefont {Verresen}},\ }\bibfield  {title} {\bibinfo {title} {Long-{{Range Entanglement}} from {{Measuring Symmetry-Protected Topological Phases}}},\ }\href {https://doi.org/10.1103/PhysRevX.14.021040} {\bibfield  {journal} {\bibinfo  {journal} {Physical Review X}\ }\textbf {\bibinfo {volume} {14}},\ \bibinfo {pages} {021040} (\bibinfo {year} {2024})}\BibitemShut {NoStop}%
\bibitem [{\citenamefont {Iqbal}\ \emph {et~al.}(2024)\citenamefont {Iqbal}, \citenamefont {Tantivasadakarn}, \citenamefont {Verresen}, \citenamefont {Campbell}, \citenamefont {Dreiling}, \citenamefont {Figgatt}, \citenamefont {Gaebler}, \citenamefont {Johansen}, \citenamefont {Mills}, \citenamefont {Moses}, \citenamefont {Pino}, \citenamefont {Ransford}, \citenamefont {Rowe}, \citenamefont {Siegfried}, \citenamefont {Stutz}, \citenamefont {{Foss-Feig}}, \citenamefont {Vishwanath},\ and\ \citenamefont {Dreyer}}]{iqbal_nonabelian_2024}%
  \BibitemOpen
  \bibfield  {author} {\bibinfo {author} {\bibfnamefont {M.}~\bibnamefont {Iqbal}}, \bibinfo {author} {\bibfnamefont {N.}~\bibnamefont {Tantivasadakarn}}, \bibinfo {author} {\bibfnamefont {R.}~\bibnamefont {Verresen}}, \bibinfo {author} {\bibfnamefont {S.~L.}\ \bibnamefont {Campbell}}, \bibinfo {author} {\bibfnamefont {J.~M.}\ \bibnamefont {Dreiling}}, \bibinfo {author} {\bibfnamefont {C.}~\bibnamefont {Figgatt}}, \bibinfo {author} {\bibfnamefont {J.~P.}\ \bibnamefont {Gaebler}}, \bibinfo {author} {\bibfnamefont {J.}~\bibnamefont {Johansen}}, \bibinfo {author} {\bibfnamefont {M.}~\bibnamefont {Mills}}, \bibinfo {author} {\bibfnamefont {S.~A.}\ \bibnamefont {Moses}}, \bibinfo {author} {\bibfnamefont {J.~M.}\ \bibnamefont {Pino}}, \bibinfo {author} {\bibfnamefont {A.}~\bibnamefont {Ransford}}, \bibinfo {author} {\bibfnamefont {M.}~\bibnamefont {Rowe}}, \bibinfo {author} {\bibfnamefont {P.}~\bibnamefont {Siegfried}}, \bibinfo {author} {\bibfnamefont {R.~P.}\ \bibnamefont {Stutz}}, \bibinfo {author}
  {\bibfnamefont {M.}~\bibnamefont {{Foss-Feig}}}, \bibinfo {author} {\bibfnamefont {A.}~\bibnamefont {Vishwanath}},\ and\ \bibinfo {author} {\bibfnamefont {H.}~\bibnamefont {Dreyer}},\ }\bibfield  {title} {\bibinfo {title} {Non-{{Abelian}} topological order and anyons on a trapped-ion processor},\ }\href {https://doi.org/10.1038/s41586-023-06934-4} {\bibfield  {journal} {\bibinfo  {journal} {Nature}\ }\textbf {\bibinfo {volume} {626}},\ \bibinfo {pages} {505} (\bibinfo {year} {2024})}\BibitemShut {NoStop}%
\bibitem [{\citenamefont {Sang}\ and\ \citenamefont {al.}(2025)}]{sang_mixedstate_2025}%
  \BibitemOpen
  \bibfield  {author} {\bibinfo {author} {\bibfnamefont {S.}~\bibnamefont {Sang}}\ and\ \bibinfo {author} {\bibfnamefont {e.}~\bibnamefont {al.}},\ }\href@noop {} {\bibinfo {title} {Mixed-state topological degeneracy}} (\bibinfo {year} {2025})\BibitemShut {NoStop}%
\bibitem [{\citenamefont {Bravyi}\ \emph {et~al.}(2022)\citenamefont {Bravyi}, \citenamefont {Kim}, \citenamefont {Kliesch},\ and\ \citenamefont {Koenig}}]{bravyi_adaptive_2022}%
  \BibitemOpen
  \bibfield  {author} {\bibinfo {author} {\bibfnamefont {S.}~\bibnamefont {Bravyi}}, \bibinfo {author} {\bibfnamefont {I.}~\bibnamefont {Kim}}, \bibinfo {author} {\bibfnamefont {A.}~\bibnamefont {Kliesch}},\ and\ \bibinfo {author} {\bibfnamefont {R.}~\bibnamefont {Koenig}},\ }\href {https://doi.org/10.48550/arXiv.2205.01933} {\bibinfo {title} {Adaptive constant-depth circuits for manipulating non-abelian anyons}} (\bibinfo {year} {2022}),\ \Eprint {https://arxiv.org/abs/2205.01933} {arXiv:2205.01933 [quant-ph]} \BibitemShut {NoStop}%
\bibitem [{\citenamefont {Haah}(2016)}]{haah_invariant_2016}%
  \BibitemOpen
  \bibfield  {author} {\bibinfo {author} {\bibfnamefont {J.}~\bibnamefont {Haah}},\ }\bibfield  {title} {\bibinfo {title} {An invariant of topologically ordered states under local unitary transformations},\ }\href {https://doi.org/10.1007/s00220-016-2594-y} {\bibfield  {journal} {\bibinfo  {journal} {Communications in Mathematical Physics}\ }\textbf {\bibinfo {volume} {342}},\ \bibinfo {pages} {771} (\bibinfo {year} {2016})},\ \Eprint {https://arxiv.org/abs/1407.2926} {arxiv:1407.2926 [cond-mat, physics:math-ph, physics:quant-ph]} \BibitemShut {NoStop}%
\bibitem [{\citenamefont {Levin}\ and\ \citenamefont {Wen}(2003)}]{levin_fermions_2003}%
  \BibitemOpen
  \bibfield  {author} {\bibinfo {author} {\bibfnamefont {M.}~\bibnamefont {Levin}}\ and\ \bibinfo {author} {\bibfnamefont {X.-G.}\ \bibnamefont {Wen}},\ }\bibfield  {title} {\bibinfo {title} {Fermions, strings, and gauge fields in lattice spin models},\ }\href {https://doi.org/10.1103/PhysRevB.67.245316} {\bibfield  {journal} {\bibinfo  {journal} {Physical Review B}\ }\textbf {\bibinfo {volume} {67}},\ \bibinfo {pages} {245316} (\bibinfo {year} {2003})}\BibitemShut {NoStop}%
\bibitem [{\citenamefont {Lieb}\ and\ \citenamefont {Ruskai}(1973{\natexlab{a}})}]{lieb_fundamental_1973}%
  \BibitemOpen
  \bibfield  {author} {\bibinfo {author} {\bibfnamefont {E.~H.}\ \bibnamefont {Lieb}}\ and\ \bibinfo {author} {\bibfnamefont {M.~B.}\ \bibnamefont {Ruskai}},\ }\bibfield  {title} {\bibinfo {title} {A {{Fundamental Property}} of {{Quantum-Mechanical Entropy}}},\ }\href {https://doi.org/10.1103/PhysRevLett.30.434} {\bibfield  {journal} {\bibinfo  {journal} {Physical Review Letters}\ }\textbf {\bibinfo {volume} {30}},\ \bibinfo {pages} {434} (\bibinfo {year} {1973}{\natexlab{a}})}\BibitemShut {NoStop}%
\bibitem [{\citenamefont {Lieb}\ and\ \citenamefont {Ruskai}(1973{\natexlab{b}})}]{lieb_proof_1973}%
  \BibitemOpen
  \bibfield  {author} {\bibinfo {author} {\bibfnamefont {E.~H.}\ \bibnamefont {Lieb}}\ and\ \bibinfo {author} {\bibfnamefont {M.~B.}\ \bibnamefont {Ruskai}},\ }\bibfield  {title} {\bibinfo {title} {Proof of the strong subadditivity of quantum-mechanical entropy},\ }\href {https://doi.org/10.1063/1.1666274} {\bibfield  {journal} {\bibinfo  {journal} {Journal of Mathematical Physics}\ }\textbf {\bibinfo {volume} {14}},\ \bibinfo {pages} {1938} (\bibinfo {year} {1973}{\natexlab{b}})}\BibitemShut {NoStop}%
\bibitem [{\citenamefont {Nielsen}\ and\ \citenamefont {Chuang}(2010)}]{nielsen_quantum_2010}%
  \BibitemOpen
  \bibfield  {author} {\bibinfo {author} {\bibfnamefont {M.~A.}\ \bibnamefont {Nielsen}}\ and\ \bibinfo {author} {\bibfnamefont {I.~L.}\ \bibnamefont {Chuang}},\ }\href {https://doi.org/10.1017/CBO9780511976667} {\emph {\bibinfo {title} {Quantum {{Computation}} and {{Quantum Information}}: 10th {{Anniversary Edition}}}}}\ (\bibinfo  {publisher} {Cambridge University Press},\ \bibinfo {address} {Cambridge},\ \bibinfo {year} {2010})\BibitemShut {NoStop}%
\bibitem [{\citenamefont {Szalay}(2015)}]{szalay_multipartite_2015}%
  \BibitemOpen
  \bibfield  {author} {\bibinfo {author} {\bibfnamefont {S.}~\bibnamefont {Szalay}},\ }\bibfield  {title} {\bibinfo {title} {Multipartite entanglement measures},\ }\href {https://doi.org/10.1103/PhysRevA.92.042329} {\bibfield  {journal} {\bibinfo  {journal} {Physical Review A}\ }\textbf {\bibinfo {volume} {92}},\ \bibinfo {pages} {042329} (\bibinfo {year} {2015})}\BibitemShut {NoStop}%
\bibitem [{\citenamefont {Landau}\ and\ \citenamefont {Streater}(1993)}]{landau_birkhoffs_1993}%
  \BibitemOpen
  \bibfield  {author} {\bibinfo {author} {\bibfnamefont {L.~J.}\ \bibnamefont {Landau}}\ and\ \bibinfo {author} {\bibfnamefont {R.~F.}\ \bibnamefont {Streater}},\ }\bibfield  {title} {\bibinfo {title} {On {{Birkhoff}}'s theorem for doubly stochastic completely positive maps of matrix algebras},\ }\href {https://doi.org/10.1016/0024-3795(93)90274-R} {\bibfield  {journal} {\bibinfo  {journal} {Linear Algebra and its Applications}\ }\textbf {\bibinfo {volume} {193}},\ \bibinfo {pages} {107} (\bibinfo {year} {1993})}\BibitemShut {NoStop}%
\bibitem [{\citenamefont {Watrous}(2018)}]{watrous_theory_2018}%
  \BibitemOpen
  \bibfield  {author} {\bibinfo {author} {\bibfnamefont {J.}~\bibnamefont {Watrous}},\ }\href {https://doi.org/10.1017/9781316848142} {\emph {\bibinfo {title} {The {{Theory}} of {{Quantum Information}}}}}\ (\bibinfo  {publisher} {Cambridge University Press},\ \bibinfo {address} {Cambridge},\ \bibinfo {year} {2018})\BibitemShut {NoStop}%
\bibitem [{\citenamefont {Raussendorf}\ \emph {et~al.}(2005)\citenamefont {Raussendorf}, \citenamefont {Bravyi},\ and\ \citenamefont {Harrington}}]{rbh_toric_2005}%
  \BibitemOpen
  \bibfield  {author} {\bibinfo {author} {\bibfnamefont {R.}~\bibnamefont {Raussendorf}}, \bibinfo {author} {\bibfnamefont {S.}~\bibnamefont {Bravyi}},\ and\ \bibinfo {author} {\bibfnamefont {J.}~\bibnamefont {Harrington}},\ }\bibfield  {title} {\bibinfo {title} {Long-range quantum entanglement in noisy cluster states},\ }\href {https://doi.org/10.1103/PhysRevA.71.062313} {\bibfield  {journal} {\bibinfo  {journal} {Phys. Rev. A}\ }\textbf {\bibinfo {volume} {71}},\ \bibinfo {pages} {062313} (\bibinfo {year} {2005})}\BibitemShut {NoStop}%
\bibitem [{\citenamefont {Sutter}\ \emph {et~al.}(2016)\citenamefont {Sutter}, \citenamefont {Fawzi},\ and\ \citenamefont {Renner}}]{sutter_universal_2016}%
  \BibitemOpen
  \bibfield  {author} {\bibinfo {author} {\bibfnamefont {D.}~\bibnamefont {Sutter}}, \bibinfo {author} {\bibfnamefont {O.}~\bibnamefont {Fawzi}},\ and\ \bibinfo {author} {\bibfnamefont {R.}~\bibnamefont {Renner}},\ }\bibfield  {title} {\bibinfo {title} {Universal recovery map for approximate {{Markov}} chains},\ }\href {https://doi.org/10.1098/rspa.2015.0623} {\bibfield  {journal} {\bibinfo  {journal} {Proceedings of the Royal Society A: Mathematical, Physical and Engineering Sciences}\ }\textbf {\bibinfo {volume} {472}},\ \bibinfo {pages} {20150623} (\bibinfo {year} {2016})},\ \Eprint {https://arxiv.org/abs/1504.07251} {arxiv:1504.07251 [math-ph, physics:quant-ph]} \BibitemShut {NoStop}%
\bibitem [{\citenamefont {Kitaev}(2003)}]{kitaev_faulttolerant_2003}%
  \BibitemOpen
  \bibfield  {author} {\bibinfo {author} {\bibfnamefont {A.}~\bibnamefont {Kitaev}},\ }\bibfield  {title} {\bibinfo {title} {Fault-tolerant quantum computation by anyons},\ }\href {https://doi.org/https://doi.org/10.1016/S0003-4916(02)00018-0} {\bibfield  {journal} {\bibinfo  {journal} {Annals of Physics}\ }\textbf {\bibinfo {volume} {303}},\ \bibinfo {pages} {2} (\bibinfo {year} {2003})}\BibitemShut {NoStop}%
\bibitem [{\citenamefont {Li}\ and\ \citenamefont {Mong}(2024)}]{liReplicaTopologicalOrder2024}%
  \BibitemOpen
  \bibfield  {author} {\bibinfo {author} {\bibfnamefont {Z.}~\bibnamefont {Li}}\ and\ \bibinfo {author} {\bibfnamefont {R.~S.~K.}\ \bibnamefont {Mong}},\ }\href {https://doi.org/10.48550/arXiv.2402.09516} {\bibinfo {title} {Replica topological order in quantum mixed states and quantum error correction}} (\bibinfo {year} {2024}),\ \Eprint {https://arxiv.org/abs/2402.09516} {2402.09516 [cond-mat, physics:quant-ph]} \BibitemShut {NoStop}%
\bibitem [{\citenamefont {Simon}(2023)}]{simon_topological_2023}%
  \BibitemOpen
  \bibfield  {author} {\bibinfo {author} {\bibfnamefont {S.~H.}\ \bibnamefont {Simon}},\ }\href@noop {} {\emph {\bibinfo {title} {Topological {{Quantum}}}}}\ (\bibinfo  {publisher} {Oxford University Press},\ \bibinfo {year} {2023})\BibitemShut {NoStop}%
\bibitem [{\citenamefont {Lu}\ \emph {et~al.}(2020)\citenamefont {Lu}, \citenamefont {Hsieh},\ and\ \citenamefont {Grover}}]{Lu_finite_T_TO_2020}%
  \BibitemOpen
  \bibfield  {author} {\bibinfo {author} {\bibfnamefont {T.-C.}\ \bibnamefont {Lu}}, \bibinfo {author} {\bibfnamefont {T.~H.}\ \bibnamefont {Hsieh}},\ and\ \bibinfo {author} {\bibfnamefont {T.}~\bibnamefont {Grover}},\ }\bibfield  {title} {\bibinfo {title} {Detecting topological order at finite temperature using entanglement negativity},\ }\href {https://doi.org/10.1103/PhysRevLett.125.116801} {\bibfield  {journal} {\bibinfo  {journal} {Phys. Rev. Lett.}\ }\textbf {\bibinfo {volume} {125}},\ \bibinfo {pages} {116801} (\bibinfo {year} {2020})}\BibitemShut {NoStop}%
\bibitem [{\citenamefont {Zhang}\ \emph {et~al.}(2024{\natexlab{b}})\citenamefont {Zhang}, \citenamefont {Li},\ and\ \citenamefont {Lu}}]{zhang_longrange_2024}%
  \BibitemOpen
  \bibfield  {author} {\bibinfo {author} {\bibfnamefont {Z.}~\bibnamefont {Zhang}}, \bibinfo {author} {\bibfnamefont {Y.}~\bibnamefont {Li}},\ and\ \bibinfo {author} {\bibfnamefont {T.-C.}\ \bibnamefont {Lu}},\ }\href {https://arxiv.org/abs/2411.05004} {\bibinfo {title} {Long-range entanglement from spontaneous non-onsite symmetry breaking}} (\bibinfo {year} {2024}{\natexlab{b}}),\ \Eprint {https://arxiv.org/abs/2411.05004} {arXiv:2411.05004 [cond-mat.str-el]} \BibitemShut {NoStop}%
\bibitem [{\citenamefont {Shao}(2024)}]{shao_whats_2024}%
  \BibitemOpen
  \bibfield  {author} {\bibinfo {author} {\bibfnamefont {S.-H.}\ \bibnamefont {Shao}},\ }\href {https://arxiv.org/abs/2308.00747} {\bibinfo {title} {What's done cannot be undone: Tasi lectures on non-invertible symmetries}} (\bibinfo {year} {2024}),\ \Eprint {https://arxiv.org/abs/2308.00747} {arXiv:2308.00747 [hep-th]} \BibitemShut {NoStop}%
\bibitem [{\citenamefont {Piroli}\ and\ \citenamefont {Cirac}(2020)}]{piroli_quantum_2020}%
  \BibitemOpen
  \bibfield  {author} {\bibinfo {author} {\bibfnamefont {L.}~\bibnamefont {Piroli}}\ and\ \bibinfo {author} {\bibfnamefont {J.~I.}\ \bibnamefont {Cirac}},\ }\bibfield  {title} {\bibinfo {title} {Quantum {{Cellular Automata}}, {{Tensor Networks}}, and {{Area Laws}}},\ }\href {https://doi.org/10.1103/PhysRevLett.125.190402} {\bibfield  {journal} {\bibinfo  {journal} {Physical Review Letters}\ }\textbf {\bibinfo {volume} {125}},\ \bibinfo {pages} {190402} (\bibinfo {year} {2020})}\BibitemShut {NoStop}%
\bibitem [{\citenamefont {Hayden}\ \emph {et~al.}(2004)\citenamefont {Hayden}, \citenamefont {Jozsa}, \citenamefont {Petz},\ and\ \citenamefont {Winter}}]{hayden_structure_2004}%
  \BibitemOpen
  \bibfield  {author} {\bibinfo {author} {\bibfnamefont {P.}~\bibnamefont {Hayden}}, \bibinfo {author} {\bibfnamefont {R.}~\bibnamefont {Jozsa}}, \bibinfo {author} {\bibfnamefont {D.}~\bibnamefont {Petz}},\ and\ \bibinfo {author} {\bibfnamefont {A.}~\bibnamefont {Winter}},\ }\bibfield  {title} {\bibinfo {title} {Structure of {{States Which Satisfy Strong Subadditivity}} of {{Quantum Entropy}} with {{Equality}}},\ }\href {https://doi.org/10.1007/s00220-004-1049-z} {\bibfield  {journal} {\bibinfo  {journal} {Communications in Mathematical Physics}\ }\textbf {\bibinfo {volume} {246}},\ \bibinfo {pages} {359} (\bibinfo {year} {2004})}\BibitemShut {NoStop}%
\bibitem [{\citenamefont {Zhu}\ \emph {et~al.}(2024)\citenamefont {Zhu}, \citenamefont {Zhang}, \citenamefont {An},\ and\ \citenamefont {Zeng}}]{zhu_unified_2024}%
  \BibitemOpen
  \bibfield  {author} {\bibinfo {author} {\bibfnamefont {X.}~\bibnamefont {Zhu}}, \bibinfo {author} {\bibfnamefont {C.}~\bibnamefont {Zhang}}, \bibinfo {author} {\bibfnamefont {Z.}~\bibnamefont {An}},\ and\ \bibinfo {author} {\bibfnamefont {B.}~\bibnamefont {Zeng}},\ }\href {https://doi.org/10.48550/arXiv.2406.19683} {\bibinfo {title} {Unified {{Framework}} for {{Calculating Convex Roof Resource Measures}}}} (\bibinfo {year} {2024}),\ \Eprint {https://arxiv.org/abs/2406.19683} {arXiv:2406.19683} \BibitemShut {NoStop}%
\bibitem [{\citenamefont {Yang}\ \emph {et~al.}(2009)\citenamefont {Yang}, \citenamefont {Horodecki}, \citenamefont {Horodecki}, \citenamefont {Horodecki}, \citenamefont {Oppenheim},\ and\ \citenamefont {Song}}]{yang_squashed_2009}%
  \BibitemOpen
  \bibfield  {author} {\bibinfo {author} {\bibfnamefont {D.}~\bibnamefont {Yang}}, \bibinfo {author} {\bibfnamefont {K.}~\bibnamefont {Horodecki}}, \bibinfo {author} {\bibfnamefont {M.}~\bibnamefont {Horodecki}}, \bibinfo {author} {\bibfnamefont {P.}~\bibnamefont {Horodecki}}, \bibinfo {author} {\bibfnamefont {J.}~\bibnamefont {Oppenheim}},\ and\ \bibinfo {author} {\bibfnamefont {W.}~\bibnamefont {Song}},\ }\bibfield  {title} {\bibinfo {title} {Squashed {{Entanglement}} for {{Multipartite States}} and {{Entanglement Measures Based}} on the {{Mixed Convex Roof}}},\ }\href {https://doi.org/10.1109/TIT.2009.2021373} {\bibfield  {journal} {\bibinfo  {journal} {IEEE Transactions on Information Theory}\ }\textbf {\bibinfo {volume} {55}},\ \bibinfo {pages} {3375} (\bibinfo {year} {2009})}\BibitemShut {NoStop}%
\bibitem [{\citenamefont {Tucci}(2002)}]{tucci_entanglement_2002}%
  \BibitemOpen
  \bibfield  {author} {\bibinfo {author} {\bibfnamefont {R.~R.}\ \bibnamefont {Tucci}},\ }\href {https://doi.org/10.48550/arXiv.quant-ph/0202144} {\bibinfo {title} {Entanglement of {{Distillation}} and {{Conditional Mutual Information}}}} (\bibinfo {year} {2002}),\ \Eprint {https://arxiv.org/abs/quant-ph/0202144} {arXiv:quant-ph/0202144} \BibitemShut {NoStop}%
\bibitem [{\citenamefont {Nagel}\ and\ \citenamefont {Raggio}(2003)}]{nagel_another_2003}%
  \BibitemOpen
  \bibfield  {author} {\bibinfo {author} {\bibfnamefont {O.~A.}\ \bibnamefont {Nagel}}\ and\ \bibinfo {author} {\bibfnamefont {G.~A.}\ \bibnamefont {Raggio}},\ }\href {https://doi.org/10.48550/arXiv.quant-ph/0306024} {\bibinfo {title} {Another state entanglement measure}} (\bibinfo {year} {2003}),\ \Eprint {https://arxiv.org/abs/quant-ph/0306024} {arXiv:quant-ph/0306024} \BibitemShut {NoStop}%
\bibitem [{\citenamefont {Vidal}\ and\ \citenamefont {Werner}(2002)}]{vidal_computable_2002}%
  \BibitemOpen
  \bibfield  {author} {\bibinfo {author} {\bibfnamefont {G.}~\bibnamefont {Vidal}}\ and\ \bibinfo {author} {\bibfnamefont {R.~F.}\ \bibnamefont {Werner}},\ }\bibfield  {title} {\bibinfo {title} {Computable measure of entanglement},\ }\href {https://doi.org/10.1103/PhysRevA.65.032314} {\bibfield  {journal} {\bibinfo  {journal} {Phys. Rev. A}\ }\textbf {\bibinfo {volume} {65}},\ \bibinfo {pages} {032314} (\bibinfo {year} {2002})}\BibitemShut {NoStop}%
\bibitem [{\citenamefont {Kuno}\ \emph {et~al.}(2025)\citenamefont {Kuno}, \citenamefont {Orito},\ and\ \citenamefont {Ichinose}}]{kuno_intrinsic_2025}%
  \BibitemOpen
  \bibfield  {author} {\bibinfo {author} {\bibfnamefont {Y.}~\bibnamefont {Kuno}}, \bibinfo {author} {\bibfnamefont {T.}~\bibnamefont {Orito}},\ and\ \bibinfo {author} {\bibfnamefont {I.}~\bibnamefont {Ichinose}},\ }\bibfield  {title} {\bibinfo {title} {Intrinsic mixed-state topological order in a stabilizer system under stochastic decoherence: {{Strong-to-weak}} spontaneous symmetry breaking from a percolation point of view},\ }\href {https://doi.org/10.1103/PhysRevB.111.064111} {\bibfield  {journal} {\bibinfo  {journal} {Physical Review B}\ }\textbf {\bibinfo {volume} {111}},\ \bibinfo {pages} {064111} (\bibinfo {year} {2025})}\BibitemShut {NoStop}%
\bibitem [{\citenamefont {Sang}\ \emph {et~al.}(2021)\citenamefont {Sang}, \citenamefont {Li}, \citenamefont {Zhou}, \citenamefont {Chen}, \citenamefont {Hsieh},\ and\ \citenamefont {Fisher}}]{hsieh_nega_2021}%
  \BibitemOpen
  \bibfield  {author} {\bibinfo {author} {\bibfnamefont {S.}~\bibnamefont {Sang}}, \bibinfo {author} {\bibfnamefont {Y.}~\bibnamefont {Li}}, \bibinfo {author} {\bibfnamefont {T.}~\bibnamefont {Zhou}}, \bibinfo {author} {\bibfnamefont {X.}~\bibnamefont {Chen}}, \bibinfo {author} {\bibfnamefont {T.~H.}\ \bibnamefont {Hsieh}},\ and\ \bibinfo {author} {\bibfnamefont {M.~P.}\ \bibnamefont {Fisher}},\ }\bibfield  {title} {\bibinfo {title} {Entanglement negativity at measurement-induced criticality},\ }\href {https://doi.org/10.1103/PRXQuantum.2.030313} {\bibfield  {journal} {\bibinfo  {journal} {PRX Quantum}\ }\textbf {\bibinfo {volume} {2}},\ \bibinfo {pages} {030313} (\bibinfo {year} {2021})}\BibitemShut {NoStop}%
\bibitem [{\citenamefont {Shi}\ \emph {et~al.}(2020)\citenamefont {Shi}, \citenamefont {Dai},\ and\ \citenamefont {Lu}}]{shi2020_nega}%
  \BibitemOpen
  \bibfield  {author} {\bibinfo {author} {\bibfnamefont {B.}~\bibnamefont {Shi}}, \bibinfo {author} {\bibfnamefont {X.}~\bibnamefont {Dai}},\ and\ \bibinfo {author} {\bibfnamefont {Y.-M.}\ \bibnamefont {Lu}},\ }\bibfield  {title} {\bibinfo {title} {Entanglement negativity at the critical point of measurement-driven transition},\ }\href@noop {} {\bibfield  {journal} {\bibinfo  {journal} {arXiv preprint arXiv:2012.00040}\ } (\bibinfo {year} {2020})}\BibitemShut {NoStop}%
\bibitem [{\citenamefont {Lu}\ and\ \citenamefont {Vijay}(2023)}]{Lu_Vijay_2023_nega}%
  \BibitemOpen
  \bibfield  {author} {\bibinfo {author} {\bibfnamefont {T.-C.}\ \bibnamefont {Lu}}\ and\ \bibinfo {author} {\bibfnamefont {S.}~\bibnamefont {Vijay}},\ }\bibfield  {title} {\bibinfo {title} {Characterizing long-range entanglement in a mixed state through an emergent order on the entangling surface},\ }\href {https://doi.org/10.1103/PhysRevResearch.5.033031} {\bibfield  {journal} {\bibinfo  {journal} {Phys. Rev. Res.}\ }\textbf {\bibinfo {volume} {5}},\ \bibinfo {pages} {033031} (\bibinfo {year} {2023})}\BibitemShut {NoStop}%
\bibitem [{\citenamefont {Kitaev}(2006)}]{kitaevAnyonsExactlySolved2006}%
  \BibitemOpen
  \bibfield  {author} {\bibinfo {author} {\bibfnamefont {A.}~\bibnamefont {Kitaev}},\ }\bibfield  {title} {\bibinfo {title} {Anyons in an exactly solved model and beyond},\ }\href {https://doi.org/10.1016/j.aop.2005.10.005} {\bibfield  {journal} {\bibinfo  {journal} {Annals of Physics}\ }\bibinfo {series} {January {{Special Issue}}},\ \textbf {\bibinfo {volume} {321}},\ \bibinfo {pages} {2} (\bibinfo {year} {2006})}\BibitemShut {NoStop}%
\bibitem [{\citenamefont {Lu}\ and\ \citenamefont {Grover}(2019)}]{Lu_2019_singularity_nega}%
  \BibitemOpen
  \bibfield  {author} {\bibinfo {author} {\bibfnamefont {T.-C.}\ \bibnamefont {Lu}}\ and\ \bibinfo {author} {\bibfnamefont {T.}~\bibnamefont {Grover}},\ }\bibfield  {title} {\bibinfo {title} {Singularity in entanglement negativity across finite-temperature phase transitions},\ }\href {https://doi.org/10.1103/PhysRevB.99.075157} {\bibfield  {journal} {\bibinfo  {journal} {Phys. Rev. B}\ }\textbf {\bibinfo {volume} {99}},\ \bibinfo {pages} {075157} (\bibinfo {year} {2019})}\BibitemShut {NoStop}%
\end{thebibliography}%

\appendix

\section{Further properties of SLCs}\label{appendix:further_props_SLC}
\subsection{Pure state phases}\label{appendix:pure_state_phases}
In this appendix, we derive the condition for two pure states to be in the same long-range entanglement phase of matter (Def.\ref{def:phase_of_matter}). The derivation follows and generalizes the one presented in the Appendix A.1 of~\cite{sang2024mixed}.

We start by making the following two observations about quantum channels and their actions:
\begin{itemize}
    \item \textbf{Observation 1:} Let $\E=\sum_i p_i \E_i$ be a convex sum of several quantum channels. If $\E[\myketbra{\psi_1}]=\myketbra{\psi_2}$ holds for a pair of pure states $\ket{\psi_1}$ and $\ket{\psi_2}$, then it also holds that $\E_i[\myketbra{\psi_1}]=\myketbra{\psi_2}$ for any $i$.
    \item \textbf{Observation 2:} Let $\E[\cdot]=\Tr[U(\myketbra{0}_a\otimes(\cdot))]$ be a quantum channel in its Stinespring dilated form. If $\E[\myketbra{\psi_1}]=\myketbra{\psi_2}$ holds for a pair of pure states $\ket{\psi_1}$ and $\ket{\psi_2}$, then it holds that $U(\ket{\psi_1}\otimes \ket{0}_a) = \ket{\psi_2}\otimes\ket{\phi}_a$ for some pure state $\ket{\phi}$.
\end{itemize}

The first observation follows from that pure states are extremal points in the convex set of density matrices. The second observation holds because if a bipartite pure state ($U(\ket{\psi_1}\otimes \ket{0}_a)$ in this context) has a pure reduced state on a party, then the bipartite state must be a product state. 

Suppose a pair of pure states $\ket{\psi_1}$ and $\ket{\psi_2}$ are in the same long-range entanglement phase. By definition, this implies the existence of a pair of SLCs $\E_1=\sum_i p_i \E_{1, i}$ and $\E_1=\sum_j q_j \E_{2, j}$ such that:
\begin{equation}
\begin{aligned}
    &\myketbra{\psi_2} = \E_1[\myketbra{\psi_1}]\\
    &\myketbra{\psi_1} = \E_2[\myketbra{\psi_2}]
\end{aligned}
\end{equation}

By applying the first observation above, we conclude that
\begin{equation}
\begin{aligned}
    &\myketbra{\psi_2} = \E_{1, i}[\myketbra{\psi_1}]\\
    &\myketbra{\psi_1} = \E_{2, j}[\myketbra{\psi_2}]\ \ \forall i,j
\end{aligned}
\end{equation}

We then focus on only one particular choice of the pair $(i,j)$, and let $U_1$ and $U_2$ be LUs that dilate $\E_{1, i}$ and $\E_{1, j}$. By using the second observation above, we obtain
\begin{equation}
\begin{aligned}
    &U_1 (\ket{\psi_1}\otimes \ket{\bf 0}_a)= \ket{\psi_2}\otimes \ket{\phi_1}_a\\
    &U_2 (\ket{\psi_2}\otimes \ket{\bf 0}_a)= \ket{\psi_1}\otimes \ket{\phi_2}_a
\end{aligned}
\end{equation}
for some states $\ket{\phi_1}$ and $\ket{\phi_2}$. 

Recalling that $U_1$ and $U_2$ are LUs, we conclude that relations above imply that $\ket{\psi_1}$ ($\ket{\psi_2}$) is in the same pure-state phase as $\ket{\psi_2}\otimes\ket{\phi_2}$ ($\ket{\psi_1}\otimes\ket{\phi_1}$).  Due to transitivity of phase equivalence, we further conclude $\ket{\psi_1}$ is in the same phase as $\ket{\psi_1}\otimes\ket{\phi_1}\otimes\ket{\phi_2}$. 

To proceed, we need to assume that there is no ``catalyst" effect in phase equivalence relation: If $\ket{a}$ and $\ket{b}$ cannot be LU connected to each other, then neither does the pair $\ket{a}\otimes\ket{x}$ and $\ket{b}\otimes\ket{x}$, for any state $\ket{x}$. This appears to hold for any known pure state phases of matter. Given the assumption, $\ket{\phi_1}\otimes\ket{\phi_2}$ must be in the trivial phase of matter. Thus $\ket{\phi_1}$ belongs to an invertible phase, \textit{e.g.} Chern insulator, and $\ket{\phi_2}$ is $\ket{\phi_1}$'s inverse.  

To summarize and conclude, two pure states $\ket{\psi_1}$ and $\ket{\psi_2}$ are in the same entanglement phase following the Def.\ref{def:phase_of_matter} if and only if $\ket{\psi_1}$ and $\ket{\psi_2}\otimes\ket{\phi}$ are in the same pure state phase (via local unitaries), for a state $\ket{\phi}$ in an invertible phase.

\subsection{\texorpdfstring{Deriving the SLC decomposition Eq.~\ref{eq:SLC_decomposition}}{Deriving the SLC decomposition}}
\label{appendix:SLC_decomp_derivation}
Suppose $\E=\sum_{i=1}^n p_i \E_i$ is an SLC acting on a system $Q$, with each $\E_i$ being an LC. By definition, each $\E_i$ admits a decomposition:
\begin{equation}
    \E_i[\sigma] = \Tr_a(U_i(\sigma\otimes\ketbra{\bf 0}{\bf 0}_a)U_i)
\end{equation}
with each $U_i$ being a local unitary circuit acting jointly on the ancillary system $a$ and the original system $Q$. We make the technical assumption that all the $U_i$s have the same circuit structure, e.g. the brick-wall circuit structure. For any $\{U_i\}$, this can always be achieved by reorganizing gates within each $U_i$. We draw an 1D $U_i$ below for the sake of illustration: 
\begin{equation*}
    \includegraphics[width=7cm]{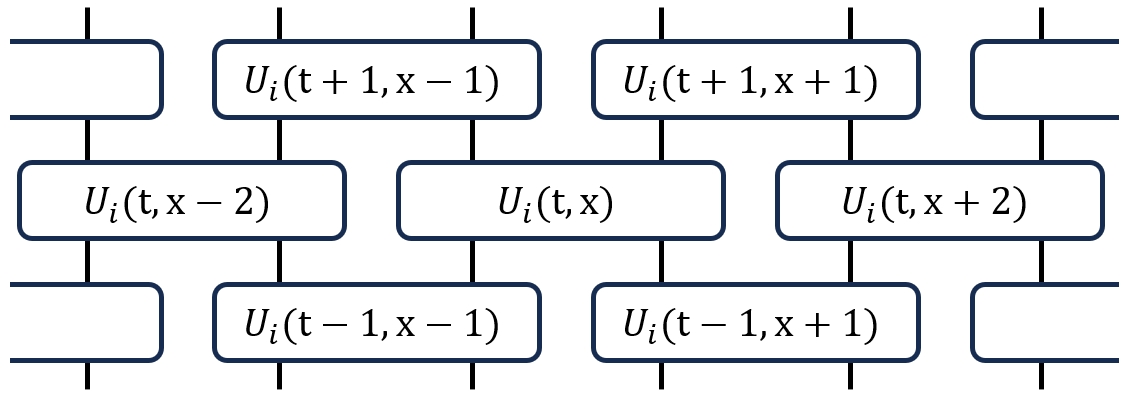}
\end{equation*}
where each $U_i(t,x)$ is a local unitary gate in the circuit $U_i$, with $t$ and $x$ labeling layer and location of the gate. 

Now we are ready to construct the decomposition Eq.\ref{eq:SLC_decomposition}. First, we introduce one $n$-level ancillary qudit to each lattice site $x$ (referred to as $b_x$, and $b=b_1 b_2 ... b_L$), which are initialized in:
\begin{equation}
    \rho_b = \sum_i p_i \ketbra{iii \ldots}{iii \ldots}_b
\end{equation}
Intuitively, $b$ is a `flag' that indicates the circuit which classical seed $i$ is sampled. Next we construct an LU $\tilde U$ that acts on $Q\cup a\cup b$. The $\tilde U$ has the same circuit structure as each $U_i$ (see the illustrating figure above), but with each gate at $(t,x)$ replaced by:
\begin{equation}
    U_i(t,x)\longrightarrow \tilde U (t,x)\equiv\sum_i \ketbra{i}{i}_{b_x}\otimes U_i(t,x)
\end{equation}
It is straightforward to verify that:
\begin{equation}
    \E[\sigma]=\Tr_{ab}[\tilde U(\sigma\otimes\ketbra{0}{0}_a\otimes \rho_b)\tilde U^\dagger]
\end{equation}
which is what we aim to prove (Eq.\ref{eq:SLC_decomposition}).

\subsection{Locality properties}
\label{appendix:locality_SLC}

To discuss the locality properties of SLCs, we employ the framework developed in \cite{piroli_quantum_2020}. There, the authors classify channels based on their locality features.  

\begin{definition}[Causality-preserving (CP) channels]\label{def:CP}
    A channel $\E$ is \textit{causality-preserving} if for any region $A$ and $X_A \in \mathfrak{A}_A$ supported there, $\E^\dagger(X_A) \in \mathfrak{A}_{\bar{A}_r}$, where $\bar{A}_r$ is the expansion of $A$ by a radius $r$ independent of $A$. The smallest of such $r \in \N$ for all $A$ is called the range of $\E$ and denoted by $\range(\E)$.
\end{definition}

\begin{definition}[Locality-preserving (LP) channels]\label{def:LP}
    A channel $\E$ is \textit{locality-preserving} if for all regions $A$ and $B$ such that $\bar{A} \sqcup \bar{B}$ partitions the entire space, where $\bar{A} = A \sqcup a$ and $\bar{B} = B \sqcup b$ are extensions of $A$ and $B$ by a finite radius then
    \begin{equation}
        \Tr_{a,b}[\E(\rho_{\bar{A}} \otimes \rho_{\bar{B}})] = \sigma_A \otimes \sigma_B
    \end{equation}
    where $\sigma_A = \Tr_{a, \bar{B}}\left[\E(\rho_{\bar{A}} \otimes \tau_{\bar{B}}) \right]$ and $\sigma_B = \Tr_{\bar{A}, b}\left[\E(\tau_{\bar{A}} \otimes \rho_{\bar{B}})\right]$ only depend on $\rho_{\bar{A}}$ and $\rho_{\bar{B}}$ respectively, as $\tau_{\bar{B}}$ and $\tau_{\bar{A}}$ are arbitrary states.
\end{definition}
An equivalent characterization of LP maps in the Heisenberg picture was proven in \cite{piroli_quantum_2020}:
\begin{theorem}\label{thm:LP_alternate_def}
    A channel $\E$ is locality-preserving iff for all local regions $A$ and $B$, if the support of $\E^\dagger(X_A)$ and $\E^\dagger(Y_B)$ are disjoint, then $\E^\dagger(X_A Y_B) = \E^\dagger(X_A) \E^\dagger(Y_B)$.
\end{theorem}

Local channels as defined in Section \ref{sec:LRE_vs_LRC} are both CP and LP. Every SLC $\E_{\rm SLC} = \sum_i p_i \E_i$ is still CP because each LC $\E_i$ in the mixture is CP with uniformly bounded range. However, there are SLCs that create long-range classical correlations, and thus are not LP (see Example \ref{example:SLC_not_LP}). Nevertheless, SLCs don't create long-range entanglement, so they preserve the separability of bipartite states, as encoded in the following definition:
\begin{definition}[Separability-preserving (SP) channels]\label{def:SP}
    Under the same notation of Definition \ref{def:LP}, a channel $\E$ is \textit{separability-preserving} if for all states $\rho_{\bar{A}}$ and $\rho_{\bar{B}}$ there exist families of states $\sigma^{(i)}_A$ and $\sigma^{(i)}_B$ such that
    \begin{align}
        \Tr_{a,b}[\E(\rho_{\bar{A}} \rho_{\bar{B}})] & = \sum_i p_i \sigma_A^{(i)} \otimes \sigma_B^{(i)} \text{ and} \\
        \sum_i p_i \sigma_A^{(i)} & = \Tr_{a, \bar{B}}\left[\E(\rho_{\bar{A}} \otimes \tau_{\bar{B}}) \right] \\
        \sum_i p_i \sigma_B^{(i)} & = \Tr_{\bar{A}, b}\left[\E(\tau_{\bar{A}} \otimes \rho_{\bar{B}}) \right],
    \end{align}
    where $\tau_{\bar{A}}$ and $\tau_{\bar{B}}$ are arbitrary states. 
\end{definition}
Note that the starting state $\rho_{\bar{A}}\rho_{\bar{B}}$ can be replaced by a separable state with classical correlations, \textit{mutatis mutandis}.

Not only SLCs are separability-preserving, but they also preserve area-law entanglement when acting on separable states:
\begin{theorem}
    For $\E$ an SLC and $\rho_{AB}$ a separable state between regions $A$ and $B = A^c$, then $\E(\rho_{AB})$ satisfies an \textit{area-law of entanglement}:
    \begin{equation}
        E^F_{A:B}(\E(\rho_{AB})) = O(|\partial A|),
    \end{equation}
    where $E^F_{A:B}$ is the entanglement of formation between $A$ and $B$.
\end{theorem}
\begin{proof}
    From the convexity of the entanglement of formation, we may assume $\rho_{AB}$ is uncorrelated, $\rho_{AB} = \rho_{A} \rho_{B}$, without loss of generality. Furthermore,
    \begin{align}
        E^F_{A:B}(\E(\rho_{AB})) \leq \sum p_i E^F_{A:B}(\E_i(\rho_{AB})).
    \end{align}
    By analyzing each LC term above, we use the fact that LCs preserve area-law of entanglement:
    \begin{align}
        E^F_{A:B}(\E_i(\rho_{AB})) & = E^F_{A:B}(\Tr_{A'B'}[U_i \rho_{A} \rho_{B} \otimes \ketbra{\bf 0}{\bf 0}_{A'B'} U_i^\dagger]) \\
        & \leq E^F_{AA':BB'}(U_i \tilde{\rho}_{A A'} \otimes \tilde{\rho}_{B B'} U_i^\dagger) \\
        & = O(|\partial A|),
    \end{align}
    where $\tilde{\rho}_{A A'} = \rho_{A} \otimes \ketbra{\bf 0}{\bf 0}_{A'}$ and similarly for $\tilde{\rho}_{B B'}$
\end{proof}

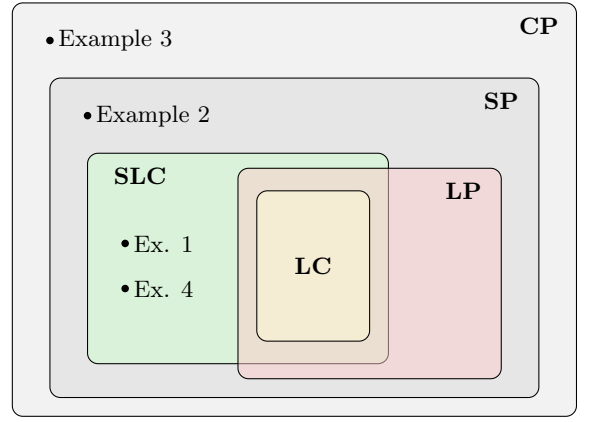
\begin{figure}[ht]
    \centering
    \begin{tikzpicture}
    \filldraw[fill=gray!10, rounded corners] (0,-0.5) rectangle (7.5,5);
    \node at (7, 4.7) {\textbf{CP}};
    \fill (0.5, 4.5) circle (0.05) node[right] {Example \ref{example:CP_not_SP}};

    \filldraw[fill=gray!20, rounded corners] (0.5,-0.25) rectangle (7,4);
    \node at (6.5, 3.7) {\textbf{SP}};
    \fill (1, 3.5) circle (0.05) node[right] {Example \ref{example:SP_not_SLC_LP}};

    \filldraw[fill=green!20, rounded corners, fill opacity=0.5] (1,0.2) rectangle (5,3);
    \node at (1.7, 2.7) {\textbf{SLC}};
    \fill (1.5, 1.8) circle (0.05) node[right] {Ex. \ref{example:SLC_not_LP}};
    \fill (1.5, 1.2) circle (0.05) node[right] {Ex. \ref{example:SLC_not_LP_piroli}};
    \filldraw[fill=red!20, rounded corners, fill opacity=0.5] (3,0) rectangle (6.5,2.8);
    \node at (6, 2.5) {\textbf{LP}};
    
    \filldraw[fill=yellow!20, rounded corners, fill opacity=0.5] (3.25,0.5) rectangle (4.75,2.5);
    \node at (4, 1.5) {\textbf{LC}};
\end{tikzpicture}
    \caption{Euler diagram representing the membership relations between the classes of channels discussed here. For the definitions of classes \textbf{LC} and \textbf{SLC}, see Defs. \ref{def:LC} and \ref{def:SLC}. For \textbf{CP}, \textbf{LP} and \textbf{SP}, see Defs. \ref{def:CP}, \ref{def:LP} and \ref{def:SP}. }
    \label{fig:channel_classes}
\end{figure}

Examples of channels that are causality-preserving but not one of the previous classes can be constructed by employing replace-by-$\sigma$ channels $\mathcal{R}_\sigma(\rho) \defeq \sigma \Tr[\rho]$. $\mathcal{R}_\sigma$ is trivially causality-preserving for all states $\sigma$ because any information of the initial state $\rho$ is discarded. Indeed, $\mathcal{R}_\sigma^\dagger(A) = \one \Tr[\sigma A]$ has empty support for all operators $A$. By choosing a resourceful state $\sigma$ for each class of channels, we get the following examples:
\begin{example}[SLC but not LP]\label{example:SLC_not_LP}
    Take $\sigma_{C}$ to be the equal mixture of $\ket{00\cdots 0}$ and $\ket{11 \cdots 1}$. Since it has long-range classical correlations, then $\mathcal{R}_{\sigma_C}$ is not LP. Nevertheless, it is an SLC since it is a convex combination of the replace-by-0's and replace-by-1's local channels.
\end{example}

\begin{example}[SP but not SLC nor LP]\label{example:SP_not_SLC_LP}
    Take $\sigma_{\rm GHZ}$ to be the GHZ state. Since it becomes fully separable after tracing out one site, $\mathcal{R}_{\sigma_{\rm GHZ}}$ is SP. However, it is not LP because far enough regions of $\sigma_{\rm GHZ}$ maintain classical correlations. Furthermore, it is not SLC because the GHZ state is not in the trivial phase. Otherwise, by pulling back the strong symmetries $\prod_i X_i$ and $Z_i Z_j$ of $\sigma_{\rm GHZ}$ to a product state, it would have long-range connected correlation function $\langle \E^*(Z_i) \E^*(Z_j) \rangle_c \to 1$, which is impossible for product states.
\end{example}

\begin{example}[CP but not SP]\label{example:CP_not_SP}
    Take $\sigma_{\rm Bells}$ to be the tensor product of Bell pairs on antipodal sites on a 1d periodic chain. That is, for even system size $L$, $\sigma_{\rm Bells} = \bigotimes_{i=1}^{L/2} \ket{\Psi_{\rm Bell}}_{i, i+L/2}$. Since $\sigma_B$ has long-range entanglement, $\mathcal{R}_{\sigma_{\rm Bells}}$ is not separability-preserving.
\end{example}

Other examples can be found without resorting to replacement channels, such as
\begin{example}[SLC but not LP \cite{piroli_quantum_2020}]\label{example:SLC_not_LP_piroli}
    The action of flipping all spins with probability half, or doing nothing otherwise, describes an SLC channel $\E(\rho) = \frac{1}{2}\rho + \frac{1}{2}\prod_i X_i \rho \prod_i X_i$ that is not LP, since $\E^\dagger(Z_iZ_j) \neq \E^\dagger(Z_i) \E^\dagger(Z_j) = 0$
\end{example}

The relations between the classes of channels discussed here, along with the examples above, are summarized in Fig. \ref{fig:channel_classes}.

\section{Further properties of the mcoCMI}
\label{appendix:further_props_mcoCMI}

Here we state and prove properties of the mcoCMI (Def. \ref{def:mcoCMI}), and discuss its relation to the co(QCMI).

\begin{enumerate}
    \item \textit{Positivity}: $I^{\sqcup}_{A:C|B}(\rho) \geq 0$. \label{app-item:positivity}
    \item \textit{Convexity}: $I^{\sqcup}_{A:C|B}(p \rho + (1-p) \sigma) \leq p I^{\sqcup}_{A:C|B}(\rho) + (1-p) I^{\sqcup}_{A:C|B}(\sigma)$. \label{app-item:convex}
    \item \textit{Reduction to CMI for pure states}: $I^{\sqcup}_{A:C|B}(\ketbra{\psi}{\psi}) = I_{A:C|B}(\ketbra{\psi}{\psi})$. \label{app-item:CMI_pure-states}
    \item \textit{Upper bound by CMI}: $I^{\sqcup}_{A:C|B}(\rho) \leq I_{A:C|B}(\rho)$. \label{app-item:CMI_upper_bound}
    \item \textit{Monotonicity under strictly local mixed-unitary channels}: $I^{\sqcup}_{A:C|B}(\E_{\text{m}U}(\rho)) \leq I^{\sqcup}_{A:C|B}(\rho)$ for $\E_{\text{m}U} = \sum_i p_i U_i (\cdot) U_i^\dagger$, where all $U_i = U_A^{(i)} \otimes U_B^{(i)} \otimes U_C^{(i)} \otimes U_E^{(i)}$ act strictly locally. \label{app-item:monotonic_mixed_unitary}
    \item \textit{Monotonicity under strictly local channels acting on $A$, $C$ and $E$}: $I^{\sqcup}_{A:C|B}(\E_R(\rho)) \leq I^{\sqcup}_{A:C|B}(\rho)$ for $\E_R$ any channel acting on $R \in \{A, C, E\}$.\footnote{It is not clear to us if the mcoCMI is monotonic under channels acting on the conditioning region $B$.} \label{app-item:monotonic_channels}
    \item \textit{Separability under $A|C$ bipartition}: If $I^{\sqcup}_{A:C|B}(\rho) = 0$, then $\rho_{AC}$ is separable. (The converse is not true: take $\rho = \ketbra{GHZ}{GHZ}$) \label{app-item:AC-separability}
    \item \textit{Zero for fully separable states} If $\rho$ is fully separable on $ABCE$, then $I^{\sqcup}_{A:C|B}(\rho) = 0$. \label{app-item:ABCE-separability}
\end{enumerate}
\begin{proof}
\begin{enumerate}
    \item Follows directly from strong subadditivity, $I(A:C|B) \geq 0$.
    
    \item Follows directly from the mixed convex roof construction: If $\{ (p_i, \rho_i) \}_i$ and $\{ (q_j, \sigma_j) \}_j$ are optimal decompositions of $\rho$ and $\sigma$ respectively, then $\{ (p p_i, \rho_i) \}_i \sqcup \{ ((1-p) q_j, \sigma_j) \}_j$ is a decomposition of $p \rho + (1-p) \sigma$. Hence, 
    \begin{align}
        I^{\sqcup}_{A:C|B}&(p \rho + (1-p) \sigma) \leq \\ 
        & \sum_i p p_i I_{A:C|B}(\rho_i) + \sum_j (1-p) q_i I_{A:C|B}(\sigma_i) \\
        & = p I^{\sqcup}_{A:C|B}(\rho) + (1-p) I^{\sqcup}_{A:C|B}(\sigma).
    \end{align}

    \item Follows from pure states being, by definition, extremal in the convex set of states.
    
    \item Follows directly from taking the trivial decomposition of $\rho = 1 \cdot \rho$.

    \item First note that $I^{\sqcup}_{A:C|B}$ is invariant under strictly local unitaries $U = U_A \otimes U_B \otimes U_C \otimes U_E$, since there is a one-to-one correspondence between decompositions of $\rho$ and $U \rho U^\dagger$ for which the CMI remains invariant. This plus the convexity of the mcoCMI proves the monotonicity.

    \item From the Stinespring dilation theorem, any channel $\E_S$ acting on a quantum system $S$ is of the form
    \begin{equation}
        \E_S(\rho) = \Tr_{S'}[U_{SS'} (\rho \otimes \ketbra{0}{0}_{S'}) U^\dagger_{SS'}],
    \end{equation}
    for $U_{S S'}$ a unitary operator acting on the extended system $SS'$. Thus, to prove monotonicity of $I^{\sqcup}_{A:C|B}$, it suffices to separately consider the operations of 1) Pure state ancilla addition, 2) local unitaries, and 3) tracing out. For the first two types of operations, the monotonicity follows easily. For the tracing out operation, consider a state $\rho$ defined on $AA'BCE$, with optimal decomposition $\rho = \sum_i p_i \rho_i$ for $I^{\sqcup}_{AA':C|B}(\rho)$. Since $\{p_i, \Tr_{A'}[\rho_i] \}$ forms a decomposition of $\Tr_{A'}[\rho]$, then
    \begin{align}
        I^{\sqcup}_{A:C|B}(\Tr_{A'}[\rho]) & \leq \sum_i p_i I_{A:C|B}(\Tr_{A'}[\rho_i]) \\
        & \leq \sum_i p_i I_{AA':C|B}(\rho_i) \\
        & = I^{\sqcup}_{AA':C|B}(\rho).
    \end{align}
    The same calculation is valid if $A$ is replaced by $C$ or $E$.

    \item If $I^{\sqcup}_{A:C|B}(\rho)=0$, then there exist a optimal decomposition $\{ p_i, \rho^{(i)} \}$ satisfying $\forall i,\ I_{A:C|B}(\rho^{(i)}) = 0$. That implies $\rho^{(i)}_{AB}$ is separable \cite{hayden_structure_2004}, which in turn implies $\rho_{AB} = \sum_i p_i \rho_{AB}^{(i)}$ is separable.

    \item Follows directly from the fact that $I_{A:C|B}$ is zero for tensor product states $\rho_A \otimes \rho_B \otimes \rho_C \otimes \rho_E$.
\end{enumerate}
\end{proof}

The pure-state convex roof extension of the CMI (``co(QCMI)''), treated in \cite{wang_analog_2024}, satisfies all of the above, with the exception of properties \ref{app-item:CMI_upper_bound} and \ref{app-item:monotonic_channels}. A counterexample for both relies on the ``GHZ mixture'' states $\rho_p = p \ketbra{GHZ_+}{GHZ_+} + (1-p) \ketbra{GHZ_-}{GHZ_-}$, where $\ket{GHZ_\pm} = \frac{1}{\sqrt{2}} (\ket{000} \pm \ket{111})$. 
By using the ``trivialization'' method described in \cite{zhu_unified_2024} for the numerical estimation of convex roof extended quantities, we verified that the co(QCMI) of $\rho_p$ is strictly greater than its CMI, for $0<p<1$, violating property \ref{app-item:CMI_upper_bound}. We also found an optimal decomposition of $\rho_p$ for the co(QCMI) as follows: $\rho_p = \frac{1}{2} \ketbra{\psi_{\theta^*}}{\psi_{\theta^*}} + \frac{1}{2} \ketbra{\psi_{\pi - \theta^*}}{\psi_{\pi - \theta^*}}$, where $\ket{\psi_\theta} \defeq \cos(\theta/2) \ket{000} + \sin(\theta/2) \ket{111}$ and $\theta^* = \arcsin(1 - 2p)$. Furthermore, by taking a purification $\ket{\phi_p}$ of $\rho_p$ on $ABCE$, it follows that $\text{co(QCMI)}[\rho_p] > I_{A:C|B}(\rho_p) = \text{co(QCMI)}[\ketbra{\phi_p}{\phi_p}]$, thus violating property \ref{app-item:monotonic_channels} for the tracing-out channel $\E_E = \Tr_E$.

We also note that a multipartite generalization of the mixed convex roof extension of the mutual information, $I^{\sqcup}_{A:C}$, where the conditioning region $B$ is empty, was studied in \cite{yang_squashed_2009} from the perspective of it being a good entanglement measure. See also \cite{tucci_entanglement_2002, nagel_another_2003} for the bipartite case.

\section{Decoupling of Pauli-X and Z noise for CSS codes}
\label{appendix:noisy_TC_decoupling}

Here, we prove that entropic measures of the noisy toric code state $\rho_{p_X, p_Z}$ contain separate contributions from the purely $Z$-dephased state $\rho_{0, p_Z}$ and the purely $X$-dephased state $\rho_{p_X,0}$. Namely, we will prove that for any region $A$,
\begin{equation}
    \label{eq:entropy_decoupling}
    S_A(\rho_{p_X, p_Z}) + S_{A}(\rho_{0,0}) = S_A(\rho_{0, p_Z}) + S_A(\rho_{p_X, 0}).
\end{equation}
An analogous relation for the CMI $I_{A:C|B}$ follows directly from the equation above, thus establishing Eq. \eqref{eq:CMI_decoupling}.

What matters for the decoupling of Eq. \eqref{eq:entropy_decoupling} is the fact that the initial state has stabilizers that are only Pauli-$Z$ or $X$ strings, also known as a CSS code \cite{nielsen_quantum_2010}. Thus, we will extend our assumptions to any CSS code (mixed) state $\rho_0$ that is acted upon by a mix of $Z$ and $X$ dephasing:
\begin{equation}
    \label{eq:noisy_CSS_state}
    \rho_{p_X, p_Z} \defeq \left[\bigotimes_{e} \E_e(X, p_X) \circ \E_e(Z, p_Z) \right](\rho_0).
\end{equation}

Since the reduced density matrix of a CSS code is also a CSS code and the noise is on-site (so $\Tr_{A^c}$ commutes with the channel of Eq. \eqref{eq:noisy_CSS_state} above), then we can drop the subregion $A$ in Eq. \eqref{eq:entropy_decoupling} and consider just the global von Neumann entropy, without loss of generality.

Let $G = G_X \times G_Z$ be the $[n,k]$ stabilizer group of $\rho_0$. Then,
\begin{equation}
    \rho_0 = \frac{1}{2^k} \prod_{g_Z \in G_Z} \frac{\one + g_Z}{2} \prod_{g_X \in G_X} \frac{\one + g_X}{2}.
\end{equation}
Each Kraus operator of the noise can either do nothing or flip the sign of a subset of stabilizers of $G_Z$ or $G_X$, exclusive. Thus, we have
\begin{align}
    \begin{split} 
        \rho_{p_X, p_Z} & = \frac{1}{2^k} (\sum_{s_X} \mathbb{P}^X_{p_Z}(s_X) \prod_{g_X \in G_X} \frac{\one + s_X(g_X) \, g_X}{2}) \\
        & \quad \times (\sum_{s_Z} \mathbb{P}^Z_{p_X}(s_Z) \prod_{g_Z \in G_Z} \frac{\one + s_Z(g_Z) \, g_Z}{2}) 
    \end{split} \\
    & =\sum_{s_X, s_Z} \mathbb{P}^X_{p_Z}(s_X) \mathbb{P}^Z_{p_X}(s_Z) \cdot\rho_G(s_X, s_Z).
\end{align}
where $\rho_G(s) = \frac{1}{2^k} \prod_{g \in G} \frac{\one + s(g) g}{2}$ for $s \in \{\pm1\}^{G}$, and $\mathbb{P}^X_{p_Z}(s_X)$ is the probability distribution arising from all the possible ways the Pauli-$Z$ noise operators can flip the sign of the stabilizers of $G_X$ and result in $s_X \in \{\pm1\}^{G_X}$ (similarly with $\mathbb{P}^Z_{p_X}(s_Z)$).

With the expression above, and knowing that the states $\rho_G(s)$ are mutually orthogonal and have the same entropy as $\rho_0$, the von Neummann entropy of $\rho_{p_X, p_Z}$ can be easily calculated:
\begin{align}
    \begin{split}
        S(\rho_{p_X, p_Z}) & = H(\mathbb{P}^X_{p_Z} \times \mathbb{P}^Z_{p_X}) + \\ 
        & \quad \sum_{s_Z, s_X} \mathbb{P}^X_{p_Z}(s_X) \mathbb{P}^Z_{p_X}(s_Z) S(\rho_G(s_X, s_Z))
    \end{split} \\
    & = H(\mathbb{P}^X_{p_Z}) + H(\mathbb{P}^Z_{p_X}) + S(\rho_0).
\end{align}
where $H$ is the classical Shannon entropy. Finally, Eq. \eqref{eq:entropy_decoupling} follows directly from the expression above.

\section{Entanglement negativity of fermionic imTO states}
\label{appendix:negativity}

In the main text, we have shown that the ZX-decohered toric code is long-range entangled due to the strong fermion 1-form symmetry. Here we aim to directly diagnose such long-range entanglement via entanglement negativity, which is a computable measure of mixed-state entanglement defined as \cite{vidal_computable_2002}
\begin{equation}
    E_N(\rho) = \log \norm{\rho^{T_A}}_1,
\end{equation}
where $\rho^{T_A}$ is the partial transpose of $\rho$ in the $A$ region, and $\norm{A}_1 = \Tr \sqrt{A^\dagger A}$ is the trace norm.

While the entanglement negativity for the ZX-decohered toric code model has been computed in Ref.~\cite{wang_intrinsic_2025} using a replica calculation, we provide an alternative derivation based on stabilizer formalism without resorting to a replica trick\footnote{During the preparation of this manuscript, we became aware of a work~\cite{kuno_intrinsic_2025} employing the stabilizer formalism descibed in Sec. \ref{sec:negativity-stab_formalism} to calculate the entanglement negativity of the ZX-decohered toric code.}. Finally, we apply the same method to calculate the entanglement negativity of the maximally mixed state with the fermionic one-form symmetry of Kitaev's honeycomb model.

\subsection{Stabilizer formalism}
\label{sec:negativity-stab_formalism}

Given a stabilizer state $\rho \propto  \prod_{ i } \frac{1+S_i}{2}$ acting on a bipartite Hilbert space $\mathcal{H} = \mathcal{H}_{\mathcal{A}} \otimes  \mathcal{H}_{\mathcal{B}} $, with $\{S_i\}$ being the set of stabilizer generators, we define a commutation matrix $M_{ij} $, which encodes the algebra of stabilizers when restricted to the subregion $\mathcal{A}$: 

\begin{equation}\label{eq:nega_M}
	M_{ij} = \begin{cases}
		0  \quad  \text{for }  [S_i]_{\mathcal{A}}    [S_j]_{\mathcal{A}}   =   [S_j]_{\mathcal{A}}  [S_i]_{\mathcal{A}}  \\
		1  \quad  \text{for }  [S_i]_{\mathcal{A}}    [S_j]_{\mathcal{A}}   =    - [S_j]_{\mathcal{A}}  [S_i]_{\mathcal{A}} 
	\end{cases}
\end{equation}
where $[S_i]_{\mathcal{A}}$ is the operator obtained by only keeping the part of $S_i$ supported on the subregion $\mathcal{A}$. The entanglement negativity is determined by the rank of $M$ (as a matrix on $\Z_2$ field):
\begin{equation}\label{eq:nega}
E_N =  \frac{\text{Rank}(M)}{2} \log 2 
\end{equation}
When all the stabilizers $S_i$ are geometrically local, only the stabilizers around the bipartition boundary can anticommute in a subregion, implying that it suffices to only consider these boundary operators in the matrix $M$ for the negativity calculation. Eq.\ref{eq:nega} was first derived in Refs.~\cite{hsieh_nega_2021,shi2020_nega}, and we present an arguably simpler derivation in Appendix.\ref{appendix:derivation}, built on the method in Ref.\cite{Lu_Vijay_2023_nega}. 

\subsection{ZX decohered toric code}\label{sec:ZX}
Given the 2d toric code described by a vertex term $A_v= \prod_{e\in \partial^T  v  }Z_v$ (i.e. the product of four Pauli-Zs around a vertex $v$) and a plaquette term $B_p=\prod_{e\in \partial p}X_e$ (the product of four Pauli-Xs around a plaquette $p$), its ground subspace density matrix is a stabilizer mixed state $\propto  \prod_{v} \frac{1+A_v}{2} \prod_{p}  \frac{1+B_p}{2}$. Under the maximal decoherence by the ZX noise channel ($p=1/2$ in \eqref{eq:ZX_decoherence_channel}), the decohered mixed state is the ground-subspace density matrix of the following stabilizer Hamiltonian

\begin{equation}
 \includegraphics[width=7cm]{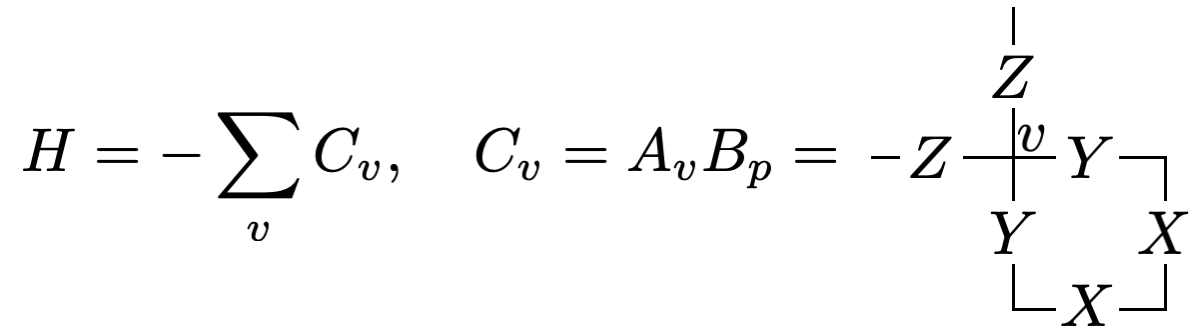}
\end{equation}
Namely, the local stabilizer $C_s$ is the product of two neighboring $A_v, B_p$ stabilizers. Below we utilize Eq.\ref{eq:nega} to calculate entanglement negativity by considering two types of bipartition boundaries. 

 \textbf{Boundary choice 1}:
We impose the periodic boundary condition along $\hat{x}$ direction and the open boundary condition along $\hat{y}$ direction, so that the 2d lattice has the topology of a cylinder. We divide the system into two subregions using a horizontal cut (the dashed line below) of size $L$:
\begin{equation}
\includegraphics[width=5cm]{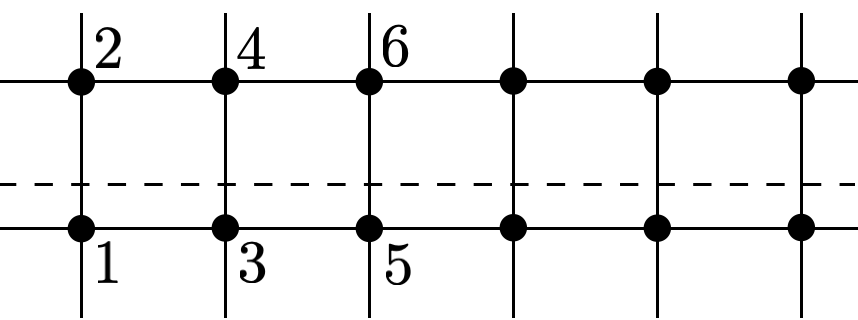} 
\end{equation}
There are $N=2L$ relevant boundary stabilizers labeled by $C_1, C_2 ,\cdots, C_{2L}$, where each $C_i$ stabilizer must anticommute with its two neighboring stabilizers ($C_{i-1}, C_{i+1}$) when restricted to a subregion. This indicates the following commutation matrix 
\begin{equation}\label{eq:commutation_ZX_toric}
M= \begin{pmatrix}
0 & 1 & 0 & 0 & ...& 1\\
1 & 0 & 1 &  0& ... & 0  \\
0 & 1 & 0 & 1 & ... & 0 \\
 \vdots & \vdots &  \vdots &\vdots &  \ddots &  \vdots
\end{pmatrix}	
 \end{equation}
In particular, the commutation matrix $M$ has a rank $2L-2$, indicating the entanglement negativity 
\begin{equation}
E_N = L \log 2  - \log 2,
\end{equation}
We also note that the algebra among those boundary $C_v$ stabilizers is the same as the alternating star and plaquette stabilizers in the un-decohered toric code, which therefore has the same commutation matrix and entanglement negativity. The subleading constant $\log 2$ can be identified as the topological entanglement negativity, which is expected to detect the long-range entanglement.

In fact, the entanglement negativity is $E_N  = L \log 2 - \log 2 $ for any noise rate $p$, which has been observed in Ref.\cite{wang_intrinsic_2025} using a replica calculation. Here we provide a simple derivation by properly choosing the stabilizers that describe the toric code.  To begin, consider the un-decohered toric code,  we multiply each plaquette stabilizer $B_p$ below the entanglement cut by its upper-left star stabilizer $A_s$. Hence, the set of stabilizers in the lower subregion can be equivalently given by the set  $(A_sB_p, B_p)$.  On the other hand, for the upper subregion,  we multiply each star stabilizer $A_s$ above the entanglement cut by its lower-right plaquette stabilizer $B_p$, and hence,  the set of stabilizers in the upper subregion can be equivalently given by the set  $(A_s,  A_sB_p)$.  This redefinition of the stabilizers has the advantage that every boundary stabilizer along the entanglement cut is the composite stabilizer $A_sB_p=C_s$.  The entanglement negativity is solely determined by these $2L$ stabilizers, and importantly,  these $2L$ stabilizers are immune to decoherence since they commute with the ZX noise operator. This simple understanding implies that entanglement negativity does not depend on $p$ for this choice of entanglement cut. 

 \textbf{Boundary choice 2}: 
Here we discuss the results of entanglement negativity when considering the following bipartition boundary: 
\begin{equation}\label{eq:nega_1}
\includegraphics[width=5cm]{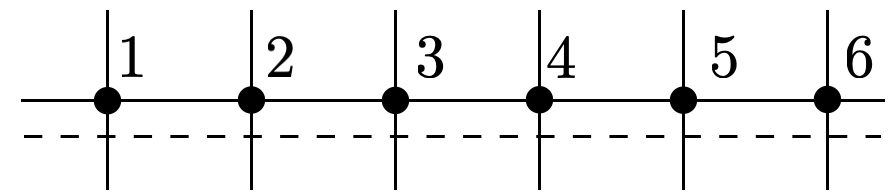} 
\end{equation}
There are only $N=L$ boundary stabilizers labeled by $C_1, C_2, \cdots , C_L$, where each $C_i$ stabilizer must anticommute with its two neighboring stabilizers ($C_{i-1}, C_{i+1}$) when restricted to a subregion. In this case, the rank of the  commutator matrix is $L-2$ for even $L$ and $L-1$ for odd $L$, implying

\begin{equation}
E_N  =  \begin{cases}
\frac{L}{2} \log 2 -\log 2 \quad \text{for even } L  \\
\frac{L}{2}  \log  2  -\frac{\log 2}{2}  \quad \text{for odd } L,  
\end{cases}
\end{equation}
Again, by identifying the subleading constant as topological negativity. One sees that it is always non-zero but the value depends on the parity of $L$.

\subsection{Honeycomb model}\label{sec:honeycomb}
Here we consider a mixed state that also features the fermion 1-form symmetry, and shows that it has the same structure of entanglement negativity as in the ZX-dephased toric code. Specifically, we consider a honeycomb lattice with each vertex accommodating a qubit. We define a mixed state $\rho \propto \prod_p \frac{1+B_p}{2}$, where $B_p$ is the plaquette operator in the Kitaev's honeycomb model \cite{kitaevAnyonsExactlySolved2006}, i.e. the product Pauli-X,-Y,-Z around a plaquette $p$: 

\begin{equation}
\includegraphics[width=3cm]{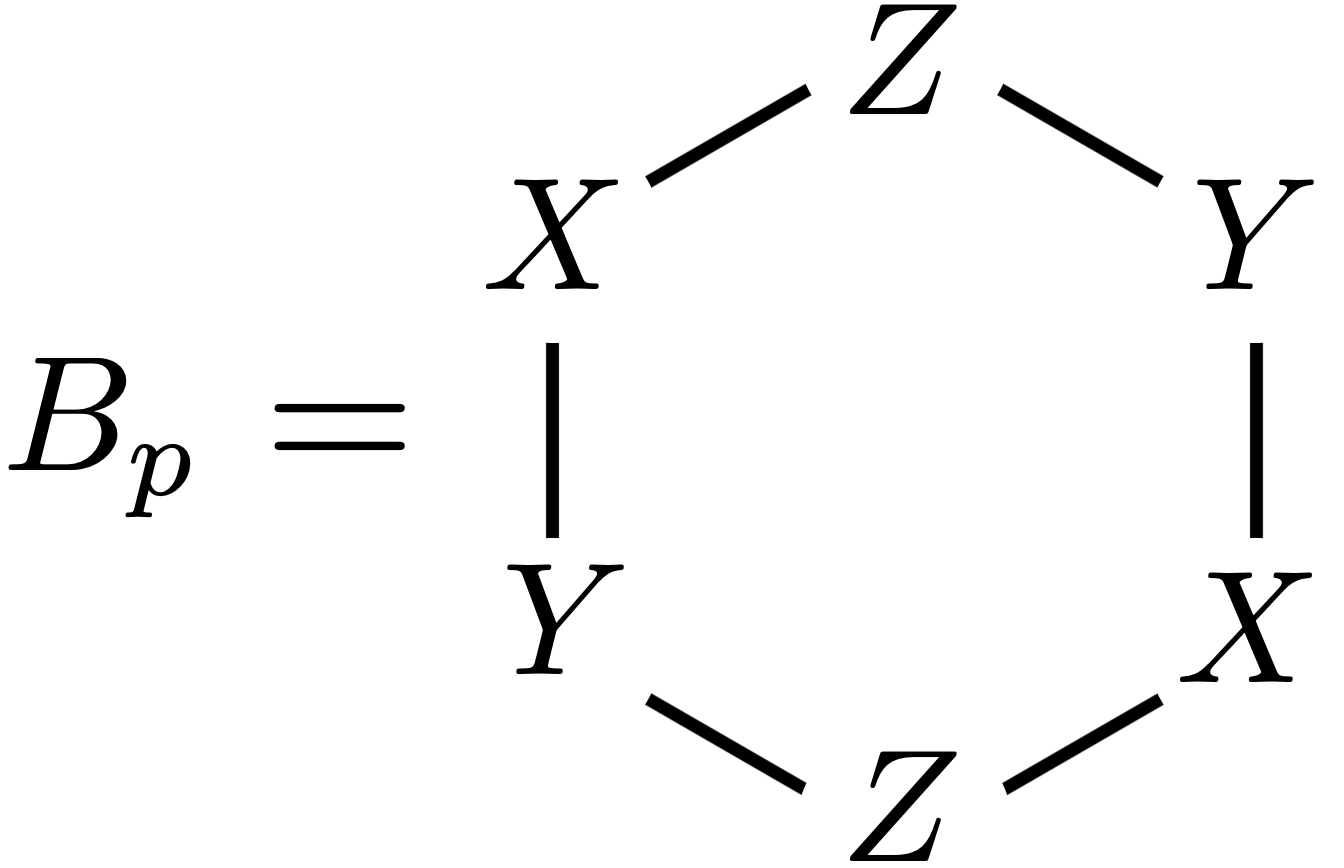} 
\end{equation}
These $B_p$ operators can be regarded as the fermion 1-form symmetry generators. Since the mixed state $\rho$ has these strong anomalous symmetries, it is long-range entangled based on the discussion in the main text. Below we present the result of entanglement negativity. In particular, the mixed state features non-zero topological entanglement negativity, whose values depend on the choice of the bipartition boundary. Specifically, we impose the periodic boundary condition along $\hat{x}$ direction and the open boundary condition along $y$ direction, so that the 2d lattice has the topology of a cylinder. 

 \textbf{Boundary choice 1}: The first choice of the bipartition boundary is depicted as follows:  

\begin{equation}
\includegraphics[width=5cm]{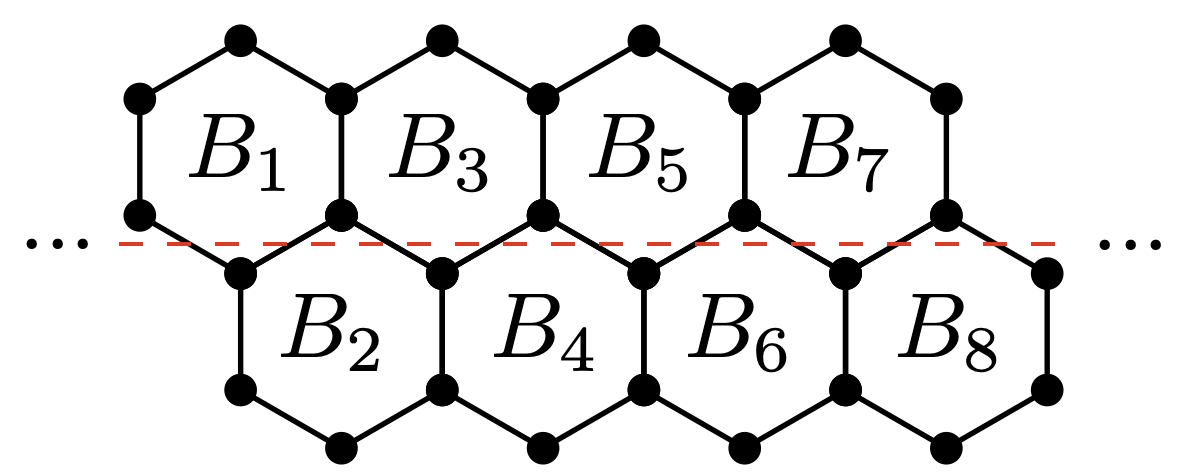} 
\end{equation}
where there are $2L$ boundary $B_p$ stabilizers along the bipartition boundary of size $L$. Their commutation relations in the subregion are exactly given by Eq.\ref{eq:commutation_ZX_toric}, indicating the entanglement negativity is $E_N = (L-1) \log 2$. The subleading constant $\log 2 $ may be identified as topological negativity.

 \textbf{Boundary choice 2}:
 The second choice of the bipartition boundary is depicted as follows:  
\begin{equation}
\includegraphics[width=5cm]{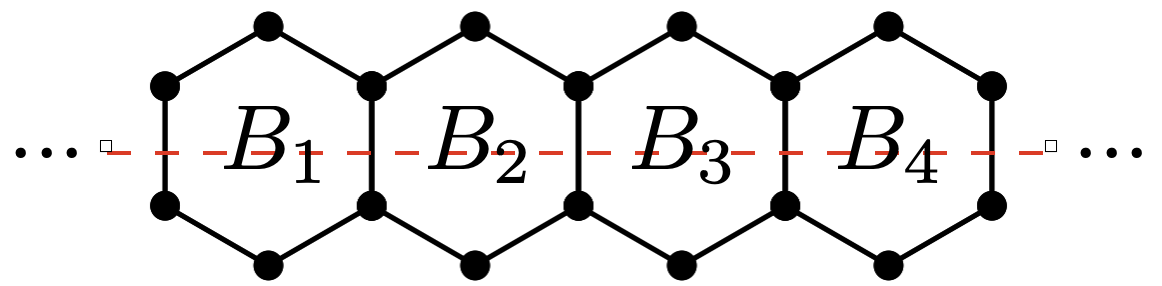} 
\end{equation}

There are only $L$ boundary stabilizers labeled by $B_1, B_2, \cdots , B_L$, where each $B_i$ stabilizer must anticommute with its two neighboring stabilizers ($B_{i-1}, B_{i+1}$) when restricted to a subregion. In this case, the rank of the  commutator matrix is $L-2$ for even $L$ and $L-1$ for odd $L$, implying

\begin{equation}
E_N  =  \begin{cases}
\frac{L}{2} \log 2 -\log 2 \quad \text{for even } L  \\
\frac{L}{2}  \log  2  -\frac{\log 2}{2}  \quad \text{for odd } L,  
\end{cases}
\end{equation}
By identifying the subleading constant as topological negativity. One sees that it is always non-zero but the value depends on the parity of $L$.

\subsection{\texorpdfstring{Derivation of Eq.\ref{eq:nega}}{Derivation of the entanglement negativity formula based on the stabilizer formalism}}\label{appendix:derivation} 
Here we present a simple derivation of Eq.\ref{eq:nega}, built on the formalism in Ref.\cite{Lu_Vijay_2023_nega} (see also Refs.\cite{Lu_2019_singularity_nega,Lu_finite_T_TO_2020}). Specifically, we will find that the negativity spectrum is completely given by the wave function of an emergent state (in a fictitious Hilbert space) defined by the commutation matrix $M$. We also note that Eq.\ref{eq:nega} was first derived in Ref.\cite{hsieh_nega_2021,shi2020_nega} with distinct approaches. 

To begin, given a stabilizer mixed state $\rho = \prod_i \frac{1+S_i}{2}$, we expand the projectors in the density matrix: $\rho \propto \prod_i (1+S_i) = \sum_{\sigma_i} \prod_i S_i^{\sigma_i} $ with $\sigma_i = 0 ,1$. Taking a partial transpose on the subregion $A$ gives the matrix 

\begin{equation}
	\rho^{T_A} \propto \sum_{\sigma_i} 
 \prod_i \tilde{S}_i^{\sigma_i} (-1)^{\sum_{i<j} \sigma_i M_{ij} \sigma_j   }, 
\end{equation}
where $\tilde{S}_i = S_i^{T_A}$,  i.e. the partial transpose of $S_i$. Then eigenvalues of $\rho^{T_A}$, i.e. negativity spectrum, can be obtained by choosing $\tilde{S}_i = 1,-1$. Below we will abuse the notation by using $\rho^{T_A}$ to denote the spectrum, as opposed to a matrix. 

To understand the structure of the negativity spectrum, it is useful to introduce a fictitious Hilbert space spanned by the Pauli-Z basis $\ket{ \{ \sigma_i \}  }$. By introducing a trivial product state $\ket{+} \propto \sum_{\sigma} \ket{\{ \sigma_i \}}$, and $\ket{\psi}$ that encodes the sign structure from taking partial transpose: $\ket{\psi}   \propto  \sum_{\sigma} (-1)^{\sum_{i<j} \sigma_i M_{ij} \sigma_j   } \ket{\{ \sigma_i \}}$, the negativity spectrum can be expressed as 

\begin{equation}
	\rho^{T_A}=\bra{+} \prod_i Z_i^{\frac{1-\tilde{S}_i}{2}}   \ket{\psi}, 
\end{equation}
where different eigenvalues are given by choosing different $\{ \tilde{S}_i =\pm 1\}$, which in turn determines the Pauli-Z insertion ($Z_i$ is a Pauli-Z matrix acting on $i$-th qubit) sandwiched between $\ket{+}$ and $\ket{\psi}$ in the fictitious Hilbert space. Note that the above expression for negativity spectrum is properly normalized, i.e. $\sum_{\{\tilde{S}_i\}}  \rho^{T_A} =\sum_{\{\tilde{S}_i\}}   \bra{+} \prod_i Z_i^{\frac{1-\tilde{S}_i}{2}}    \ket{\psi}  =1 $ since the partial transpose does not alter the trace of a matrix. This can be checked as follows

\begin{equation}
\begin{split}
\sum_{\{\tilde{S}_i\}}   \bra{+} \prod_i Z_i^{\frac{1-\tilde{S}_i}{2}}    \ket{\psi}&=   \bra{+} \prod_i (1+Z_i) \prod_{i<j } (\text{CZ}_{ij})^{M_{ij}} \ket{+} \\
		&  =   \bra{+} \prod_i (1+Z_i)  \ket{+} \\
		&=1, 
	\end{split}
\end{equation}
where we have used $\ket{\psi} = \prod_{i<j}   (\text{CZ}_{ij})^{M_{ij}} \ket{+}$, with $\text{CZ}$ being the two-qubit controlled-Z gate. 

Entanglement negativity can be obtained by summing over all absolute values of the negativity spectrum: $E_N = \log \norm{\rho^{T_A}}_1$, where 

\begin{equation}
	\norm{\rho^{T_A}}_1 = \sum_{\{\tilde{S}_i\} }  \abs{\bra{+} \prod_i Z_i^{\frac{1-\tilde{S}_i}{2}}    \ket{\psi}    }.
\end{equation} 
To proceed, we observe that it is the sum of the absolute value of the wave function of $\ket{\psi}$ in the X basis. Since $\ket{\psi}$ is a stabilizer state with the parent Hamiltonian $-\sum_{i}  X_i\prod_{j} Z_j^{M_{ij}}$, one has $\ket{\psi} \propto   \prod_i \frac{ 1+  X_i\prod_{j} Z_j^{M_{ij}} }{2}  \ket{+}  $. By expanding the product, one finds 

\begin{equation}
	\ket{\psi} = \frac{1}{\sqrt{\mathcal{N}}} \sum_{\{ \tilde{S}_i \} }  \psi(\{\tilde{S}_i\})   \ket{\{ \tilde{S}_i\}} 
\end{equation}
where $\mathcal{N}$ is a normalization constant, $ \psi(\{ \tilde{S}_i\}) = 1, -1 , 0$, and $\ket{\{\tilde{S}_i\}}$ denotes a Pauli-X basis state. $\mathcal{N}$ is the number of $\ket{\{\tilde{S}_i \}}$ with non-zero wave function value, which is given by $2^{\text{Rank}(M)  }$. Therefore, 

\begin{equation}
\begin{split}
\sum_{\{\tilde{S}_i\} }  \abs{     \bra{+} \prod_i Z_i^{\frac{1-\tilde{S}_i}{2}}    \ket{\psi}}  &= 2^{\text{Rank}(M)  }  2^{-\frac{1}{2} \text{Rank}(M)  } \\
&=2^{\frac{1}{2} \text{Rank}(M)  }.   
 \end{split}
\end{equation}
and one finds the entanglement negativity $E_N =  \frac{\text{Rank}(M)}{2}\log 2$.

\section{\texorpdfstring{Entanglement negativity of the maximally mixed state in Eq. \eqref{eq:MMS_Z2_cubed}}{Entanglement negativity of the maximally mixed state}}
\label{appendix:ent_neg_MMS}

Here, we compute the negativity $\mathcal{N}(\rho) \defeq (\norm{\rho^{T_A}}_1 - 1)/2$ of the maximally mixed state of Eq. \eqref{eq:MMS_Z2_cubed}. For simplicity of the argument, we consider the maximally mixed state (MMS) with only (strong) joint symmetry $X \defeq X_\circ X_\bullet$ instead of $X_\circ$ and $X_\bullet$ separately. This state, which we denote here by $\rho_\infty$, is the equal-weight mixture of the MMS with $X_\circ = X_\bullet = +1$ and the one with $X_\circ = X_\bullet = -1$. Hence, by the convexity of the negativity \cite{vidal_computable_2002}, it suffices to compute the negativity for $\rho_\infty$ to lower bound the negativity of the original MMS.

In terms of cat states of the $X$ symmetry, we can decompose $\rho_\infty$ into:
\begin{equation}
    \rho_\infty \propto \sum_{\alpha \in B_{CZ}} \ketbra{\alpha^{(+)}}{\alpha^{(+)}},
\end{equation}
where $\ket{\alpha^{(\pm)}} \propto (\one \pm X) \ket{\alpha} = \ket{\alpha} \pm \ket{\overline \alpha}$ is a cat state, with $\ket{\alpha}, \alpha \in \{0,1\}^N$, being a product state in the $Z$-basis, and $B_{CZ} = \{ \alpha \mid \forall \text{ loop }\ell, W_{CZ}(\ell)\ket{\alpha} = \ket{\alpha} \}$. Note that the sum above is overcounting terms, since bitstrings $\alpha$ differing by global bit flip give the same cat state $\ket{\alpha^{(+)}}$. Since the overcounting does not depend on $\alpha$, it will not alter the analysis below.

Taking the partial transpose of the above, we have
\begin{align}
    \rho_\infty^{T_A} \propto & \sum_{\alpha \in B_{CZ}} \ketbra{\alpha^{(+)}}{\alpha^{(+)}}^{T_A} \\
    \begin{split}\label{eq:MMS_PT_decomp}
    = & \sum_{\alpha \in B_{CZ}} \frac{1}{2} (\ketbra{\alpha_A \alpha_B}{\alpha_A \alpha_B} + \ketbra{\overline\alpha_A \alpha_B}{\alpha_A \overline\alpha_B} \\ 
    & \qquad +\ketbra{\alpha_A \overline\alpha_B}{\overline\alpha_A \alpha_B} + \ketbra{\overline\alpha_A \overline\alpha_B}{\overline\alpha_A \overline\alpha_B})
    \end{split}.
\end{align}
We can see from the above that the matrix $\rho_\infty^{T_A}$ is block diagonal with respect to the subspaces $V_{\alpha} \defeq \mathrm{span} \{\ket{\alpha_A \alpha_B}, \ket{\overline\alpha_A \alpha_B}, \ket{\alpha_A \overline\alpha_B}, \ket{\overline\alpha_A \overline\alpha_B}\}$. In particular, the negativity (absolute value of sum of negative eingenvalues) of $\rho_\infty$ is the probability-weighted sum of negativities of each of the block diagonal submatrices.

To find the block matrix of $\rho_\infty^{T_A}$ in the subspace $V_\alpha$, it is crucial to know if $\overline\alpha_A \alpha_B \in B_{CZ}$ or not. If so, then we have to sum the contribution of the term explicit in \eqref{eq:MMS_PT_decomp} with the one coming from $\overline\alpha_A \alpha_B \in B_{CZ}$ (and these two will be the only two distinct contributions for the submatrix in $V_\alpha$), resulting in
\begin{equation}
    \frac{1}{2}\begin{pmatrix}
        1 & 0 & 0 & 0 \\
        0 & 0 & 1 & 0 \\
        0 & 1 & 0 & 0 \\
        0 & 0 & 0 & 1
    \end{pmatrix}
    +
    \frac{1}{2}\begin{pmatrix}
        0 & 0 & 0 & 1 \\
        0 & 1 & 0 & 0 \\
        0 & 0 & 1 & 0 \\
        1 & 0 & 0 & 0
    \end{pmatrix}
    =
    \frac{1}{2}\begin{pmatrix}
        1 & 0 & 0 & 1 \\
        0 & 1 & 1 & 0 \\
        0 & 1 & 1 & 0 \\
        1 & 0 & 0 & 1
    \end{pmatrix},
\end{equation}
where we have written a 4-by-4 submatrix in the $(\ket{\alpha_A \alpha_B}, \ket{\overline\alpha_A \alpha_B}, \ket{\alpha_A \overline\alpha_B}, \ket{\overline\alpha_A \overline\alpha_B})$ basis. This submatrix has only non-negative eigenvalues and, thus, does not contribute to the negativity of $\rho_\infty$. If $\overline\alpha_A \alpha_B \notin B_{CZ}$ instead, then the submatrix is just
\begin{equation}
    \frac{1}{2}\begin{pmatrix}
        1 & 0 & 0 & 0 \\
        0 & 0 & 1 & 0 \\
        0 & 1 & 0 & 0 \\
        0 & 0 & 0 & 1
    \end{pmatrix},
\end{equation}
which has negativity equal to $1/2$. Hence, the negativity of $\rho_\infty$ is
\begin{equation}\label{eq:negativity_condition}
    \mathcal{N}(\rho_\infty) = \frac{1}{2}\mathrm{Prob}(\overline\alpha_A \alpha_B \notin B_{CZ} | \alpha_A \alpha_B \in B_{CZ}),
\end{equation}
where the conditional probability above is over all constrained bitstrings $\alpha = \alpha_A \alpha_B \in B_{CZ}$.

Note that, in the limit of very large regions $A$, the probability above converges to 1, since almost all bitstrings $\alpha \in B_{CZ}$ will satisfy $W_{CZ}(\hexagon) \ket{\overline\alpha_A \alpha_B}= -\ket{\overline\alpha_A \alpha_B}$ for at least one hexagonal loop $\hexagon$ that intersects both $A$ and $A^c$, which is sufficient for $\overline\alpha_A \alpha_B \notin B_{CZ}$. As such,
\begin{align}
    \lim_{|A|\to \infty} \mathcal{N}(\rho_\infty)= \frac{1}{2},
\end{align}
and we conclude that the MMS of Eq. \eqref{eq:MMS_Z2_cubed} is indeed entangled.

\end{document}